\documentclass[journal,onecolumn,10pt]{IEEEtran}
\usepackage{amsmath,amssymb,amsthm}
\usepackage{bm}
\usepackage{iftex}
\ifPDFTeX\usepackage{graphicx}\else\ifLuaTeX\usepackage{graphicx}\else\ifXeTeX\usepackage{graphicx}\else
  \usepackage[dvipdfmx]{graphicx}% (u)pLaTeX + dvipdfmx
\fi\fi\fi
\usepackage{enumitem}
\usepackage{cite}
\usepackage[font=footnotesize,labelsep=period]{caption}

\newtheoremstyle{hy}{6pt}{6pt}{\normalfont}{}{\bfseries}{.}{.5em}{}
\theoremstyle{hy}
\newtheorem{theorem}{Theorem}
\newtheorem{lemma}{Lemma}
\newtheorem{corollary}{Corollary}
\newtheorem{proposition}{Proposition}
\newtheorem{definitionx}{Definition}
\newtheorem{remarkx}{Remark}
\newtheorem{examplex}{Example}
\newenvironment{definition}{\begin{definitionx}}{\hfill$\Box$\end{definitionx}}
\newenvironment{remark}{\begin{remarkx}}{\hfill$\Box$\end{remarkx}}
\newenvironment{example}{\begin{examplex}}{\hfill$\Box$\end{examplex}}
\renewenvironment{proof}[1][Proof]{\par\noindent\textit{#1:}\ \ }{\hfill$\Box$\par\medskip}

\newcommand{\bX}{\bm{X}}
\newcommand{\bY}{\bm{Y}}
\newcommand{\bZ}{\bm{Z}}
\newcommand{\bW}{\bm{W}}
\newcommand{\bV}{\bm{V}}
\newcommand{\bx}{\bm{x}}
\newcommand{\by}{\bm{y}}
\newcommand{\bz}{\bm{z}}
\newcommand{\bu}{\bm{u}}
\newcommand{\bj}{\bm{j}}
\newcommand{\bJ}{\bm{J}}
\newcommand{\cM}{\mathcal{M}}
\newcommand{\bv}{\bm{v}}
\newcommand{\bR}{\bm{R}}
\newcommand{\bP}{\bm{P}}
\newcommand{\br}{\bm{r}}
\newcommand{\bone}{\mathbf{1}}
\newcommand{\R}{\mathbb{R}}
\newcommand{\E}{\mathbb{E}}
\newcommand{\V}{\mathbb{V}}
\newcommand{\Cl}{\mathrm{Cl}}
\newcommand{\cG}{\mathcal{G}}
\newcommand{\cB}{\mathcal{B}}
\newcommand{\cE}{\mathcal{E}}
\newcommand{\cX}{\mathcal{X}}
\newcommand{\cY}{\mathcal{Y}}
\newcommand{\cZ}{\mathcal{Z}}
\newcommand{\cS}{\mathcal{S}}
\newcommand{\cO}{\mathcal{O}}
\newcommand{\cP}{\mathcal{P}}
\newcommand{\cN}{\mathcal{N}}
\newcommand{\cK}{\mathcal{K}}
\newcommand{\cT}{\mathcal{T}}
\newcommand{\cA}{\mathcal{A}}
\newcommand{\sfC}{\mathsf{C}}
\IfFileExists{dsfont.sty}{\usepackage{dsfont}\newcommand{\ind}{\mathds{1}}}{\newcommand{\ind}{\mathbf{1}}}% indicator
\newcommand{\plimi}{\mathop{\mbox{p-}\liminf}_{n\to\infty}}
\newcommand{\plims}{\mathop{\mbox{p-}\limsup}_{n\to\infty}}
\newcommand{\essinf}{\mbox{$w$-ess.inf}\,}
\newcommand{\st}{\,\big|\,}
\newcommand{\bGamma}{\boldsymbol{\Gamma}}
\newcommand{\cC}{\mathcal{C}}
\newcommand{\cD}{\mathcal{D}}
\newcommand{\cL}{\mathcal{L}}

\begin{document}

\title{Information-Spectrum Methods for Multiple-Access\\ Channels---Part~I: $\varepsilon$-Capacity Regions\\ in the Mixed Gaussian Case}
\author{Te Sun Han, \textit{Life Fellow, IEEE} and Hideki Yagi, \textit{Member, IEEE}%
\thanks{T. S. Han is with The University of Electro-Communications, Chofu, Tokyo 182-8585, Japan (email: han@is.uec.ac.jp).}%
\thanks{H. Yagi is with The University of Electro-Communications, Chofu, Tokyo 182-8585, Japan (email: h.yagi@uec.ac.jp).}}
\date{}
\maketitle

\begin{abstract}
We study the $\varepsilon$-capacity regions of mixed multiple-access channels (MACs) with general mixture and cost constraints by the
information spectrum method. The starting point is a general formula for the $\varepsilon$-capacity
region of general MACs with cost constraint, together with the random-coding and converse lemmas from which it is derived. For the
quasi-static fading Gaussian MAC, with random fading coefficients and noise variance, we show that the
$\varepsilon$-capacity region is the $\varepsilon$-outage capacity region, which gives the latter a Shannon-theoretic operational meaning.
The direct part follows from the general formula, and the converse part from the strong converse property of the Gaussian MAC,
which is stronger than the $\varepsilon$-strong converse, applied to each component. The $\varepsilon$-capacity
regions coincide in all four scenarios of availability of channel state information (CSI), and the results extend to $K$ users. With
CSI at both terminals, variable-rate coding enlarges the $0$-capacity region to the polymatroid whose rank function is the average of
those of the components. For mixed MACs with general components and a countable parameter space, CSI at the receiver alone never
enlarges the $\varepsilon$-capacity region, and the region with CSI at the transmitters is given by an information spectrum formula;
unlike the Gaussian case, the regions with and without CSI at the transmitters may differ. For uncountable parameter spaces, only the
formula for the region with CSI at both terminals survives in general. Parts II and III treat finite-alphabet mixed MACs and the
strong converse for the $K$-user Gaussian MAC, respectively.
\end{abstract}

\begin{center}
\textbf{Index Terms}
\end{center}
\begin{quote}\small
Information spectrum, mixed channel, multiple-access channel, capacity region, quasi-static fading channel,
outage capacity, strong converse, channel state information, variable-rate coding
\end{quote}

%% ====================================================================
\section{Introduction}
\subsection{Background and Main Results}

A multiple-access channel (MAC) is one of the fundamental models in channel coding for multiuser communications,
consisting of multiple encoders and a single decoder. The main object of interest in the coding problem for a MAC is the set
of rate pairs for which there exist codes whose probability of decoding error vanishes asymptotically as the blocklength tends
to infinity. This set is referred to as the \emph{channel capacity region} or simply \emph{capacity region}. When a nonzero error probability
up to a fixed constant $\varepsilon$ is allowed, the set of achievable rate pairs is called the \emph{$\varepsilon$-capacity region}. An important class of MACs,
of particular interest to us, is the class of \emph{mixed MACs}, which are characterized as a weighted mixture of parameterized
MACs over a general parameter space (state space). As is well known, mixed MACs are in general nonergodic, so that,
although simply structured, they are far less tractable than ergodic channels. To overcome this difficulty, we resort to the
information spectrum method.

When the state of a fading channel stays constant over a whole codeword (quasi-static fading), the channel is \emph{nonergodic}:
coding over a long block does not average over the channel states, and, for many fading distributions, no positive rate can be supported
with vanishing error probability (cf.\ Theorem~\ref{thm:qsf0} and Remark~\ref{rem:zero}). Performance is measured instead by the outage
probability, the probability that the realized channel cannot support the transmitted rates. Since Ozarow et al.~\cite{OSW94}, the
$\varepsilon$-outage capacity has been the standard figure of merit for such channels, while its Shannon-theoretic justification long
remained open; Shamai and Wyner~\cite{SW97} pointed out that the information spectrum approach~\cite{VH94} may provide a resolution.

In this paper we follow that suggestion for the MAC: we model a quasi-static fading MAC as a mixed MAC and study its
$\varepsilon$-capacity region by the information spectrum method, which requires neither stationarity nor ergodicity.

Our main result is that, for the quasi-static fading Gaussian MAC with an arbitrary fading distribution and a random noise variance,
the $\varepsilon$-capacity region \emph{is} the $\varepsilon$-outage capacity region (Theorem~\ref{thm:qsf}). The outage region thus
acquires an operational meaning. Moreover, this region is the same whether channel state information (CSI) is available at neither
terminal, at the receiver, at the transmitters or at both (Theorem~\ref{thm:csi}), and the results extend to $K$ users
(Theorem~\ref{thm:kuser}). For a single user with fixed noise variance this was shown by Yang et al.~\cite{YDKP14} as a by-product of a
second-order analysis; our proof concerns the first-order problem only and is elementary. The converse part combines the strong
converse of the Gaussian MAC, applied to each component, with Fatou's lemma. The point is which strong converse: for a component in
outage, the $\varepsilon$-strong converse gives only that the limit superior of its error probability is one, to which Fatou's lemma
does not apply, whereas the strong converse property of Definition~\ref{def:sc} gives that the error probability itself tends to one
(Remark~\ref{rem:sc-role}). If CSI is available at both terminals and
the rates may depend on the realized state, the $0$-capacity region $\tilde C^{\rm rt}_{\bP}(0\,|\,\bW)$ ({\em nonergodic capacity region}) of this variable-rate coding becomes the polymatroid whose
rank function is the average of those of the components (Theorem~\ref{thm:vr}), which exactly coincides in form with the ergodic
capacity region of Tse and Hanly~\cite{TH98}.

Why does CSI not matter in the Gaussian case? Because a single input, Gaussian with full power, is asymptotically optimal for all
components simultaneously. To make this visible, we also treat mixed MACs with general components (Sec.~\ref{sec:general-comp}),
where no such common input need exist. For a countable parameter space, CSI at the receiver alone still never enlarges the
$\varepsilon$-capacity region, but CSI at the transmitters can help to achieve the $\varepsilon$-outage capacity region, an outer bound whenever the
components have the strong converse property (Theorem~\ref{thm:gencsi}, Example~\ref{ex:zchannel}; for finite alphabets, see the outage region of \cite[Eq.~(46)]{HY-II} and \cite[Proposition~3]{HY-II}). This contrast is the reason why the general theory is placed in the present
paper.

Theorem~\ref{thm:general} is a slight modification of \cite[Theorem~5]{Han98}, Theorem~\ref{thm:wyner} and Corollary~\ref{cor:zero} are known,
and Theorem~\ref{thm:sc} is due to Fong and Tan~\cite{FT16}. To our knowledge, Theorems~\ref{thm:qsf} and \ref{thm:csi}--\ref{thm:gencsi} and Proposition~\ref{prop:uncountable}
are new; their relation to earlier work is discussed where they appear
(Remarks~\ref{rem:outage}, \ref{rem:yang}, \ref{rem:polymatroid} and~\ref{rem:effros}).

This paper is Part~I of a series of three (Fig.~\ref{fig:series}). Part~II~\cite{HY-II} treats finite-alphabet mixed MACs from the same general formula, and
Part~III~\cite{HY-III} gives an information spectrum proof of the strong converse for the $K$-user Gaussian MAC; Part~III can be read
independently of Parts~I and~II.

\subsection{A Unified Perspective via Information Spectrum Methods}\label{sec:intro-unified}

The results above may look diverse---four CSI scenarios, two or $K$ users, fixed- and variable-rate coding, Gaussian and general
components---but their proofs follow a single pattern, and it is this pattern, rather than any individual result, that we wish to
convey. It is convenient to distinguish two levels. By the \emph{information spectrum theory} we mean general formulas, such as
Theorem~\ref{thm:general}, which identify, without any assumption on the channel, the quantity that answers a coding problem. By the
\emph{information spectrum calculus} we mean the evaluation of that quantity for concrete channels.
The two levels support each other, as summarized in Fig.~\ref{fig:series}: the theory tells the calculus what to compute, and
results of the calculus, such as the strong converse of each component, are in turn the input to general arguments such as the
identity~\eqref{eq:avg-error} below.
The three parts share the same two fundamental lemmas. Parts I and II start from the general formula derived from these lemmas, whereas Part III returns to the lemmas themselves to establish the strong converse. Thus, the three parts develop a common information-spectrum approach at different levels, as illustrated in Fig. 1.

\begin{figure*}[!t]
\centering
\includegraphics{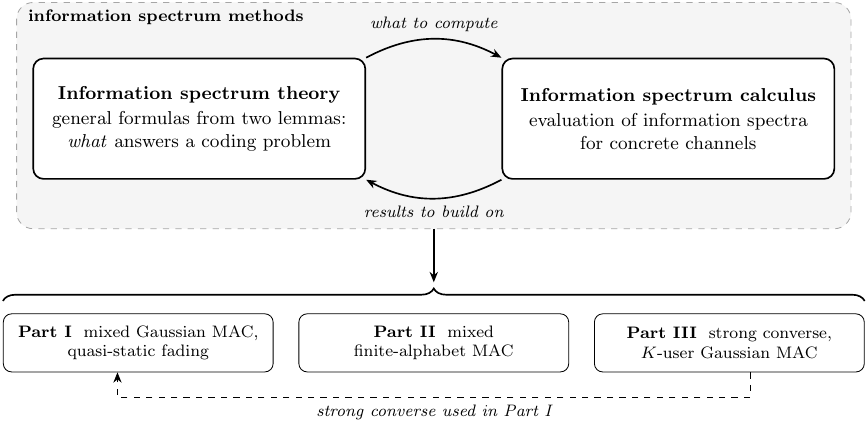}
\caption{Overview of the series. Information spectrum theory identifies, without any assumption on the channel, the quantity that
answers a coding problem; information spectrum calculus evaluates it for concrete channels and returns results on which general
arguments build. All three parts start from the same two lemmas and follow the same pattern of proof. The dashed arrow indicates
that the strong converse, on which the converse parts of Part~I rest, is proved in Part~III by the same method.}
\label{fig:series}
\end{figure*}

For mixed MACs the calculus is simple in principle. On the direct side, the general formula involves the information spectrum of the
mixture, which can be bounded by that of the realized component (Lemma~\ref{lem:bridge}); for each component the law of large
numbers then applies. On the converse side, the error probability over a mixed MAC is the average of the error probabilities over its
components (the identity~\eqref{eq:avg-error}). Hence it suffices to know that a component whose capacity region does not contain the
rate pair fails with probability tending to one, that is, the strong converse in the sense of Wolfowitz
(Definition~\ref{def:sc}). The weaker $\varepsilon$-strong converse, the independence of the $\varepsilon$-capacity region of
$\varepsilon$, does not suffice; the mixed setting thus gives the strong converse a new operational meaning (Remark~\ref{rem:sc-role}).

\paragraph*{The main line}
For the Gaussian components of Secs.~\ref{sec:qsf}--\ref{sec:vr}, the direct parts take three steps: the general formula
(Theorem~\ref{thm:general}), an input lemma that evaluates the component densities of a Gaussian input by the law of large numbers
(Lemma~\ref{lem:gauss-input}), and the reduction from the mixture to the realized component (Lemma~\ref{lem:bridge}); in the variable-rate
coding of Theorem~\ref{thm:vr} the decoder knows the component, so that the random-coding lemma underlying Theorem~\ref{thm:general}
\cite[Lemma~1]{HY-III} and Lemma~\ref{lem:gauss-input} suffice. The converse parts
for $0<\varepsilon<1$ follow from \eqref{eq:avg-error} and the fact that a component in outage fails with probability tending to one
(Lemma~\ref{lem:comp-sc}); for $\varepsilon=0$ a weak converse suffices (Remark~\ref{rem:why-sc}). For general components
(Sec.~\ref{sec:general-comp}), which need not have the strong converse property, the regions are left as multi-letter formulas:
Theorem~\ref{thm:general} and Lemma~\ref{lem:bridge} give \eqref{eq:norx}, and a CSIRT converse lemma and a CSIT coding lemma
(Lemmas~\ref{lem:csirt-conv} and~\ref{lem:csit}) give \eqref{eq:Btilde-formula}. Part~II~\cite{HY-II} follows
the same pattern for finite alphabets, and Part~III~\cite{HY-III} gives the strong converse itself an elementary information spectrum proof,
without resorting to the wringing technique or a hypothesis-testing meta-converse, on the basis of the converse lemma underlying
Theorem~\ref{thm:general} and a spectral-sup bound.

\subsection{Organization}

Sec.~\ref{sec:general} presents the general formula and introduces mixed MACs, and Sec.~\ref{sec:sc} recalls the strong converse
property of the Gaussian MAC. Sec.~\ref{sec:qsf} identifies the $\varepsilon$-capacity region of the quasi-static fading Gaussian MAC
with the $\varepsilon$-outage capacity region. Secs.~\ref{sec:csi} and~\ref{sec:kuser} treat the four CSI scenarios and $K$ users,
and Sec.~\ref{sec:vr} variable-rate coding. Sec.~\ref{sec:general-comp} treats mixed MACs with general components, and
Sec.~\ref{sec:conc} concludes the paper.

%% ====================================================================
\section{General MAC with Cost Constraint}\label{sec:general}
\subsection{Coding System}

\begin{figure}[htbp]
\centering
\includegraphics[width=0.75\textwidth]{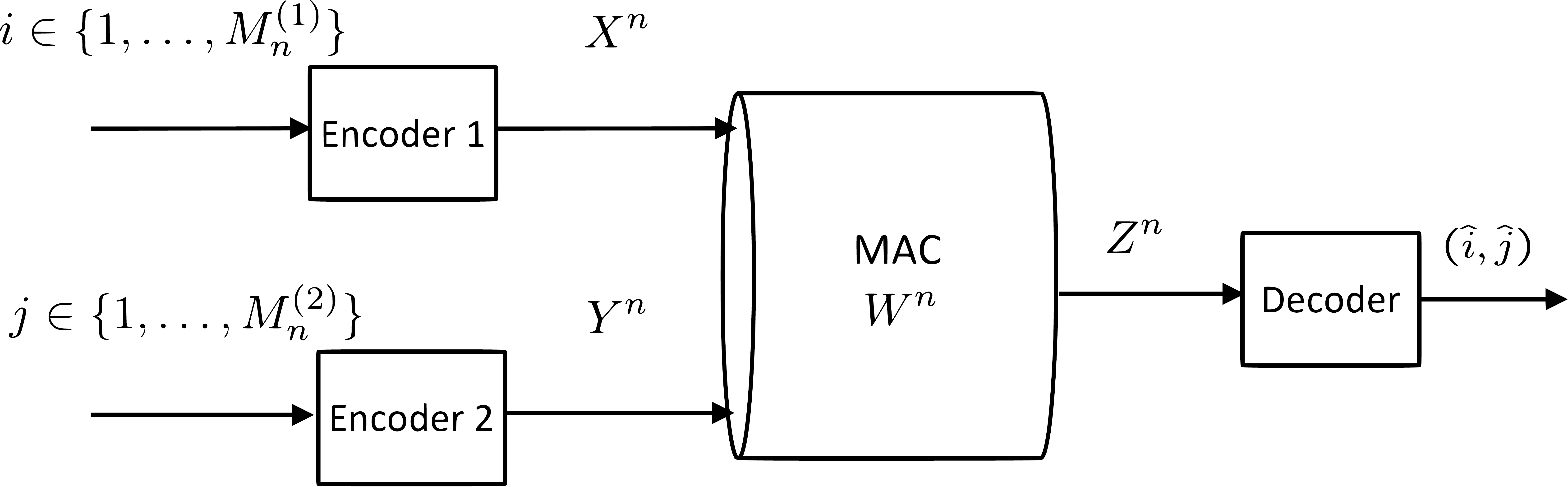}
\caption{The coding system over a two-user MAC.}
\label{fig:coding}
\end{figure}

Let $\cX$ and $\cY$ denote input alphabets for user 1 and user 2, respectively, and let $\cZ$ denote an output
alphabet, where $\cX$, $\cY$ and $\cZ$ may be arbitrary (finite, countably infinite, continuous or abstract). We denote
by $\cP(\cX)$ (resp.\ $\cP(\cY)$) the set of all probability distributions on $\cX$ (resp.\ $\cY$). Consider a two-user
\emph{multiple-access channel} (MAC) $W^n:\cX^n\times\cY^n\to\cZ^n$, where $n=1,2,\cdots$ denotes the blocklength, and the transition
probability $\bW=\{W^n\}_{n=1}^\infty$ is general, i.e., it may be nonstationary and/or nonergodic. The coding system
over the two-user MAC is illustrated in Fig.~\ref{fig:coding}.

Let $C_n^{(1)}$ and $C_n^{(2)}$ be codebooks in $\cX^n$ and $\cY^n$, respectively, with the number of codewords $|C_n^{(\alpha)}|=M_n^{(\alpha)}$
($\alpha=1,2$). As usual, assume that $M_n^{(\alpha)}$ messages are uniformly generated and input to the encoding function
\begin{equation}
f_n^{(\alpha)}:[1:M_n^{(\alpha)}]\to C_n^{(\alpha)}\quad (\alpha\in\{1,2\}),
\end{equation}
where $[1:m]$ denotes the set of integers $\{1,2,\cdots,m\}$. Let $\bu_j=f_n^{(1)}(j)$ and $\bv_k=f_n^{(2)}(k)$ denote the
codeword in $C_n^{(1)}$ for message $j$ and the codeword in $C_n^{(2)}$ for message $k$, respectively. Let
\begin{equation}
g_n:\cZ^n\to[1:M_n^{(1)}]\times[1:M_n^{(2)}]
\end{equation}
be the decoding function with (mutually disjoint) \emph{decoding region} $D_{jk}\subseteq\cZ^n$ for message $(j,k)\in[1:M_n^{(1)}]\times[1:M_n^{(2)}]$,
that is, $D_{jk}=g_n^{-1}(j,k)$. Then, the \emph{average} probability of decoding error over message
pairs is defined as
\begin{equation}\label{eq:avgerr-def}
\varepsilon_n\equiv\frac{1}{M_n^{(1)}M_n^{(2)}}\sum_{j=1}^{M_n^{(1)}}\sum_{k=1}^{M_n^{(2)}}W^n(D_{jk}^c|\bu_j,\bv_k),
\end{equation}
where $D_{jk}^c$ denotes the complement of $D_{jk}$. Such a code pair $C_n=C_n^{(1)}\times C_n^{(2)}$ is referred to as an
$(n,M_n^{(1)},M_n^{(2)},\varepsilon_n)$ MAC code.

\subsection{$\varepsilon$-Capacity Region for General MAC with Cost Constraint}

Let us consider cost functions $c_1:\cX^n\to\R_+$ and $c_2:\cY^n\to\R_+$, where $\R_+$ denotes the set of nonnegative
real numbers, and let $\bX=\{X^n\}_{n=1}^\infty$, $\bY=\{Y^n\}_{n=1}^\infty$ be channel inputs such that $X^n$ and $Y^n$ are independent
and
\begin{align}
\Pr\left\{\frac1n c_1(X^n)\le\Gamma_1\right\}&=1,\label{eq:cost1}\\
\Pr\left\{\frac1n c_2(Y^n)\le\Gamma_2\right\}&=1\label{eq:cost2}
\end{align}
for all $n=1,2,\cdots$. If \eqref{eq:cost1} and \eqref{eq:cost2} hold, we say that the channel input $(\bX,\bY)$ satisfies cost constraint $(\Gamma_1,\Gamma_2)$.
Writing $\Gamma\equiv(\Gamma_1,\Gamma_2)$, we denote by $\cG_{\Gamma}$ the set of all those independent input pairs $(\bX,\bY)$ satisfying cost constraint $(\Gamma_1,\Gamma_2)$;
the same symbol $\cG_{\Gamma}$ ($\cG_{\bGamma}$ for $K$ users) is used for the set of inputs throughout this series.

\begin{definition}\label{def:achievable}
A pair $(R_1,R_2)$ is said to be \emph{$\varepsilon$-achievable} $(0\le\varepsilon<1)$ with cost constraint $(\Gamma_1,\Gamma_2)$ if there
exists a sequence of $(n,M_n^{(1)},M_n^{(2)},\varepsilon_n)$ MAC codes $C_n=C_n^{(1)}\times C_n^{(2)}$ satisfying
\begin{equation}\label{eq:achievable}
\limsup_{n\to\infty}\varepsilon_n\le\varepsilon,\quad \liminf_{n\to\infty}\frac1n\log M_n^{(1)}\ge R_1,\quad \liminf_{n\to\infty}\frac1n\log M_n^{(2)}\ge R_2
\end{equation}
in addition to \eqref{eq:cost1} and \eqref{eq:cost2} with $X^n$ and $Y^n$ uniformly distributed on $C_n^{(1)}$ and $C_n^{(2)}$, respectively. The closure
of the set of all $\varepsilon$-achievable rate pairs is called the \emph{$\varepsilon$-capacity region} with cost constraint $(\Gamma_1,\Gamma_2)$, and is
denoted by $C_{\Gamma}(\varepsilon\,|\,\bW)$, $\Gamma=(\Gamma_1,\Gamma_2)$. In particular, when $\varepsilon=0$, the set $C_{\Gamma}(0\,|\,\bW)$ is called simply the capacity
region with cost constraint $(\Gamma_1,\Gamma_2)$.
\end{definition}

First, we present a very general formula for the $\varepsilon$-capacity region with cost constraint $(\Gamma_1,\Gamma_2)$ by extending
the established formulas for the $\varepsilon$-capacity region \emph{without} cost constraints. Let $\bW=\{W^n\}_{n=1}^\infty$ with $W^n:
\cX^n\times\cY^n\to\cZ^n$ be a general MAC, where alphabets $\cX,\cY$ and $\cZ$ are arbitrary. With output $Z^n$ via channel
$W^n$ due to input $(X^n,Y^n)$, define the \emph{information densities}
\begin{equation}\label{eq:densities}
\iota^n_1\equiv\frac1n\log\frac{W^n(Z^n|X^n,Y^n)}{P_{Z^n|Y^n}(Z^n|Y^n)},\qquad
\iota^n_2\equiv\frac1n\log\frac{W^n(Z^n|X^n,Y^n)}{P_{Z^n|X^n}(Z^n|X^n)},\qquad
\iota^n_3\equiv\frac1n\log\frac{W^n(Z^n|X^n,Y^n)}{P_{Z^n}(Z^n)}
\end{equation}
(logarithms are natural throughout), and, for a threshold vector $a=(a_1,a_2,a_3)\in\R^3$, the event
\begin{equation}\label{eq:En}
E_n(a\,|\,\bX,\bY)\equiv\bigcup_{k=1}^3\{\iota^n_k\le a_k\},
\end{equation}
which grows with each $a_k$. For a pair $R=(R_1,R_2)$ we write $E_n(R\,|\,\bX,\bY)\equiv E_n(R_1,R_2,R_1+R_2\,|\,\bX,\bY)$, and define
\begin{equation}\label{eq:J}
J_{\bW}(R_1,R_2\,|\,\bX,\bY)\equiv\limsup_{n\to\infty}\Pr\bigl\{E_n(R_1,R_2,R_1+R_2\,|\,\bX,\bY)\bigr\},
\end{equation}
which is nondecreasing in each of $R_1$ and $R_2$. When the channel is obvious from the context, we simply denote this
function as $J(R_1,R_2\,|\,\bX,\bY)$. For a pair $R$ and $c\in\R$ we write $R-c\equiv(R_1-c,R_2-c)$. Then, the
general formula for the $\varepsilon$-capacity region with cost constraints immediately follows by literally paralleling
the argument as developed to prove the case without cost constraints in \cite[Theorem 5]{Han98}, \cite[Theorem 7.11.1]{Han03}:

\begin{theorem}\label{thm:general}
For a general MAC $\bW$, the $\varepsilon$-capacity region for $0\le\varepsilon<1$ is given by
\begin{equation}\label{eq:general}
C_{\Gamma}(\varepsilon\,|\,\bW)=\Cl_+\biggl(\bigcup_{(\bX,\bY)\in\cG_{\Gamma}}A_\varepsilon(\bX,\bY)\biggr),
\end{equation}
where
\begin{equation}\label{eq:Aeps}
A_\varepsilon(\bX,\bY)\equiv\bigl\{(R_1,R_2)\in\R^2\st J(R_1,R_2\,|\,\bX,\bY)\le\varepsilon\bigr\},
\end{equation}
and, for a set $A\subseteq\R^2$,
\begin{equation}\label{eq:Clplus}
\Cl_+(A)\equiv\Cl(A)\cap\R^2_+
\end{equation}
with $\Cl$ denoting the closure in $\R^2$. \hfill$\Box$
\end{theorem}

The cost constraint enters the proof only through $\cG_{\Gamma}$: random codewords drawn from an input in $\cG_{\Gamma}$ satisfy it, and the input
uniform on the codewords of a code satisfying it belongs to $\cG_{\Gamma}$; otherwise the proof of \cite[Theorem 5]{Han98} applies without change (single-user
analogue: \cite{VH94}).

\begin{remark}[Convention on closures and relation to Part II]\label{rem:closure}
In \eqref{eq:general}, the sets $A_\varepsilon(\bX,\bY)$ are allowed to
contain rate pairs with negative components, and the nonnegativity of rates is imposed only after the
closure of the union has been taken. This convention cannot be dropped: restricting $A_\varepsilon(\bX,\bY)$ to $\R^2_+$ before
taking the closure would make the formula false. Indeed, consider a MAC in which the input of user 2 does
not affect the output, i.e., $W^n(\bz|\bx,\by)=V^n(\bz|\bx)$. For every independent input, $P_{Z^n|X^n}=W^n$, so that
$\iota^n_2\equiv0$ in \eqref{eq:densities} and $J(R_1,R_2\,|\,\bX,\bY)=1$ whenever $R_2\ge0$. Thus
$A_\varepsilon(\bX,\bY)\cap\R^2_+=\emptyset$ for every input and every $\varepsilon<1$, whereas every pair $(R_1,0)$ with $R_1$ less than the
$\varepsilon$-capacity of $\bV=\{V^n\}_{n=1}^\infty$ is $\varepsilon$-achievable with $M_n^{(2)}=1$; these pairs are recovered in \eqref{eq:general} as limits of pairs
in $A_\varepsilon(\bX,\bY)$ with $R_2<0$ (see \cite[Remark~1]{HY-II}).

Formula \eqref{eq:general} is identical with the general formula \cite[Theorem~1, Eq.~(7)]{HY-II} of Part II (cf.\ \cite[Eqs.~(1), (6)]{HY-II}), the starting
point of its single-letter analysis; Part III \cite{HY-III} starts one step earlier, from the random-coding and converse lemmas from which it is derived.
The operation $\Cl_+$ is used in the same sense for all regions described by information spectrum quantities (Theorem~\ref{thm:gencsi}); regions stated
directly in $\R^2_+$ in terms of strict outage conditions (\eqref{eq:outage-form} and \eqref{eq:kuser} below) do not need this convention.
\end{remark}

\begin{remark}[Relation to the preliminary version]\label{rem:prelim}
In the preliminary version \cite[v2]{HY-arxiv} of this paper, the general formula was stated with the operation $\Cl_+$ placed in
front of each $A_\varepsilon(\bX,\bY)$ rather than in front of the union, with the convention on negative rates given in a footnote.
The two forms are equal:
\begin{equation}\label{eq:clplus-swap}
\bigcup_{(\bX,\bY)\in\cG_{\Gamma}}\Cl_+\bigl(A_\varepsilon(\bX,\bY)\bigr)=\Cl_+\biggl(\bigcup_{(\bX,\bY)\in\cG_{\Gamma}}A_\varepsilon(\bX,\bY)\biggr).
\end{equation}
The inclusion ``$\subseteq$'' is obvious, and ``$\supseteq$'' follows from the closedness of $\bigcup_{(\bX,\bY)}\Cl(A_\varepsilon(\bX,\bY))$, which is shown by a
diagonal argument over the inputs in Part II \cite[Remark~A]{HY-II}. Hence all results stated in the form of \cite[v2]{HY-arxiv} remain valid.
We adopt the form \eqref{eq:general}, which is manifestly closed and is literally identical with the general formula of Part II
\cite[Theorem~1, Eq.~(7)]{HY-II}.
\end{remark}

The capacity region ($\varepsilon=0$) can be obtained from Theorem~\ref{thm:general} as follows. To do so, we introduce quantities
used in the information spectrum method \cite{Han03}. Given an arbitrary sequence of random variables $V^n$, we
define the \emph{liminf in probability} and the \emph{limsup in probability} (see \cite{HV93}, \cite{Han03}) as
\begin{align}
\plimi V^n&\equiv\sup\Bigl\{\beta\st\lim_{n\to\infty}\Pr\{V^n<\beta\}=0\Bigr\},\notag\\
\plims V^n&\equiv\inf\Bigl\{\alpha\st\lim_{n\to\infty}\Pr\{V^n>\alpha\}=0\Bigr\},\label{eq:plim}
\end{align}
respectively; p-$\liminf$ is superadditive and p-$\limsup$ subadditive \cite[Sec.~1.3]{Han03}.
Let $\bW=\{W^n\}_{n=1}^\infty$ with $W^n:\cX^n\times\cY^n\to\cZ^n$ be a general MAC. Given an input pair
$(\bX,\bY)\in\cG_{\Gamma}$, let us define the (conditional) \emph{information spectral-inf} by
\begin{align}
\underline{I}(\bX;\bZ\,|\,\bY)&\equiv\plimi\frac1n\log\frac{W^n(Z^n|X^n,Y^n)}{P_{Z^n|Y^n}(Z^n|Y^n)},\\
\underline{I}(\bY;\bZ\,|\,\bX)&\equiv\plimi\frac1n\log\frac{W^n(Z^n|X^n,Y^n)}{P_{Z^n|X^n}(Z^n|X^n)},\\
\underline{I}(\bX\bY;\bZ)&\equiv\plimi\frac1n\log\frac{W^n(Z^n|X^n,Y^n)}{P_{Z^n}(Z^n)}.
\end{align}
In terms of these quantities, we also define the set $\mathcal{R}_{\bW}(\bX,\bY)$ as
\begin{align}
\mathcal{R}_{\bW}(\bX,\bY)\equiv\{(R_1,R_2)\,|\,&0\le R_1\le\underline{I}(\bX;\bZ\,|\,\bY),\notag\\
&0\le R_2\le\underline{I}(\bY;\bZ\,|\,\bX),\notag\\
&R_1+R_2\le\underline{I}(\bX\bY;\bZ)\}.
\end{align}
Obviously, the set $\mathcal{R}_{\bW}(\bX,\bY)$ is closed for any given general sources $(\bX,\bY)$. Then, the $0$-capacity region
for general MAC $\bW$ can be formulated as

\begin{corollary}[Han {\cite[Theorem 7.13.1]{Han03}}]\label{cor:zero}
For a general MAC $\bW$, the $0$-capacity region is given by
\begin{equation}
C_{\Gamma}(0\,|\,\bW)=\bigcup_{(\bX,\bY)\in\cG_{\Gamma}}\mathcal{R}_{\bW}(\bX,\bY).
\end{equation}
\end{corollary}

\subsection{A Basic Lemma}\label{sec:basic}

We shall repeatedly use the following change-of-measure bound \cite[Lemma 3.2.1]{Han03}, whose proof is provided in
Appendix~\ref{app:lemma1} for completeness. Here and in what follows, for random variables taking values in a general (e.g.,
continuous) alphabet, $P_{Z^n}$, $P_{Z^n|Y^n}$, etc.\ denote densities with respect to a common $\sigma$-finite reference measure
(e.g., the Lebesgue measure in the Gaussian case), and sums are to be read as integrals.

\begin{lemma}[Change of measure]\label{lem:basic}
Let $(T,U)$ be random variables, where $U$ takes values in a general space and has the conditional density
$P_{U|T}(\cdot|t)$ given $T=t$, and let $Q(\cdot|t)$ be a probability density for every $t$, measurable in $(t,u)$. Then, for every
$a>0$,
\begin{equation}\label{eq:basic}
\Pr\bigl\{Q(U|T)>e^aP_{U|T}(U|T)\bigr\}\le e^{-a}.
\end{equation}
\end{lemma}

In applications, $a=n\gamma$, and \eqref{eq:basic} states that $\frac1n\log\frac{1}{Q(U|T)}\ge\frac1n\log\frac{1}{P_{U|T}(U|T)}-\gamma$ with probability at least
$1-e^{-n\gamma}$. The same lemma is \cite[Lemma~1]{HY-II}.

\subsection{Mixed MACs}\label{sec:mixed}

Let $\Theta$ be an arbitrary parameter space (state space) with a probability measure $w$, and assign a general
MAC $\bW_\theta=\{W_\theta^n\}_{n=1}^\infty$ to each $\theta\in\Theta$, which are called \emph{component MACs}. Here, we assume that each $\bW_\theta$
has the same input alphabets $\cX,\cY$ and output alphabet $\cZ$, and that $\theta\mapsto W^n_\theta(\cdot|\bx,\by)$ is measurable. The
\emph{mixed MAC} $\bW=\{W^n\}_{n=1}^\infty$ is defined by the transition probability
\begin{equation}\label{eq:mixed}
W^n(\bz|\bx,\by)=\int_{\Theta}W^n_\theta(\bz|\bx,\by)\,dw(\theta)\quad(\forall n=1,2,\cdots;\forall(\bx,\by,\bz)\in\cX^n\times\cY^n\times\cZ^n).
\end{equation}
That is, the component $\theta$ is drawn according to $w$ prior to the transmission and is held fixed during the
whole block. When all the components are stationary memoryless, i.e., $W^n_\theta(\bz|\bx,\by)=\prod_{i=1}^nW_\theta(z_i|x_i,y_i)$,
the mixed MAC is called a \emph{mixed memoryless MAC}, and the component is denoted simply by $\bW_\theta=\{W_\theta\}$.

For a general input $(\bX,\bY)$ and $\theta\in\Theta$, let $Z^n_\theta$ denote the output of the component $W^n_\theta$ due to $(X^n,Y^n)$. We
evaluate the information densities $\iota^n_k$ of \eqref{eq:densities}, which are defined through the mixture $W^n$, at $(X^n,Y^n,Z^n_\theta)$,
and write $E_n(a\,|\,\bX,\bY;\theta)$ for the event \eqref{eq:En} so obtained. Since $Z^n$ is generated by first drawing $\theta$ according
to $w$,
\begin{equation}\label{eq:mixrep}
\Pr\{E_n(a\,|\,\bX,\bY)\}=\int_{\Theta}\Pr\{E_n(a\,|\,\bX,\bY;\theta)\}\,dw(\theta).
\end{equation}
We also use the \emph{component densities}, formed with the densities of the component itself:
\begin{equation}\label{eq:compdens}
\iota^n_{1,\theta}\equiv\frac1n\log\frac{W^n_\theta(Z^n_\theta|X^n,Y^n)}{P_{Z^n_\theta|Y^n}(Z^n_\theta|Y^n)},\qquad
\iota^n_{2,\theta}\equiv\frac1n\log\frac{W^n_\theta(Z^n_\theta|X^n,Y^n)}{P_{Z^n_\theta|X^n}(Z^n_\theta|X^n)},\qquad
\iota^n_{3,\theta}\equiv\frac1n\log\frac{W^n_\theta(Z^n_\theta|X^n,Y^n)}{P_{Z^n_\theta}(Z^n_\theta)},
\end{equation}
and the corresponding event
\begin{equation}\label{eq:Eneta}
E^\theta_n(a\,|\,\bX,\bY)\equiv\bigcup_{k=1}^3\{\iota^n_{k,\theta}\le a_k\}.
\end{equation}
The same notation is used in Part II \cite[Sec.~IV-B]{HY-II}; in Secs.~\ref{sec:qsf}--\ref{sec:vr}, $\theta=(h,\nu)$, where $\nu$ indexes the noise variance.

The two families of densities differ only in their denominators (and, in the numerators, $W^n$ versus $W^n_\theta$): the mixture
densities are what a decoder that does not know $\theta$ can compute, whereas the component densities are those of the
realized component, to which the law of large numbers or a strong converse can be applied. The following lemma, used
in all direct parts of this paper and of Part II, states that the former are not much smaller than the latter as soon as
the mixture dominates the realized component up to a factor $c_n$.

\begin{lemma}[From mixture to components]\label{lem:bridge}
Let $(\bX,\bY)$ be a general input, $\theta\in\Theta$ and $\gamma>0$. Suppose that, for each $n$, there are a set $G_{\theta,n}\subseteq\cX^n\times\cY^n\times\cZ^n$
and a constant $0<c_n\le1$ such that
\begin{equation}\label{eq:bridge-hyp}
W^n(\bz|\bx,\by)\ge c_n\,W^n_\theta(\bz|\bx,\by)\qquad\mbox{for all }(\bx,\by,\bz)\in G_{\theta,n}.
\end{equation}
Then, with $\gamma_n\equiv\gamma+\frac1n\log\frac{1}{c_n}$, for every $a\in\R^3$,
\begin{equation}\label{eq:bridge}
\Pr\{E_n(a\,|\,\bX,\bY;\theta)\}\le\Pr\{E^\theta_n(a+\gamma_n\bone\,|\,\bX,\bY)\}+\Pr\{(X^n,Y^n,Z^n_\theta)\notin G_{\theta,n}\}+3e^{-n\gamma},
\end{equation}
where $\bone\equiv(1,1,1)$.
\end{lemma}

\begin{proof}
On $G_{\theta,n}$, the numerator of each mixture density is at least $c_n$ times that of the corresponding component density.
For the denominators, Lemma~\ref{lem:basic} with $a=n\gamma$, applied with $T=Y^n$, $U=Z^n_\theta$ and $Q=P_{Z^n|Y^n}$, gives
$P_{Z^n|Y^n}(Z^n_\theta|Y^n)\le e^{n\gamma}P_{Z^n_\theta|Y^n}(Z^n_\theta|Y^n)$ with probability at least $1-e^{-n\gamma}$; similarly with $T=X^n$ and
$Q=P_{Z^n|X^n}$, and with $T$ trivial and $Q=P_{Z^n}$. On the intersection of $G_{\theta,n}$ with these three events, each mixture
density evaluated at $(X^n,Y^n,Z^n_\theta)$ is at least the corresponding component density minus $\gamma_n$. The union bound
gives \eqref{eq:bridge}.
\end{proof}

\begin{remark}[Meaning and instances of Lemma~\ref{lem:bridge}]\label{rem:bridge}
In engineering terms, $\frac1n\log\frac1{c_n}$ is the rate that a decoder ignorant of $\theta$ gives up compared with one that knows $\theta$; it
vanishes as soon as the realized component, or a set of components behaving like it, has probability not exponentially small. The lemma is
applied (i) for a countable $\Theta$, with $c_n=w(\theta)$ and $G_{\theta,n}$ the whole space (Sec.~\ref{sec:general-comp}); (ii) for Gaussian components, with
$c_n=w(\cK(\theta))e^{-n\gamma}$ for a small cell $\cK(\theta)$ of parameters around $\theta$ and $G_{\theta,n}$ a set of typical noise vectors (proof of
Theorem~\ref{thm:qsf}); and (iii) for finite-alphabet components, with $c_n=e^{-\sqrt n}$ for all $\theta$ outside a set of probability tending to zero
\cite[Lemma~5]{HY-II}. In all cases the law of large numbers, or the strong converse, is applied only to the component densities on the right-hand
side of \eqref{eq:bridge}, one realized component at a time, never to the mixture.
\end{remark}

For an $(n,M_n^{(1)},M_n^{(2)},\varepsilon_n)$ MAC code with codewords $\bu_j,\bv_k$ and decoding regions $D_{jk}$ used over a mixed
MAC, let
\begin{equation}
\varepsilon_{\theta,n}\equiv\frac{1}{M_n^{(1)}M_n^{(2)}}\sum_{j=1}^{M_n^{(1)}}\sum_{k=1}^{M_n^{(2)}}W^n_\theta(D^c_{jk}|\bu_j,\bv_k)
\end{equation}
denote the average probability of decoding error of the same code used over the component MAC $\bW_\theta$. Then,
by \eqref{eq:avgerr-def} and \eqref{eq:mixed},
\begin{equation}\label{eq:avg-error}
\varepsilon_n=\int_{\Theta}\varepsilon_{\theta,n}\,dw(\theta).
\end{equation}
This elementary identity, combined with the strong converse property of the components, will yield the
converse parts of our coding theorems for quasi-static fading Gaussian MACs.

For finite-alphabet mixed memoryless MACs with additive cost constraints, Theorem~\ref{thm:general} is specialized in Part II \cite{HY-II} (additive MACs,
a single-letter inner bound, the $0$-capacity region in terms of essential infima, which reduces to the formulas of \cite{Ahl71}, \cite{Liao72} and
\cite[Theorem 1]{Ahl68}, \cite[Theorem 3.6.5]{Han03}, and an outer bound valid in all four CSI scenarios). Here we focus on Gaussian components
(Secs.~\ref{sec:qsf}--\ref{sec:vr}) and then treat general components (Sec.~\ref{sec:general-comp}).

%% ====================================================================
\section{Strong Converse of Gaussian MAC}\label{sec:sc}

In this section, we assume that $W^n$ is a \emph{Gaussian memoryless} MAC (without mixture), that is, the channel
is subject to additive i.i.d.\ Gaussian noise
\begin{equation}
V_i^{(n)}\sim\cN(0,N)\quad(i=1,2,\ldots,n),
\end{equation}
where $\cN(0,\sigma^2)$ denotes the Gaussian distribution of mean 0 and variance $\sigma^2$. The alphabets are assumed
to be $\cX=\cY=\cZ=\mathcal{V}=\R$. In this case, the cost of input symbols is measured by its power as $c_1(x)=x^2$
and $c_2(y)=y^2$, and according to the tradition, we write cost constraint $(P_1,P_2)$ instead of $(\Gamma_1,\Gamma_2)$ (the set of inputs is still denoted
$\cG_{\Gamma}$, now with $\Gamma=(P_1,P_2)$) as
\begin{align}
\Pr\left\{\frac1n\sum_{i=1}^n\bigl(X_i^{(n)}\bigr)^2\le P_1\right\}&=1,\label{eq:power1}\\
\Pr\left\{\frac1n\sum_{i=1}^n\bigl(Y_i^{(n)}\bigr)^2\le P_2\right\}&=1,\label{eq:power2}
\end{align}
where $X^n=\bigl(X_1^{(n)},\ldots,X_n^{(n)}\bigr)$ and $Y^n=\bigl(Y_1^{(n)},\ldots,Y_n^{(n)}\bigr)$ are input sequences. The output sequence $Z^n=
(Z_1^{(n)},Z_2^{(n)},\ldots,Z_n^{(n)})$ due to input $(X^n,Y^n)$ is given by $Z_i^{(n)}=X_i^{(n)}+Y_i^{(n)}+V_i^{(n)}$ $(i=1,2,\ldots,n)$. The
cost constraints are now referred to as \emph{power constraint} $(P_1,P_2)$.

The $0$-capacity region of the Gaussian MAC with power constraint $(P_1,P_2)$ is given by

\begin{theorem}[Cover \cite{Cov75}, Wyner \cite{Wyn74}]\label{thm:wyner}
For a Gaussian memoryless MAC $\bW=\{W\}$, the $0$-capacity region with power
constraint $(P_1,P_2)$ and noise power $N$ is given by
\begin{align}
C_{P_1,P_2}(0\,|\,\bW)=\biggl\{(R_1,R_2)\st&0\le R_1\le\frac12\log\Bigl(1+\frac{P_1}{N}\Bigr),\notag\\
&0\le R_2\le\frac12\log\Bigl(1+\frac{P_2}{N}\Bigr),\notag\\
&R_1+R_2\le\frac12\log\Bigl(1+\frac{P_1+P_2}{N}\Bigr)\biggr\}.
\end{align}
\end{theorem}

Here, we introduce the notion of the strong converse property (cf.\ \cite[Definition 7.12.1]{Han03})

\begin{definition}\label{def:sc}
A general MAC $\bW$ with cost constraint $(\Gamma_1,\Gamma_2)$ is said to have the \emph{strong converse property}
if, for every $(R_1,R_2)\in\R^2_+$ with $(R_1,R_2)\notin C_{\Gamma}(0\,|\,\bW)$, every sequence of $(n,M_n^{(1)},M_n^{(2)},\varepsilon_n)$ codes satisfying the cost constraint with
\begin{align}
\liminf_{n\to\infty}\frac1n\log M_n^{(1)}&\ge R_1,\label{eq:sc-rate1}\\
\liminf_{n\to\infty}\frac1n\log M_n^{(2)}&\ge R_2\label{eq:sc-rate2}
\end{align}
satisfies
\begin{equation}
\lim_{n\to\infty}\varepsilon_n=1.
\end{equation}
\end{definition}

The basic result of this section is the following theorem (see Remark~\ref{rem:fongtan}).

\begin{theorem}[Strong Converse {\cite{FT16}}]\label{thm:sc}
For every $P_1,P_2>0$ and $N>0$, the Gaussian memoryless MAC $\bW=\{W\}$
with noise power $N$ has the strong converse property with power constraint $(P_1,P_2)$. \hfill$\Box$
\end{theorem}

\begin{remark}[$K$-user Gaussian MAC]\label{rem:kuser-sc}
Theorem~\ref{thm:sc} holds also for the $K$-user Gaussian memoryless MAC with
output $Z_i^{(n)}=X_{1,i}^{(n)}+\cdots+X_{K,i}^{(n)}+V_i^{(n)}$, $V_i^{(n)}\sim\cN(0,N)$, and power constraint $\bP=(P_1,\ldots,P_K)$, whose
$0$-capacity region is
\begin{equation}
\biggl\{(R_1,\ldots,R_K)\in\R^K_+\st\sum_{k\in\cS}R_k\le\frac12\log\Bigl(1+\frac{\sum_{k\in\cS}P_k}{N}\Bigr)\mbox{ for all nonempty }\cS\subseteq\{1,\ldots,K\}\biggr\}:
\end{equation}
any sequence of codes whose rate tuple lies outside this region has average probability of decoding error
tending to one. An information spectrum proof of this $K$-user strong converse, which contains Theorem~\ref{thm:sc}
as the case $K=2$ and does not use the wringing technique, is given in Part III \cite{HY-III}. The $K$-user version is used in Secs.~\ref{sec:kuser} and \ref{sec:vr}.
\end{remark}

\begin{remark}\label{rem:fongtan}
Fong and Tan \cite{FT16} showed that the $\varepsilon$-capacity region of the Gaussian stationary memoryless MAC with
power constraint $(P_1,P_2)$ remains the same as the $0$-capacity region, i.e.,
\begin{equation}\label{eq:eps-sc}
C_{P_1,P_2}(\varepsilon\,|\,\bW)=C_{P_1,P_2}(0\,|\,\bW)\quad(0\le\varepsilon<1).
\end{equation}
The strong converse property of Definition~\ref{def:sc} is stronger than \eqref{eq:eps-sc}: if, for some $0<\varepsilon<1$, $C_{P_1,P_2}(\varepsilon\,|\,\bW)$
contained a rate pair not in $C_{P_1,P_2}(0\,|\,\bW)$, Definition~\ref{def:sc} would give $\varepsilon\ge\lim_{n\to\infty}\varepsilon_n=1$. The converse implication
fails \cite[Remark 3.5.1]{Han03}. The $\varepsilon$-strong converse goes back to Wolfowitz \cite{Wol57} (see also \cite{Ahl82}, \cite{CK11}); the strong converse of Definition~\ref{def:sc}
appeared, at latest, in Gallager \cite{Gal68} for single-user finite-alphabet channels, who credited it  to the same Wolfowitz \cite{Wol57}.

What does imply the strong converse property is a \emph{per-blocklength} bound: for every $0<\lambda<1$ there is $\delta_n(\lambda)\to0$ such
that, for every $n$, every $(n,M_n^{(1)},M_n^{(2)},\varepsilon_n)$ code satisfying the cost constraint with $\varepsilon_n\le\lambda$ has a rate pair
$(\frac1n\log M_n^{(1)},\frac1n\log M_n^{(2)})$ within distance $\delta_n(\lambda)$ of $C_{\Gamma}(0\,|\,\bW)$. Indeed, if a sequence of codes satisfying
\eqref{eq:sc-rate1} and \eqref{eq:sc-rate2} with $(R_1,R_2)\notin C_{\Gamma}(0\,|\,\bW)$ had $\liminf_{n\to\infty}\varepsilon_n<1$, then $\varepsilon_{n_j}\le\lambda$ along a
subsequence for some $\lambda<1$, and the bound at the blocklengths $n_j$ would place $(R_1,R_2)$ in $C_{\Gamma}(0\,|\,\bW)$, which is closed and
contains, together with every rate pair, all smaller nonnegative rate pairs; a contradiction. The non-asymptotic bound of Ahlswede
\cite[Theorem]{Ahl82} for the discrete memoryless MAC (cf.\ \cite[Lemma~10]{HY-II}) and the argument of Fong and Tan \cite[Theorem 1]{FT16},
which uses the wringing technique of Dueck \cite{Dueck81} and Ahlswede \cite{Ahl82} combined with quantization of the input spaces, give bounds of
this form; hence the credit for Theorem~\ref{thm:sc} (including the $K$-user case of Remark~\ref{rem:kuser-sc}) goes to \cite{FT16}.
\end{remark}

\begin{remark}\label{rem:why-sc}
Theorem~\ref{thm:sc} is the key tool in the converse parts of Theorems~\ref{thm:qsf} and \ref{thm:csi}, and, through Remark~\ref{rem:kuser-sc}, of
Theorem~\ref{thm:kuser}, for $0<\varepsilon<1$: it is through the strong converse property that the information spectrum calculus of the general formula
is completed on the converse side. Theorem~\ref{thm:sc} is applied to each component separately: for each fixed component $\theta$ whose
capacity region does not contain the rate pair, $\varepsilon_{\theta,n}$ of \eqref{eq:avg-error} tends to one, and Fatou's lemma bounds
$\liminf_{n\to\infty}\varepsilon_n$ from below; no uniformity in $\theta$ is required. Thus a component in outage fails almost surely as the blocklength grows. The $\varepsilon$-strong converse \eqref{eq:eps-sc} would be insufficient:
it gives only $\limsup_{n\to\infty}\varepsilon_{\theta,n}=1$ for such components, which does not bound the $\liminf$ of the integral in \eqref{eq:avg-error}. An
alternative route via Lemma~\ref{lem:csirt-conv} is mentioned in Remark~\ref{rem:alt}.

For $\varepsilon=0$ the strong converse is not needed: if $\varepsilon_n\to0$, then $\varepsilon_{\theta,n_j}\to0$ along a subsequence for $w$-almost all $\theta$, and
the weak converse (Fano's inequality) for each such component places the rate pair in its capacity region. This yields the converse parts of
Theorem~\ref{thm:qsf0}, Corollary~\ref{cor:kuser0} and Theorem~\ref{thm:vr}. For finite alphabets, the counterpart of the first paragraph of this remark is
\cite[Remark~7]{HY-II}, and the $0$-capacity region is determined there with Fano's inequality applied to the common code \cite[Theorem~5]{HY-II}.
\end{remark}

\begin{remark}[On the role of the strong converse]\label{rem:sc-role}
The use of the strong converse above might suggest that it is the main ingredient of this paper. It is not. All results rest on
Theorem~\ref{thm:general} and on the random-coding and converse lemmas from which it is derived (for the direct parts, also through the random-coding
bound of Remark~\ref{rem:kuser-general}, and, in Sec.~\ref{sec:general-comp}, together with the CSIRT converse and CSIT coding lemmas), and the converse parts
reduce, through \eqref{eq:avg-error}, to an evaluation of the information spectra of the components. Theorem~\ref{thm:sc} enters only through Lemma~\ref{lem:comp-sc}
(a component in outage fails with probability tending to one), which is itself such an evaluation, i.e., a part of the calculus of
Sec.~\ref{sec:intro-unified}: in Part~III~\cite{HY-III} it is obtained from the converse fundamental lemma \cite[Lemma~2]{HY-III} (Lemma~\ref{lem:csirt-conv} below for a
single channel) and the spectral-sup bound \cite[Lemma~4]{HY-III}.

For a single channel the distinction between the strong converse of Definition~\ref{def:sc} and the $\varepsilon$-strong converse does not affect the
capacity region (cf.\ Remark~\ref{rem:fongtan}). For mixed channels it is decisive: it is exactly what makes the converse parts of the main theorems
follow from \eqref{eq:avg-error} and Fatou's lemma, and it thereby gives the strong converse in the sense of Definition~\ref{def:sc} a new operational meaning.
\end{remark}

%% ====================================================================
\section{Quasi-Static Fading MAC as a Mixed Gaussian MAC}\label{sec:qsf}
\subsection{Channel Model}

In this section, we consider a generalization of the quasi-static fading MAC, a popular model in wireless communication (cf.\ \cite{BPS98},
\cite{EGL10}, \cite{OSW94}, \cite{Tel99}, \cite{YDKP13}, \cite{YDKP14}), as a special case of the mixed memoryless MACs of Sec.~\ref{sec:mixed}; we refer to
it as the quasi-static fading Gaussian MAC.

In order to formulate this problem, let $\Phi=(-\infty,+\infty)\times(-\infty,+\infty)$ and $[0,1]$ be the set of fading coefficients
(channel gains) $h=(h_1,h_2)$ and the set of indices $\nu$ which parameterize the variance $N_\nu$ of Gaussian noise,
respectively. The overall parameter space $\Theta$ is given by $\Theta=\Phi\times[0,1]$, where $\theta=(h_1,h_2,\nu)\in\Theta$ is a random
variable generated subject to a probability measure $w(\theta)$ on $\Theta$ in advance to encoding. The component
MACs $\bW_\theta=\{W_\theta^n\}_{n=1}^\infty$ with $\theta=(h_1,h_2,\nu)\in\Theta$ fixed, denoted simply by $\bW_\theta=\{W_\theta\}$, are stationary and
memoryless MACs with input-output alphabets $\cX=\cY=\cZ=\R$. The output $Z^n_\theta$ via $W^n_\theta$ due to input
$(X^n,Y^n)$ is given by
\begin{equation}\label{eq:qsf-out}
Z^n_\theta=h_1X^n+h_2Y^n+V^n_\nu,
\end{equation}
where $V^n_\nu=\bigl(V^{(n)}_{\nu,1},V^{(n)}_{\nu,2},\ldots,V^{(n)}_{\nu,n}\bigr)$ is an additive i.i.d.\ Gaussian noise sequence with $V^{(n)}_{\nu,i}\sim\cN(0,N_\nu)$. The parameters $h_1,h_2$ and $\nu$ may be correlated; the noise variance depends on $\nu$ to account for thermal noise and
inter-cell interference that may vary over coherence times.

The transition probability of the quasi-static fading MAC $\bW=\{W^n\}_{n=1}^\infty$ is given by
\begin{equation}
W^n(\bz|\bx,\by)=\int_{\Theta}W^n_\theta(\bz|\bx,\by)\,dw(\theta)
\end{equation}
for $\bx\in\cX^n,\by\in\cY^n$ and $\bz\in\cZ^n$, i.e., $\bW$ is a mixed memoryless MAC with Gaussian components $\bW_\theta=\{W_\theta\}$, so that
Theorem~\ref{thm:general} and \eqref{eq:avg-error} apply directly. When the dependence on $(h,\nu)$ is emphasized, $\bW_\theta$ and $W^n_\theta$ are denoted by
$\bW_{h,\nu}$ and $W^n_{h,\nu}$. When $N_\nu=N$ for all $\nu$, only $h=(h_1,h_2)$ varies, and the model is the quasi-static fading MAC of the literature,
e.g., \cite{KAFP19}, \cite{Nar07}, \cite{TH98}, \cite{TVZ04}.

Throughout Secs.~\ref{sec:qsf}--\ref{sec:vr}, we impose the following assumption on the noise variance.
\begin{enumerate}[leftmargin=2.2em,label=(A)]
\item\label{assumpA} The variance $N_\nu$ is continuous in $\nu\in[0,1]$ and positive for all $\nu\in[0,1]$.
\end{enumerate}
Since $[0,1]$ is a compact set, assumption (A) implies that $N_\nu$ is uniformly continuous and uniformly
bounded, and it holds that
\begin{equation}\label{eq:Nminmax}
N_{\min}\equiv\min_{\nu\in[0,1]}N_\nu>0\quad\mbox{and}\quad N_{\max}\equiv\max_{\nu\in[0,1]}N_\nu<+\infty.
\end{equation}
No assumption is imposed on the probability measure $w$ on $\Theta$ (other than that it is defined on the Borel
sets of $\Theta\subseteq\R^3$): $w$ may have atoms, an absolutely continuous part and a singular continuous part, and $h_1$,
$h_2$ and $\nu$ may be correlated in an arbitrary manner. Throughout the paper, we let $F$ denote the random
variable taking values in $\Theta$ subject to $P_F=w$, and we assume $P_1>0$ and $P_2>0$.

\subsection{$\varepsilon$-Capacity Region}

For $\theta=(h,\nu)\in\Theta$, we write
\begin{equation}\label{eq:C123}
\sfC_1(\theta)\equiv\frac12\log\Bigl(1+\frac{h_1^2P_1}{N_\nu}\Bigr),\quad
\sfC_2(\theta)\equiv\frac12\log\Bigl(1+\frac{h_2^2P_2}{N_\nu}\Bigr),\quad
\sfC_{12}(\theta)\equiv\frac12\log\Bigl(1+\frac{h_1^2P_1+h_2^2P_2}{N_\nu}\Bigr).
\end{equation}
By Theorem~\ref{thm:wyner}, the $0$-capacity region of the component MAC $\bW_\theta=\{W_\theta\}$ is
\begin{equation}\label{eq:Pi}
\Pi_{P_1,P_2,\theta}=\Bigl\{(R_1,R_2)\in\R^2_+\st R_1\le\sfC_1(\theta),\ R_2\le\sfC_2(\theta),\ R_1+R_2\le\sfC_{12}(\theta)\Bigr\}.
\end{equation}
(For $h_1=0$ or $h_2=0$, this is the capacity region of the Gaussian MAC with the corresponding input multiplied by zero.)
We establish the following theorem.

\begin{theorem}\label{thm:qsf}
Fix $\varepsilon\in[0,1)$ arbitrarily. For a quasi-static fading MAC $\bW$ with Gaussian components $\bW_{h,\nu}=
\{W_{h,\nu}\}$ satisfying assumption \textup{(A)}, the $\varepsilon$-capacity region with power constraint $(P_1,P_2)$ is the
$\varepsilon$-outage capacity region
\begin{align}
C_{P_1,P_2}(\varepsilon\,|\,\bW)=\cO(\varepsilon)&\equiv\bigl\{(R_1,R_2)\in\R^2_+\st w\{\theta\,|\,(R_1,R_2)\notin\Pi_{P_1,P_2,\theta}\}\le\varepsilon\bigr\}\notag\\
&=\bigl\{(R_1,R_2)\in\R^2_+\st\Pr\{(R_1,R_2)\notin\Pi_{P_1,P_2,F}\}\le\varepsilon\bigr\}.\label{eq:outage-form}
\end{align}
\end{theorem}

\begin{remark}\label{rem:pentagon}
Theorem~\ref{thm:qsf} states that a rate pair is $\varepsilon$-achievable if and only if the probability that it falls outside the capacity region of the realized
component is at most $\varepsilon$. The region $\cO(\varepsilon)$ is the outage region $\cO_\Gamma(\varepsilon)$ of Part II \cite[Eq.~(46)]{HY-II}, written for the Gaussian
components. The outage event $(R_1,R_2)\notin\Pi_{P_1,P_2,\theta}$ is $\{\sfC_1(\theta)<R_1\mbox{ or }\sfC_2(\theta)<R_2\mbox{ or }\sfC_{12}(\theta)<R_1+R_2\}$, with strict
inequalities. The set $\cO(\varepsilon)$ is closed (see the proof below), so that, unlike \eqref{eq:general}, no closure and no convention on negative rates
are needed in \eqref{eq:outage-form}; the description of the same region in \cite[v2]{HY-arxiv}, as $\Cl_+$ of a set with non-strict outage conditions, coincides
with it.
\end{remark}

\paragraph*{Intuitive picture}
Fix the input to be Gaussian with (almost) full power. Given the realized component $\theta$, the three normalized information
densities of the component concentrate at $\sfC_1(\theta)$, $\sfC_2(\theta)$ and $\sfC_{12}(\theta)$ by the law of large numbers, and, by
Lemma~\ref{lem:bridge}, the information densities of the mixture are, up to a vanishing margin, those of the realized component. Hence
the information spectrum of the mixed MAC is simply the distribution, under $w$, of the random vector $(\sfC_1(F),\sfC_2(F),\sfC_{12}(F))$,
and the spectral outage probability \eqref{eq:J} of the general formula becomes the familiar outage probability. Fig.~\ref{fig:spectrum}
shows this for a single user: the $\varepsilon$-capacity is the point at which the distribution function of $\sfC(F)$ reaches $\varepsilon$. The
converse says that no other input, and no knowledge of $F$ at either terminal (Sec.~\ref{sec:csi}), can do better, because a component in
outage fails with probability tending to one (Lemma~\ref{lem:comp-sc}).

\begin{figure}[htbp]
\centering
\includegraphics[width=0.6\textwidth]{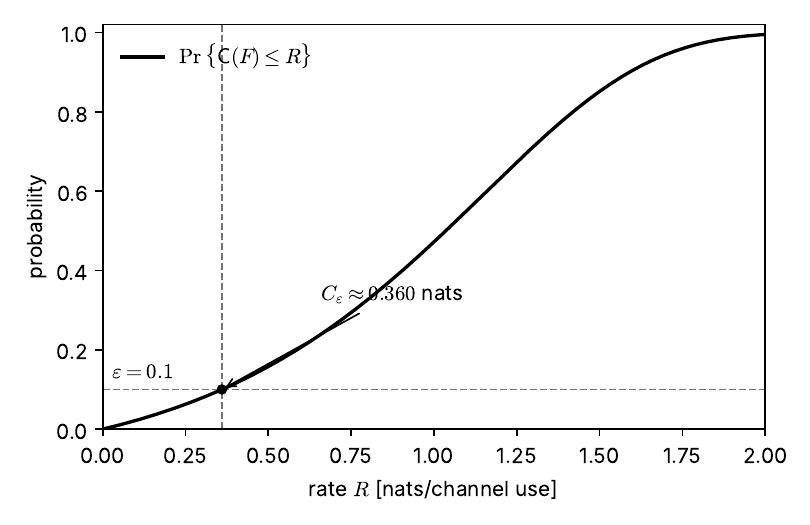}
\caption{The information spectrum of a single-user quasi-static Rayleigh fading channel ($|h|^2$ exponentially distributed with mean
one, $P/N=10$\,dB, $N_\nu\equiv N$). With the full-power Gaussian input, the normalized information density concentrates, for each
realized gain $h$, at $\sfC(h)=\frac12\log(1+h^2P/N)$, so that its limiting distribution is the distribution function of $\sfC(F)$ (solid
curve). The $\varepsilon$-capacity \eqref{eq:single-cap} is the rate $C_\varepsilon$ at which this distribution function reaches $\varepsilon$ (here
$\varepsilon=0.1$); a code at rate $R>C_\varepsilon$ fails with probability tending to one on every realization with $\sfC(h)<R$.}
\label{fig:spectrum}
\end{figure}

As a special case of Theorem~\ref{thm:qsf}, consider the single-user quasi-static fading Gaussian channel, where the
fading coefficient $h_1$ is simply denoted by $h$, i.e., $h\in(-\infty,+\infty)$ and the channel output via $W^n_\theta=W^n_{h,\nu}$ is
given by
\begin{equation}\label{eq:single}
Z^n_\theta=hX^n+V^n_\nu,
\end{equation}
where $X^n$ denotes the channel input. Then, we obtain the following corollary.

\begin{corollary}\label{cor:single}
Fix $\varepsilon\in[0,1)$ arbitrarily. For a single-user quasi-static fading channel $\bW$ with Gaussian
components $\bW_{h,\nu}=\{W_{h,\nu}\}$ satisfying assumption \textup{(A)}, the $\varepsilon$-capacity with power constraint $P$, denoted by $C_P(\varepsilon\,|\,\bW)$, is given by
\begin{align}
C_P(\varepsilon\,|\,\bW)&=\max\bigl\{R\ge0\st w\{\theta\,|\,R\notin\Pi_{P,\theta}\}\le\varepsilon\bigr\}\notag\\
&=\max\bigl\{R\ge0\st\Pr\{\sfC(F)<R\}\le\varepsilon\bigr\},\label{eq:single-cap}
\end{align}
where, for $\theta=(h,\nu)$, $\sfC(\theta)\equiv\frac12\log\bigl(1+\frac{h^2P}{N_\nu}\bigr)$ and $\Pi_{P,\theta}\equiv[0,\sfC(\theta)]$ is the capacity region (interval) of the
component $\bW_\theta$.
\end{corollary}

As in \eqref{eq:outage-form}, the outage event $R\notin\Pi_{P,\theta}$, i.e., $\sfC(\theta)<R$, is strict, and the set over which the maximum is taken
is a closed interval $[0,C_P(\varepsilon\,|\,\bW)]$ (the function $R\mapsto\Pr\{\sfC(F)<R\}$ is left-continuous and nondecreasing), so that the
maximum is attained. The right-hand side of \eqref{eq:single-cap} coincides with the traditional form
\begin{equation}\label{eq:single-cap-trad}
\sup\biggl\{R\st\int_{\bigl\{(h,\nu)\,\big|\,\frac12\log\bigl(1+\frac{h^2P}{N_\nu}\bigr)\le R\bigr\}}dw(h,\nu)\le\varepsilon\biggr\}=\sup\bigl\{R\st\Pr\{\sfC(F)\le R\}\le\varepsilon\bigr\}
\end{equation}
of the $\varepsilon$-outage capacity. Indeed, let $a$ and $b$ denote the right-hand sides of \eqref{eq:single-cap-trad} and \eqref{eq:single-cap}, respectively. If $R<a$, then
$\Pr\{\sfC(F)\le R\}\le\varepsilon$, and letting $R\uparrow a$ gives $\Pr\{\sfC(F)<a\}\le\varepsilon$, i.e., $a\le b$. If $R>a$, then $\Pr\{\sfC(F)\le R'\}>\varepsilon$ for
$a<R'<R$, whence $\Pr\{\sfC(F)<R\}>\varepsilon$, i.e., $b\le a$.

\begin{remark}\label{rem:outage}
The RHS of \eqref{eq:single-cap}, equivalently of \eqref{eq:single-cap-trad}, is the well-known $\varepsilon$-outage capacity. The notion of outage was introduced,
to our knowledge, by Ozarow et al.~\cite{OSW94} and has been widely used for nonergodic fading channels (cf.\ \cite{BPS98}, \cite{EGL10}, \cite{Tel99},
\cite{TVZ04}), while its Shannon-theoretic justification long remained open. Shamai and Wyner \cite{SW97} called the ultimate capacity region of such
channels ``a long-standing open problem'' and Biglieri et al. \cite{BPS98} pointed out that the information spectrum approach of \cite{VH94} may provide a resolution. Caire
et al.~\cite{CTB99} took this approach for block-fading channels with CSIRT, without addressing the strong converse in the sense of Definition~\ref{def:sc},
on which the converse parts of the present paper rest (Remark~\ref{rem:why-sc}); Effros et al.~\cite{EGL10} formalized capacity versus outage for composite
channels with CSI at the receiver, in which outages are recognized by the decoder (Remark~\ref{rem:effros}). Yang et al.~\cite{YDKP14} showed, as a
by-product of their second-order analysis (zero dispersion), that the $\varepsilon$-capacity coincides with the $\varepsilon$-outage capacity of the single-user
quasi-static fading channel with fixed noise variance in all four CSI scenarios, including multiple-antenna channels.

In contrast, the information spectrum proof of Theorem~\ref{thm:qsf} (and hence of Corollary~\ref{cor:single}) addresses the first-order
problem only and is elementary: the direct part follows from Theorem~\ref{thm:general} by an elementary calculus, and the converse part from
Theorem~\ref{thm:sc} applied to each component. Since the channel is nonergodic, the law of large numbers is applied not to the mixture but, after
Lemma~\ref{lem:bridge}, to each realized component. The proof makes transparent how the $\varepsilon$-outage capacity acquires its Shannon-theoretic
meaning and brings to light the role of the strong converse property (Remark~\ref{rem:why-sc}), which is not visible in the finite-blocklength
approach. In Theorem~\ref{thm:qsf} and Corollary~\ref{cor:single} the fading coefficients are available to neither the transmitters nor the receiver.
Following Yang et al.~\cite{YDKP14}, Theorem~\ref{thm:qsf} thus establishes the Shannon-theoretic justification of the $\varepsilon$-outage capacity for the
quasi-static fading Gaussian MAC, with a random noise variance, an arbitrary fading distribution and fixed transmit powers, answering the
question raised in \cite{SW97} in the MAC setting.
\end{remark}

The proof of Theorem~\ref{thm:qsf} rests on two lemmas, one for each direction. The first is the elementary computation with a
Gaussian input; it concerns the component densities \eqref{eq:compdens} only. For $0<\rho<\min\{P_1,P_2\}$, let $\sfC_1^\rho(\theta)$, $\sfC_2^\rho(\theta)$,
$\sfC_{12}^\rho(\theta)$ be defined as in \eqref{eq:C123} with $(P_1,P_2)$ replaced by $(P_1-\rho,P_2-\rho)$.

\begin{lemma}[Gaussian input]\label{lem:gauss-input}
Let $0<\rho<\min\{P_1,P_2\}$, let $X^n$ and $Y^n$ be independent i.i.d.\ sequences with components $\cN(0,P_1-\rho)$ and
$\cN(0,P_2-\rho)$, and let $\overline X^n$ and $\overline Y^n$ be $X^n$ and $Y^n$ conditioned on $\{\|X^n\|^2\le nP_1\}$ and $\{\|Y^n\|^2\le nP_2\}$,
respectively, independently. Then $(\overline\bX,\overline\bY)\in\cG_{\Gamma}$, and, for every $\theta\in\Theta$, the component densities
$\bar\iota^n_{k,\theta}$ $(k=1,2,3)$ of \eqref{eq:compdens} for $(\overline\bX,\overline\bY)$ satisfy
\begin{equation}\label{eq:gauss-input}
\lim_{n\to\infty}\Pr\{\bar\iota^n_{k,\theta}\le b\}=0\qquad\mbox{whenever }b<\sfC^\rho_k(\theta),
\end{equation}
where $\sfC^\rho_3\equiv\sfC^\rho_{12}$. The same holds for $K$ users, with the $2^K-1$ component densities and the rates
$\frac12\log(1+\sum_{k\in\cS}h_k^2(P_k-\rho)/N_\nu)$.
\end{lemma}

\begin{proof}
For the untruncated input, the component densities are normalized sums of i.i.d.\ random variables with finite variances, since
$X^{(n)}_i$, $Y^{(n)}_i$ and $V^{(n)}_{\nu,i}$ are independent Gaussian random variables; for instance,
\[
\log\frac{W^n_\theta(Z^n_\theta|X^n,Y^n)}{P_{Z^n_\theta|Y^n}(Z^n_\theta|Y^n)}=\sum_{i=1}^n\log\frac{W_\theta(Z^{(n)}_{\theta,i}|X^{(n)}_i,Y^{(n)}_i)}{P_{Z^{(n)}_{\theta,i}|Y^{(n)}_i}(Z^{(n)}_{\theta,i}|Y^{(n)}_i)},
\]
with means $I(X;Z_\theta|Y)=\sfC_1^\rho(\theta)$, $I(Y;Z_\theta|X)=\sfC_2^\rho(\theta)$ and $I(XY;Z_\theta)=\sfC_{12}^\rho(\theta)$. By the weak law of large
numbers, $\iota^n_{k,\theta}\to\sfC^\rho_k(\theta)$ in probability. Let $\pi_n^{(1)}\equiv\Pr\{\|X^n\|^2\le nP_1\}$ and $\pi_n^{(2)}\equiv\Pr\{\|Y^n\|^2\le nP_2\}$.
Since $\E[(X_i^{(n)})^2]=P_1-\rho$ and $\V[(X_i^{(n)})^2]=2(P_1-\rho)^2$, Chebyshev's inequality gives
\begin{equation}\label{eq:cheb}
1-\pi_n^{(1)}\le\frac{2(P_1-\rho)^2}{n\rho^2},\qquad 1-\pi_n^{(2)}\le\frac{2(P_2-\rho)^2}{n\rho^2},
\end{equation}
so that $\pi_n^{(\alpha)}\to1$. Let $\overline Z^n_\theta$ be the output of $W^n_\theta$ due to $(\overline X^n,\overline Y^n)$. Restricting the integral over $\bx$ to the ball,
\begin{equation}\label{eq:trunc-out}
P_{Z^n_\theta|Y^n}(\bz|\by)=\int W^n_\theta(\bz|\bx,\by)P_{X^n}(\bx)\,d\bx\ge\pi_n^{(1)}P_{\overline Z^n_\theta|\overline Y^n}(\bz|\by),
\end{equation}
and analogously $P_{Z^n_\theta|X^n}\ge\pi_n^{(2)}P_{\overline Z^n_\theta|\overline X^n}$ and $P_{Z^n_\theta}\ge\pi_n^{(1)}\pi_n^{(2)}P_{\overline Z^n_\theta}$. Hence, at every point of the support of
$(\overline X^n,\overline Y^n,\overline Z^n_\theta)$, $\bar\iota^n_{k,\theta}\ge\iota^n_{k,\theta}-\frac1n\log\frac{1}{\pi_n^{(1)}\pi_n^{(2)}}$. Since moreover the joint density of
$(\overline X^n,\overline Y^n,\overline Z^n_\theta)$ is at most $1/(\pi_n^{(1)}\pi_n^{(2)})$ times that of $(X^n,Y^n,Z^n_\theta)$,
\begin{equation}\label{eq:trunc-bound}
\Pr\{\bar\iota^n_{k,\theta}\le b\}\le\frac{1}{\pi_n^{(1)}\pi_n^{(2)}}\Pr\Bigl\{\iota^n_{k,\theta}\le b+\frac1n\log\frac{1}{\pi_n^{(1)}\pi_n^{(2)}}\Bigr\},
\end{equation}
which tends to zero for $b<\sfC^\rho_k(\theta)$. The $K$-user case is identical.
\end{proof}

The second lemma is the only place where the strong converse of the Gaussian MAC (Theorem~\ref{thm:sc}) enters. In words: a
component in outage does not merely fail sometimes; it fails with probability tending to one.

\begin{lemma}[Components in outage fail]\label{lem:comp-sc}
Let $\theta\in\Theta$ and $(R_1,R_2)\in\R^2_+\setminus\Pi_{P_1,P_2,\theta}$. Every sequence of $(n,M_n^{(1)},M_n^{(2)},\varepsilon_n)$ codes satisfying the power
constraint $(P_1,P_2)$ and \eqref{eq:sc-rate1}, \eqref{eq:sc-rate2} has error probability $\varepsilon_{\theta,n}\to1$ over the component $W_\theta$. The same
holds for $K$ users, with $\Pi_{\bP,\theta}$ of \eqref{eq:PiK} below.
\end{lemma}

For nonzero gains, Lemma~\ref{lem:comp-sc} is Theorem~\ref{thm:sc} (Remark~\ref{rem:kuser-sc} for $K$ users) applied to the Gaussian MAC with powers
$h_k^2P_k$ and noise power $N_\nu$; zero gains require a short additional argument. The proof is given in Appendix~\ref{app:comp-sc}.

\begin{proof}[Proof of Theorem~\ref{thm:qsf}]
\textit{The set $\cO(\varepsilon)$ is closed.} Let $(R_1^{(k)},R_2^{(k)})\in\cO(\varepsilon)$ converge to $(R_1,R_2)$. Since $\Pi_{P_1,P_2,\theta}$ is closed,
$(R_1,R_2)\notin\Pi_{P_1,P_2,\theta}$ implies $(R_1^{(k)},R_2^{(k)})\notin\Pi_{P_1,P_2,\theta}$ for all sufficiently large $k$; thus, by Fatou's lemma,
$w\{\theta\,|\,(R_1,R_2)\notin\Pi_{P_1,P_2,\theta}\}\le\liminf_kw\{\theta\,|\,(R_1^{(k)},R_2^{(k)})\notin\Pi_{P_1,P_2,\theta}\}\le\varepsilon$.

\smallskip\noindent\textit{Direct part: $\cO(\varepsilon)\subseteq C_{P_1,P_2}(\varepsilon\,|\,\bW)$.}
Fix $\gamma>0$, and choose $0<\rho<\min\{P_1,P_2\}$ with $\frac12\log\frac{P_\alpha}{P_\alpha-\rho}\le\gamma$ for $\alpha=1,2$. Since
$\frac{N+aP_1+bP_2}{N+a(P_1-\rho)+b(P_2-\rho)}\le\max_\alpha\frac{P_\alpha}{P_\alpha-\rho}$ for all $a,b\ge0$ and $N>0$, we have, for all $\theta\in\Theta$,
\begin{equation}\label{eq:Crho}
\sfC_1^\rho(\theta)\ge\sfC_1(\theta)-\gamma,\quad\sfC_2^\rho(\theta)\ge\sfC_2(\theta)-\gamma,\quad\sfC_{12}^\rho(\theta)\ge\sfC_{12}(\theta)-\gamma.
\end{equation}
Let $(\overline\bX,\overline\bY)\in\cG_{\Gamma}$ be the input of Lemma~\ref{lem:gauss-input}. We evaluate the information densities \eqref{eq:densities} of
the mixture for this input and reduce them to the component densities by Lemma~\ref{lem:bridge}. To this end we construct, for
$w$-almost every $\theta$, a set $G_{\theta,n}$ and a constant $c_n$ as required there.

\emph{(a) Numerators.} Let $\theta=(h,\nu)$ and $\theta'=(h',\nu')$ satisfy $|h_\alpha-h'_\alpha|\le\delta$ $(\alpha=1,2)$ and $|\nu-\nu'|\le\delta$, and
let $(\bx,\by,\bz)$ satisfy
\begin{equation}\label{eq:typ}
\|\bx\|^2\le nP_1,\quad\|\by\|^2\le nP_2,\quad\|\bv\|^2\le2nN_\nu,\quad\mbox{where }\bv\equiv\bz-h_1\bx-h_2\by.
\end{equation}
Let $g_\nu$ denote the density of $\cN(0,N_\nu)$ and $\bu\equiv(h_1'-h_1)\bx+(h_2'-h_2)\by$, so that
$W^n_{h,\nu}(\bz|\bx,\by)=\prod_{i=1}^ng_\nu(v_i)$ and $W^n_{h',\nu'}(\bz|\bx,\by)=\prod_{i=1}^ng_{\nu'}(v_i-u_i)$. We decompose
\begin{equation}\label{eq:decomp}
\log\frac{W^n_{h,\nu}(\bz|\bx,\by)}{W^n_{h',\nu'}(\bz|\bx,\by)}=\underbrace{\frac n2\log\frac{N_{\nu'}}{N_\nu}+\|\bv\|^2\Bigl(\frac{1}{2N_{\nu'}}-\frac{1}{2N_\nu}\Bigr)}_{=\log\frac{W^n_{h,\nu}}{W^n_{h,\nu'}}}
+\underbrace{\frac{\|\bv-\bu\|^2-\|\bv\|^2}{2N_{\nu'}}}_{=\log\frac{W^n_{h,\nu'}}{W^n_{h',\nu'}}}.
\end{equation}
Let $\tau(\delta)\equiv\sup\{|N_\nu-N_{\nu'}|\,:\,\nu,\nu'\in[0,1],|\nu-\nu'|\le\delta\}$, which tends to zero as $\delta\downarrow0$ by the uniform
continuity of $N_\nu$. If $\tau(\delta)\le N_{\min}/2$, then $|\log(N_{\nu'}/N_\nu)|\le2\tau(\delta)/N_{\min}$, and by \eqref{eq:Nminmax} and \eqref{eq:typ}, the first
term in \eqref{eq:decomp} is bounded in absolute value by $n\tau(\delta)(1/N_{\min}+N_{\max}/N_{\min}^2)$. For the second term, the Cauchy--Schwarz
inequality and $\|\bu\|\le\delta(\|\bx\|+\|\by\|)\le\delta\sqrt n(\sqrt{P_1}+\sqrt{P_2})$, $\|\bv\|\le\sqrt{2nN_{\max}}$ give
\begin{equation}\label{eq:second-term}
\left|\frac{\|\bv-\bu\|^2-\|\bv\|^2}{2N_{\nu'}}\right|\le\frac{2\|\bv\|\|\bu\|+\|\bu\|^2}{2N_{\min}}\le\frac{n}{N_{\min}}\Bigl(\delta(\sqrt{P_1}+\sqrt{P_2})\sqrt{2N_{\max}}+\frac{\delta^2}{2}(\sqrt{P_1}+\sqrt{P_2})^2\Bigr).
\end{equation}
These bounds, divided by $n$, do not depend on $n$, $\theta$, $\theta'$, nor on $(\bx,\by,\bz)$ satisfying \eqref{eq:typ}. Hence we can
choose $\delta=\delta(\gamma)>0$, independently of $n$, so that
\begin{equation}\label{eq:numer}
W^n_{\theta'}(\bz|\bx,\by)\ge e^{-n\gamma}\,W^n_\theta(\bz|\bx,\by)
\end{equation}
whenever $\theta,\theta'$ and $(\bx,\by,\bz)$ are as above. Now partition $\Theta=(-\infty,+\infty)^2\times[0,1]$ into the countably many cells
$\cK_{k,\ell,m}\equiv\Theta\cap([k\delta,(k+1)\delta)\times[\ell\delta,(\ell+1)\delta)\times[m\delta,(m+1)\delta))$, $k,\ell\in\mathbb{Z}$, $m=0,1,2,\ldots$, let $\cK(\theta)$
denote the cell containing $\theta$, and let $\cN_0$ be the union of all cells $\cK$ with $w(\cK)=0$. Then $w(\cN_0)=0$, and for every
$\theta\in\Theta\setminus\cN_0$ we have $w(\cK(\theta))>0$ and, by \eqref{eq:numer},
\begin{equation}\label{eq:cell}
W^n(\bz|\bx,\by)\ge\int_{\cK(\theta)}W^n_{\theta'}(\bz|\bx,\by)\,dw(\theta')\ge w(\cK(\theta))\,e^{-n\gamma}\,W^n_\theta(\bz|\bx,\by)
\end{equation}
for all $(\bx,\by,\bz)$ satisfying \eqref{eq:typ}. That is, the hypothesis of Lemma~\ref{lem:bridge} holds with $G_{\theta,n}$ the set \eqref{eq:typ} and
$c_n\equiv w(\cK(\theta))e^{-n\gamma}$. Since $w(\cK(\theta))$ does not depend on $n$, there is $n_1(\theta)$ with $\frac1n\log\frac1{c_n}\le2\gamma$ for $n\ge n_1(\theta)$.

\emph{(b) Reduction to components.} The input $(\overline\bX,\overline\bY)$ satisfies the first two conditions in \eqref{eq:typ} with probability one, and the
third fails with probability at most $2/n$ by Chebyshev's inequality applied to the i.i.d.\ variables $(V^{(n)}_{\nu,i})^2/N_\nu$ (with
mean 1 and variance 2). Lemma~\ref{lem:bridge} therefore gives, for $\theta\in\Theta\setminus\cN_0$, $n\ge n_1(\theta)$ and every $a\in\R^3$,
\begin{equation}\label{eq:reduction}
\Pr\{E_n(a\,|\,\overline\bX,\overline\bY;\theta)\}\le\Pr\{E^\theta_n(a+3\gamma\bone\,|\,\overline\bX,\overline\bY)\}+\frac2n+3e^{-n\gamma}.
\end{equation}

\emph{(c) Conclusion.} Let $R=(R_1,R_2)\in\cO(\varepsilon)$ and $a\equiv(R_1-5\gamma,R_2-5\gamma,R_1+R_2-10\gamma)$, so that $a+3\gamma\bone=(R_1-2\gamma,R_2-2\gamma,R_1+R_2-7\gamma)$.
By \eqref{eq:mixrep}, the reverse Fatou lemma (the integrands are bounded by 1), \eqref{eq:reduction} and Lemma~\ref{lem:gauss-input}, which
shows that $\Pr\{E^\theta_n(b\,|\,\overline\bX,\overline\bY)\}\to0$ whenever $b_k<\sfC^\rho_k(\theta)$ for $k=1,2,3$,
\begin{align}
J(R_1-5\gamma,R_2-5\gamma\,|\,\overline\bX,\overline\bY)&\le\int_{\Theta\setminus\cN_0}\limsup_{n\to\infty}\Pr\{E^\theta_n(a+3\gamma\bone\,|\,\overline\bX,\overline\bY)\}\,dw(\theta)\notag\\
&\le w\bigl\{\theta\st\sfC_1^\rho(\theta)\le R_1-2\gamma\mbox{ or }\sfC_2^\rho(\theta)\le R_2-2\gamma\mbox{ or }\sfC_{12}^\rho(\theta)\le R_1+R_2-7\gamma\bigr\}\notag\\
&\le w\bigl\{\theta\st\sfC_1(\theta)<R_1\mbox{ or }\sfC_2(\theta)<R_2\mbox{ or }\sfC_{12}(\theta)<R_1+R_2\bigr\}\le\varepsilon,\label{eq:direct-final}
\end{align}
where the third inequality is due to \eqref{eq:Crho}, and the last one is $R\in\cO(\varepsilon)$ (Remark~\ref{rem:pentagon}). Hence
$R-5\gamma\in A_\varepsilon(\overline\bX,\overline\bY)$; note that $R-5\gamma$ may have negative components, which is admitted in \eqref{eq:general}
(Remark~\ref{rem:closure}). Letting $\gamma\downarrow0$ gives $R\in\Cl_+(\bigcup A_\varepsilon(\bX,\bY))=C_{P_1,P_2}(\varepsilon\,|\,\bW)$ by Theorem~\ref{thm:general}.

\smallskip\noindent\textit{Converse part: $C_{P_1,P_2}(\varepsilon\,|\,\bW)\subseteq\cO(\varepsilon)$.}
Let $(R_1,R_2)\in\R^2_+$ be $\varepsilon$-achievable, with codes satisfying \eqref{eq:achievable}, \eqref{eq:power1} and \eqref{eq:power2}, and let $\varepsilon_{\theta,n}$ be the
error probability of these codes over the component $\bW_\theta$, so that $\varepsilon_n=\int_\Theta\varepsilon_{\theta,n}\,dw(\theta)$ by \eqref{eq:avg-error}. By
Lemma~\ref{lem:comp-sc}, $\varepsilon_{\theta,n}\to1$ for every $\theta$ with $(R_1,R_2)\notin\Pi_{P_1,P_2,\theta}$. By the condition $\limsup_n\varepsilon_n\le\varepsilon$ and Fatou's lemma,
\begin{equation}\label{eq:fatou}
\varepsilon\ge\liminf_{n\to\infty}\int_\Theta\varepsilon_{\theta,n}\,dw(\theta)\ge\int_\Theta\liminf_{n\to\infty}\varepsilon_{\theta,n}\,dw(\theta)\ge w\{\theta\,|\,(R_1,R_2)\notin\Pi_{P_1,P_2,\theta}\},
\end{equation}
that is, $(R_1,R_2)\in\cO(\varepsilon)$. Only the pointwise convergence of Lemma~\ref{lem:comp-sc} is used, and no uniformity in $\theta$ is
required. Since $\cO(\varepsilon)$ is closed, it contains the closure $C_{P_1,P_2}(\varepsilon\,|\,\bW)$ of the set of $\varepsilon$-achievable pairs.
\end{proof}

\subsection{$0$-Capacity Region}

We characterize the $0$-capacity region for the quasi-static fading Gaussian MAC under assumption (A). By
setting $\varepsilon=0$ in \eqref{eq:outage-form}, we immediately obtain the following theorem, because $w\{\theta\,|\,(R_1,R_2)\notin\Pi_{P_1,P_2,\theta}\}=0$
holds if and only if $R_1$, $R_2$ and $R_1+R_2$ do not exceed the essential infima of $\sfC_1(F)$, $\sfC_2(F)$ and $\sfC_{12}(F)$,
respectively.

\begin{theorem}[Delay-Limited Capacity Region without Power Control]\label{thm:qsf0}
For a quasi-static fading MAC $\bW$ with Gaussian components
$\bW_{h,\nu}=\{W_{h,\nu}\}$ satisfying assumption \textup{(A)}, the $0$-capacity region with power constraint $(P_1,P_2)$ is given by
\begin{align}
C_{P_1,P_2}(0\,|\,\bW)=\biggl\{(R_1,R_2)\st&0\le R_1\le\essinf\frac12\log\Bigl(1+\frac{h_1^2P_1}{N_\nu}\Bigr),\notag\\
&0\le R_2\le\essinf\frac12\log\Bigl(1+\frac{h_2^2P_2}{N_\nu}\Bigr),\notag\\
&R_1+R_2\le\essinf\frac12\log\Bigl(1+\frac{h_1^2P_1+h_2^2P_2}{N_\nu}\Bigr)\biggr\},\label{eq:qsf0}
\end{align}
where $w$-ess.inf stands for the essential infimum with respect to the measure $w$, i.e., $w$-ess.inf$f(\theta)\equiv
\sup\{a\,|\,w\{\theta\,|\,f(\theta)<a\}=0\}$.
\end{theorem}

This theorem can be specialized to the quasi-static fading Gaussian channel in the single-user case, given
in \eqref{eq:single}, to have Corollary~\ref{cor:single} with $\varepsilon=0$:

\begin{corollary}[Delay-Limited Capacity without Power Control]\label{cor:single0}
For a quasi-static fading channel $\bW$ with Gaussian components
$\bW_{h,\nu}=\{W_{h,\nu}\}$ satisfying assumption \textup{(A)}, the channel $0$-capacity, denoted by $C_P(0\,|\,\bW)$, with power constraint $P$ is given by
\begin{equation}
C_P(0\,|\,\bW)=\essinf\frac12\log\Bigl(1+\frac{h^2P}{N_\nu}\Bigr).
\end{equation}
\end{corollary}

\begin{remark}\label{rem:zero}
Theorem~\ref{thm:qsf0} and Corollary~\ref{cor:single0} show that $C_{P_1,P_2}(0\,|\,\bW)=\{(0,0)\}$ (resp.\ $C_P(0\,|\,\bW)=0$) whenever
$\Pr\{|h_\alpha|<\delta\}>0$ for every $\delta>0$ and $\alpha=1,2$ (resp.\ $\Pr\{|h|<\delta\}>0$ for every $\delta>0$), because $N_\nu\le N_{\max}$ under
assumption (A); only the marginal distributions of $h_1$ and $h_2$ matter, since $R_1\le0$ and $R_2\le0$ already force $(R_1,R_2)=(0,0)$. This is the
case, e.g., for Rayleigh, Rician or Nakagami-$m$ distributions of $|h_1|$, $|h_2|$ and $|h|$ (cf.\ Goldsmith \cite{Gol05}). In the single-user case it is the
familiar fact that the zero-outage capacity of such fading channels without power control is zero; Theorem~\ref{thm:qsf0} gives the corresponding
statement for the MAC as a coding theorem. The delay-limited capacity region was introduced by Hanly and Tse \cite{HT98b} for the fading MAC with
power control and fixed noise variance; Theorem~\ref{thm:qsf0} is its counterpart without power control, with a random noise variance, obtained
as the case $\varepsilon=0$ of Theorem~\ref{thm:qsf}.
\end{remark}

%% ====================================================================
\section{Quasi-Static Fading Gaussian MAC with CSI}\label{sec:csi}

In this section, for the quasi-static fading Gaussian MAC, we consider the following four scenarios depending
on availability of channel state information (CSI) $\theta$, which is taken by $F$ with probability $P_F(\theta)=w(\theta)$:
\begin{enumerate}[leftmargin=2.2em,label=(\roman*)]
\item channel $\bW_{\rm rx}=\{W^n_{\rm rx}\}_{n=1}^\infty$ with CSI $F=\theta$ available at the receiver (CSIR) is specified by\footnote{As far as there is no fear of confusion, when we write a channel as $W^n(U^n|V^n)$ it denotes the conditional probability of $U^n$
given $V^n$. Specifically, $W^n(\bz,\theta|\bx,\by)$ denotes the conditional probability of $(Z^n,F)=(\bz,\theta)$ given $(X^n,Y^n)=(\bx,\by)$.}
\begin{equation}\label{eq:Wrx}
W^n_{\rm rx}(\bz,\theta|\bx,\by)=W^n(\bz,\theta|\bx,\by),
\end{equation}
whose $\varepsilon$-capacity region is denoted by $C^{\rm rx}_{P_1,P_2}(\varepsilon\,|\,\bW)$,
\item channel $\bW_{\rm tx}=\{W^n_{\rm tx}\}_{n=1}^\infty$ with CSI $F=\theta$ available at the transmitter (CSIT) is specified by
\begin{equation}\label{eq:Wtx}
W^n_{\rm tx}(\bz|\bx,\by,\theta)=W^n(\bz|\bx,\by,\theta),
\end{equation}
whose $\varepsilon$-capacity region is denoted by $C^{\rm tx}_{P_1,P_2}(\varepsilon\,|\,\bW)$,
\item channel $\bW_{\rm rt}=\{W^n_{\rm rt}\}_{n=1}^\infty$ with CSI $F=\theta$ available at both the receiver and the transmitter (CSIRT)
is specified by
\begin{equation}\label{eq:Wrt}
W^n_{\rm rt}(\bz,\theta|\bx,\by,\theta)=W^n(\bz,\theta|\bx,\by,\theta),
\end{equation}
whose $\varepsilon$-capacity region is denoted by $C^{\rm rt}_{P_1,P_2}(\varepsilon\,|\,\bW)$,
\item channel $\bW_{\rm no}=\{W^n_{\rm no}\}_{n=1}^\infty$ with CSI unavailable (no-CSI) is specified by
\begin{equation}\label{eq:Wno}
W^n_{\rm no}(\bz|\bx,\by)=W^n(\bz|\bx,\by),
\end{equation}
whose $\varepsilon$-capacity region is denoted by $C^{\rm no}_{P_1,P_2}(\varepsilon\,|\,\bW)$.
\end{enumerate}
The objective of this section is to establish the formulas for the $\varepsilon$-capacity regions and investigate the
relationship among them.

With CSI at the transmitters (CSIT or CSIRT), a family of $(n,M_n^{(1)},M_n^{(2)},\varepsilon_n)$ MAC codes $\bigl\{C^{(1)}_{\theta,n}\times C^{(2)}_{\theta,n}\bigr\}_{\theta\in\Theta}$
with $|C^{(\alpha)}_{\theta,n}|=M_n^{(\alpha)}$ for all $\theta$ is arranged, and $C^{(1)}_{\theta,n}\times C^{(2)}_{\theta,n}$ is used when $F=\theta$ is observed at the encoders;
let $\varepsilon_{\theta,n}$ denote its error probability over the component $W^n_\theta$. With CSI at the receiver (CSIR and CSIRT), the decoding regions,
denoted $D_{\theta,jk}$, may also depend on $\theta$; in the CSIT scenario they do not. As is implicit in the definition below, the codes are chosen so
that $\theta\mapsto\varepsilon_{\theta,n}$ is measurable. Then, the average error probability is given by
\begin{equation}\label{eq:avg-error-csi}
\varepsilon_n=\int_{\Theta}\varepsilon_{\theta,n}\,dw(\theta).
\end{equation}
A rate pair $(R_1,R_2)$ is said to be $\varepsilon$-achievable if there exists a set of $(n,M_n^{(1)},M_n^{(2)},\varepsilon_n)$ MAC codes $\bigl\{C^{(1)}_{\theta,n}\times
C^{(2)}_{\theta,n}\bigr\}_{\theta\in\Theta}$ satisfying \eqref{eq:achievable}. When $F=\theta$, codewords denoted by $X^n_\theta$ and $Y^n_\theta$ are uniformly distributed on $C^{(1)}_{\theta,n}$
and $C^{(2)}_{\theta,n}$, respectively. Therefore, the inputs $X^n_\theta$ and $Y^n_\theta$ to the channel are conditionally independent when
CSI $F=\theta$ is fixed. Thus, the inputs are denoted by $(X^n_F,Y^n_F)$ when the dependence on $F$ is emphasized.

By definition, we have
\begin{align}
C^{\rm no}_{P_1,P_2}(\varepsilon\,|\,\bW)\subseteq C^{\rm rx}_{P_1,P_2}(\varepsilon\,|\,\bW)\subseteq C^{\rm rt}_{P_1,P_2}(\varepsilon\,|\,\bW),\label{eq:incl1}\\
C^{\rm no}_{P_1,P_2}(\varepsilon\,|\,\bW)\subseteq C^{\rm tx}_{P_1,P_2}(\varepsilon\,|\,\bW)\subseteq C^{\rm rt}_{P_1,P_2}(\varepsilon\,|\,\bW)\label{eq:incl2}
\end{align}
for all $\varepsilon\in[0,1)$.

In this section, we show the following theorem:

\begin{theorem}\label{thm:csi}
Fix $\varepsilon\in[0,1)$ arbitrarily. For a quasi-static fading MAC $\bW$ with stationary memoryless
Gaussian components $\bW_\theta=\{W_\theta\}$ satisfying assumption \textup{(A)}, the $\varepsilon$-capacity regions with power constraint $(P_1,P_2)$ satisfy
\begin{equation}\label{eq:csi}
C^{\rm no}_{P_1,P_2}(\varepsilon\,|\,\bW)=C^{\rm rx}_{P_1,P_2}(\varepsilon\,|\,\bW)=C^{\rm tx}_{P_1,P_2}(\varepsilon\,|\,\bW)=C^{\rm rt}_{P_1,P_2}(\varepsilon\,|\,\bW).
\end{equation}
\end{theorem}

\begin{remark}\label{rem:yang}
For one user with fixed noise variance, the coincidence follows from Yang et al.~\cite{YDKP14}, who showed that the maximal coding rate is
$C_\varepsilon+O(\frac{\log n}{n})$ in all four scenarios (also for multiple antennas); that CSIR does not change the $\varepsilon$-capacity was observed in
\cite[p.~2632]{BPS98}. Theorem~\ref{thm:csi} extends the first-order part to the two-user MAC with a random noise variance.
\end{remark}

\begin{proof}[Proof of Theorem~\ref{thm:csi}]
By Theorem~\ref{thm:qsf}, $C^{\rm no}_{P_1,P_2}(\varepsilon\,|\,\bW)=\cO(\varepsilon)$. The inclusions \eqref{eq:incl1} and \eqref{eq:incl2} imply that it suffices to show that
\begin{equation}\label{eq:csi-goal}
C^{\rm rt}_{P_1,P_2}(\varepsilon\,|\,\bW)\subseteq\cO(\varepsilon).
\end{equation}
Let $(R_1,R_2)\in\R^2_+$ be $\varepsilon$-achievable in the scenario of CSIRT. Then, there exists a set of $(n,M_n^{(1)},M_n^{(2)},\varepsilon_n)$
MAC codes $\bigl\{C^{(1)}_{\theta,n}\times C^{(2)}_{\theta,n}\bigr\}_{\theta\in\Theta}$ with decoding regions $\{D_{\theta,jk}\}$ satisfying \eqref{eq:achievable} and the power constraint, and
$\varepsilon_n=\int_{\Theta}\varepsilon_{\theta,n}\,dw(\theta)$. For each fixed $\theta$, the code $C^{(1)}_{\theta,n}\times C^{(2)}_{\theta,n}$ with decoding regions $\{D_{\theta,jk}\}$ is an ordinary code
for the component $\bW_\theta$, with rates $\frac1n\log M_n^{(\alpha)}$ not depending on $\theta$. Hence Lemma~\ref{lem:comp-sc} applies to each component separately:
$\varepsilon_{\theta,n}\to1$ for every $\theta$ with $(R_1,R_2)\notin\Pi_{P_1,P_2,\theta}$, and Fatou's lemma gives
\begin{equation}
\varepsilon\ge\liminf_{n\to\infty}\int_{\Theta}\varepsilon_{\theta,n}\,dw(\theta)\ge w\{\theta\,|\,(R_1,R_2)\notin\Pi_{P_1,P_2,\theta}\},
\end{equation}
i.e., $(R_1,R_2)\in\cO(\varepsilon)$. Since $\cO(\varepsilon)$ is closed (see the proof of Theorem~\ref{thm:qsf}), \eqref{eq:csi-goal} follows. In short, the no-CSI
direct part and the CSIRT converse part sandwich all four regions.
\end{proof}

\begin{remark}\label{rem:alt}
Once the no-CSI direct part (Theorem~\ref{thm:qsf}) is established, the strong converse property of each component alone yields the converse part in
all four scenarios; the no-CSI region attains this outer bound because one input (Gaussian with full power) is asymptotically optimal for all
components. For general components this fails, and the regions with and without CSIT may differ (Sec.~\ref{sec:general-comp}). A proof of
\eqref{eq:csi-goal} via Lemma~\ref{lem:csirt-conv} is also possible, but requires an additional evaluation of the information spectra of the components.
\end{remark}

%% ====================================================================
\section{$K$-User Quasi-Static Fading Gaussian MAC}\label{sec:kuser}

It is straightforward to extend Theorems~\ref{thm:qsf} and \ref{thm:csi} for the two-user quasi-static fading Gaussian MACs
discussed in Sec.~\ref{sec:qsf} to the case of $K$ users.

Let $M_n^{(k)}$ and $X_k^n=\bigl(X_{k,1}^{(n)},X_{k,2}^{(n)},\ldots,X_{k,n}^{(n)}\bigr)$ be the uniformly distributed message and the corresponding
codeword at encoder $k\in[1:K]=\{1,\ldots,K\}$. We assume that each codeword satisfies
\begin{equation}
\Pr\left\{\frac1n\sum_{i=1}^n\bigl(X^{(n)}_{k,i}\bigr)^2\le P_k\right\}=1\quad\mbox{for }k=1,2,\ldots,K,
\end{equation}
which is referred to as \emph{power constraint} $\bP=(P_1,P_2,\ldots,P_K)$.

The tuple of fading coefficients is denoted by $h=(h_1,h_2,\ldots,h_K)$ which is randomly generated prior to
encoding. As in Sec.~\ref{sec:qsf}, the set of CSIs is denoted by $\Theta$. When $\theta=(h,\nu)\in\Theta$, the output of component
channel $\bW_\theta=\{W^n_\theta\}_{n=1}^\infty$ is given by
\begin{equation}
Z^n_\theta=h_1X^n_1+h_2X^n_2+\cdots+h_KX^n_K+V^n_\nu,
\end{equation}
where $V^n_\nu=\bigl(V^{(n)}_{\nu,1},V^{(n)}_{\nu,2},\ldots,V^{(n)}_{\nu,n}\bigr)$ is an additive Gaussian noise sequence with $V^{(n)}_{\nu,i}\sim\cN(0,N_\nu)$. The
channel law of the quasi-static fading Gaussian MAC is characterized as a mixed MAC with probability
measure $w$ on $\Theta$ (as was discussed in Sec.~\ref{sec:qsf}) and is denoted by $\bW=\{W^n\}_{n=1}^\infty$. As in Sec.~\ref{sec:csi}, the $\varepsilon$-capacity
regions with power constraint $\bP$ are defined analogously and denoted by $C^{\rm no}_{\bP}(\varepsilon\,|\,\bW)$, $C^{\rm rx}_{\bP}(\varepsilon\,|\,\bW)$,
$C^{\rm tx}_{\bP}(\varepsilon\,|\,\bW)$, and $C^{\rm rt}_{\bP}(\varepsilon\,|\,\bW)$ depending on availability of CSI $\theta\in\Theta$.

Let $R_k$ denote the nonnegative rate for user $k\in[1:K]$. For a nonempty subset $\cS\subseteq[1:K]$, let $R_{\cS}$ denote
the sum of rates
\begin{equation}
R_{\cS}=\sum_{k\in\cS}R_k,
\end{equation}
and the rate tuple is denoted by $\bR=(R_1,\ldots,R_K)$. For $\theta=(h,\nu)\in\Theta$, the $0$-capacity region with power
constraint $\bP=(P_1,P_2,\ldots,P_K)$ with component Gaussian MAC $\bW_\theta=\{W_\theta\}$ is given by
\begin{equation}\label{eq:PiK}
\Pi_{\bP,\theta}=\biggl\{\bR\in\R^K_+\st R_{\cS}\le\frac12\log\Bigl(1+\frac{\sum_{k\in\cS}h_k^2P_k}{N_\nu}\Bigr)\mbox{ for all nonempty }\cS\subseteq[1:K]\biggr\}.
\end{equation}
Let $F$ denote the random variable corresponding to the CSI taking values in $\Theta$.

The direct part below uses the $K$-user version of Theorem~\ref{thm:general}, which we state for completeness.

\begin{remark}[$K$-user general formula]\label{rem:kuser-general}
Let $\bW=\{W^n\}_{n=1}^\infty$ be a general $K$-user MAC with input alphabets $\cX_1,\ldots,\cX_K$, cost functions $c_k$ and cost constraint
$\bGamma=(\Gamma_1,\ldots,\Gamma_K)$, and let $\cG_{\bGamma}$ be the set of inputs $\bX_{[1:K]}=(\bX_1,\ldots,\bX_K)$ with $X^n_1,\ldots,X^n_K$ independent and
$\Pr\{\frac1nc_k(X^n_k)\le\Gamma_k\}=1$ for all $k$ and $n$. For a nonempty $\cS\subseteq[1:K]$, write $X^n_{\cS}\equiv(X^n_k)_{k\in\cS}$, $\cS^c\equiv[1:K]\setminus\cS$, and
\begin{equation}
\iota^n_{\cS}\equiv\frac1n\log\frac{W^n(Z^n|X^n_{[1:K]})}{P_{Z^n|X^n_{\cS^c}}(Z^n|X^n_{\cS^c})},\qquad
J(\bR\,|\,\bX_{[1:K]})\equiv\limsup_{n\to\infty}\Pr\Bigl\{\bigcup_{\cS\ne\emptyset}\{\iota^n_{\cS}\le R_{\cS}\}\Bigr\}
\end{equation}
for $\bR\in\R^K$. Then, for $0\le\varepsilon<1$,
\begin{equation}
C_{\bGamma}(\varepsilon\,|\,\bW)=\Cl\Bigl(\bigcup_{\bX_{[1:K]}\in\cG_{\bGamma}}\bigl\{\bR\in\R^K\st J(\bR\,|\,\bX_{[1:K]})\le\varepsilon\bigr\}\Bigr)\cap\R^K_+.
\end{equation}
The event for $\cS$ corresponds to decoding the messages of the users in $\cS$ wrongly and those of the others correctly; for $K=2$ the sets
$\{1\}$, $\{2\}$ and $\{1,2\}$ give the three events of \eqref{eq:En}. The proof is that of Theorem~\ref{thm:general} with the $2^K-1$ densities $\iota^n_{\cS}$; the
$K$-user versions of \cite[Lemmas~3 and 4]{Han98} that it requires are \cite[Lemmas~1 and 2]{HY-III}. The random-coding lemma \cite[Lemma~1]{HY-III} states: for
independent inputs satisfying the cost constraint, $M_k$ codewords for user $k$, $\cA\equiv\{k:M_k\ge2\}$ and $\gamma>0$, there is a code, with the threshold
decoder that requires the inequalities for all nonempty $\cS\subseteq\cA$, whose average error probability (as averaged over the random code in its proof)
is at most $\Pr\{\bigcup_{\emptyset\ne\cS\subseteq\cA}\{\iota^n_{\cS}\le\frac1n\log\prod_{k\in\cS}M_k+\gamma\}\}+(2^K-1)e^{-n\gamma}$.
\end{remark}

We establish the following theorem for which the proof is similar to the one for Theorems~\ref{thm:qsf} and \ref{thm:csi} (see also
Remark~\ref{rem:kuser-sc}).

\begin{theorem}\label{thm:kuser}
Fix $\varepsilon\in[0,1)$ arbitrarily. For a $K$-user quasi-static fading Gaussian MAC $\bW$ satisfying
assumption \textup{(A)}, the $\varepsilon$-capacity regions with power constraint $\bP=(P_1,P_2,\ldots,P_K)$ are given by
\begin{align}
C^{\rm no}_{\bP}(\varepsilon\,|\,\bW)=C^{\rm rx}_{\bP}(\varepsilon\,|\,\bW)&=C^{\rm tx}_{\bP}(\varepsilon\,|\,\bW)=C^{\rm rt}_{\bP}(\varepsilon\,|\,\bW)\notag\\
&=\bigl\{\bR\in\R^K_+\,|\,\Pr\{\bR\notin\Pi_{\bP,F}\}\le\varepsilon\bigr\}.\label{eq:kuser}
\end{align}
\end{theorem}

\emph{What changes with $K$ users.} %Only the bookkeeping. 
The pentagon becomes the polymatroid $\Pi_{\bP,\theta}$, with one constraint for each
nonempty set $\cS$ of users whose messages can be confused, and a rate tuple is in outage when the realized component cannot carry the sum rate
$R_\cS$ of some $\cS$. In the direct part, each of the $2^K-1$ densities of the mixture is reduced to that of the realized component (Lemma~\ref{lem:bridge}),
one event at a time; in the converse part, a component in outage fails by the strong converse, whatever the number of users. Hence the proof is
that of the two-user case.

\begin{proof}
The proof is that of Theorems~\ref{thm:qsf} and \ref{thm:csi} with $2^K-1$ densities in place of three. Converse part (CSIRT, hence all
four scenarios): Lemma~\ref{lem:comp-sc} in its $K$-user form (Appendix~\ref{app:comp-sc}), the identity \eqref{eq:avg-error} and Fatou's lemma show
that every $\varepsilon$-achievable $\bR$ satisfies $\Pr\{\bR\notin\Pi_{\bP,F}\}\le\varepsilon$; the right-hand side of \eqref{eq:kuser} is closed by the two-user argument.
Direct part (no-CSI): the $K$-user general formula (Remark~\ref{rem:kuser-general}), Lemma~\ref{lem:gauss-input} in its $K$-user form, and
Lemma~\ref{lem:bridge} with $2^K-1$ denominators (each handled by Lemma~\ref{lem:basic}) and the same cells $\cK(\theta)$, so that the shift becomes
$(2^K-1)e^{-n\gamma}$ in place of $3e^{-n\gamma}$.
\end{proof}

Now, as a special case of Theorem~\ref{thm:kuser}, we can characterize the $0$-capacity region as follows: let
\begin{equation}\label{eq:GPeta}
G_{\bP,\theta}(\cS)\equiv\frac12\log\Bigl(1+\frac{\sum_{k\in\cS}h_k^2P_k}{N_\nu}\Bigr)
\end{equation}
for $\theta=(h,\nu)\in\Theta$. Then, we have

\begin{corollary}[Delay-Limited Capacity Region without Power Control]\label{cor:kuser0}
For a $K$-user quasi-static fading Gaussian MAC $\bW$ satisfying assumption \textup{(A)}, the
$0$-capacity regions with power constraint $\bP=(P_1,P_2,\ldots,P_K)$ are given by
\begin{align}
C^{\rm no}_{\bP}(0\,|\,\bW)=C^{\rm rx}_{\bP}(0\,|\,\bW)&=C^{\rm tx}_{\bP}(0\,|\,\bW)=C^{\rm rt}_{\bP}(0\,|\,\bW)\notag\\
&=\Bigl\{\bR\in\R^K_+\st R_{\cS}\le\essinf G_{\bP,F}(\cS)\mbox{ for all nonempty }\cS\subseteq[1:K]\Bigr\}.\label{eq:kuser0}
\end{align}
\end{corollary}

\begin{remark}\label{rem:polymatroid}
It is well known that $\Pi_{\bP,\theta}$ of \eqref{eq:PiK} is a polymatroid. For constant $N_\nu$, Tse and Hanly \cite{TH98} characterized, by the theory of
polymatroids, the \emph{throughput capacity} region of the fading Gaussian MAC with CSI at the transmitters and the receiver, a kind of
\emph{ergodic capacity} region (cf.\ also \cite{Gal94}, \cite{SW97}), defined under ergodic fading and hence not directly applicable to the quasi-static
scenario. For quasi-static fading MACs, the outage event $\{\bR\notin\Pi_{\bP,F}\}$ in \eqref{eq:kuser} was used by Tse, Viswanath and Zheng \cite{TVZ04} for the
diversity-multiplexing tradeoff, and outage capacity regions with CSIRT and power allocation were studied by Li, Jindal and Goldsmith \cite{LJG05}
(see also \cite{Nar07}). Theorem~\ref{thm:kuser} gives an operational ($\varepsilon$-capacity) characterization of the region defined by this outage event,
for fixed power and in all four CSI scenarios.

The polymatroid (or contra-polymatroid) structure in information theory was found earlier by Fujishige \cite{Fuj78} and Han
\cite{Han79a}, \cite{Han79b}, and developed by Han \cite{Han80} and Han and Kobayashi \cite{HK80}, all originating from Edmonds \cite{Edm70}.
For example, the maximum of a linear function $\sum_{k=1}^K\mu_kR_k$ over a polymatroid is attained at an extreme point computed by the greedy
algorithm of Edmonds \cite{Edm70} (cf.\ \cite{TH98} for fading MACs); we do not pursue this here.
\end{remark}

%% ====================================================================
\section{Variable-Rate Coding}\label{sec:vr}

Throughout Secs.~\ref{sec:qsf}--\ref{sec:kuser}, we have confined ourselves to \emph{fixed-rate coding}. On the other hand, if \emph{variable-rate
coding} is also allowed, we may enlarge the $0$-capacity region \eqref{eq:kuser0}. (Operational notions of variable-rate capacity for
nonergodic channels were formulated by Verd\'u and Shamai \cite{VS10} for single-user memoryless channels with finitely many
states, with emphasis on the case without CSI at the transmitter; here we consider the CSIRT scenario for the $K$-user
quasi-static fading Gaussian MAC with a general state space.) This is indeed possible, for example, in the
case of CSIRT, because we may vary the \emph{coding rate} depending on the parameter $\theta\in\Theta$; the rate tuple used for
the component $\theta$ is denoted by $\bR_\theta=(R_{1,\theta},R_{2,\theta},\ldots,R_{K,\theta})$, and the overall coding rate is its average
\begin{equation}\label{eq:avgrate}
\E\bR_F=\int_{\Theta}\bR_\theta\,dw(\theta).
\end{equation}
Instead of \eqref{eq:GPeta}, let us define
\begin{equation}\label{eq:GP}
G_{\bP}(\cS)\equiv\frac12\E\log\Bigl(1+\frac{\sum_{k\in\cS}F_k^2P_k}{N_{\tilde\nu}}\Bigr)=\int_{\Theta}G_{\bP,\theta}(\cS)\,dw(\theta),
\end{equation}
where $F=(F_1,\ldots,F_K,\tilde\nu)$ is the realized component, $\tilde\nu$ being its noise index.

\paragraph*{Variable-rate coding with CSIRT}
A \emph{variable-rate CSIRT code} of blocklength $n$ consists, for each $\theta\in\Theta$,
of a $K$-user MAC code for the component $W^n_\theta$ with $M_{k,\theta,n}$ codewords for user $k$, all satisfying the power
constraint $\bP$, and of decoding regions that may depend on $\theta$; both the encoders and the decoder know $\theta$. As
in Sec.~\ref{sec:csi}, $\varepsilon_{\theta,n}$ denotes the average error probability of the code for $\theta$ over $W^n_\theta$, the codes are chosen so that
$\theta\mapsto\varepsilon_{\theta,n}$ is measurable, and $\varepsilon_n=\int_{\Theta}\varepsilon_{\theta,n}\,dw(\theta)$. A rate tuple $\bR\in\R^K_+$ is said to be \emph{$0$-achievable by variable-rate
coding with CSIRT} if there exist a measurable map $\theta\mapsto\bR_\theta\in\R^K_+$ and a sequence of variable-rate
CSIRT codes such that
\begin{equation}
\int_{\Theta}\bR_\theta\,dw(\theta)\ge\bR\ \ \mbox{(cf.\ \eqref{eq:avgrate})},\qquad\liminf_{n\to\infty}\frac1n\log M_{k,\theta,n}\ge R_{k,\theta}\ \ (k\in[1:K],\ w\mbox{-a.e. }\theta),\qquad\lim_{n\to\infty}\varepsilon_n=0,
\end{equation}
where the first inequality is componentwise. The closure of the set of all such $\bR$ is denoted by $\tilde C^{\rm rt}_{\bP}(0\,|\,\bW)$.
With $\bR_\theta\equiv\bR$ for all $\theta$, this reduces to fixed-rate coding, so that $\tilde C^{\rm rt}_{\bP}(0\,|\,\bW)\supseteq C^{\rm rt}_{\bP}(0\,|\,\bW)$.

We obtain the $0$-capacity region with variable-rate coding in the CSIRT scenario as follows:

\begin{theorem}[CSIRT-Variable-Rate Coding]\label{thm:vr}
For a $K$-user quasi-static fading Gaussian MAC $\bW$ satisfying
assumption \textup{(A)} and $G_{\bP}([1:K])<\infty$,
\begin{equation}\label{eq:vr}
\tilde C^{\rm rt}_{\bP}(0\,|\,\bW)=\Bigl\{\bR\in\R^K_+\st R_{\cS}\le G_{\bP}(\cS)\mbox{ for all nonempty }\cS\subseteq[1:K]\Bigr\},
\end{equation}
which exactly coincides with the \emph{ergodic capacity} (throughput) region of the fading Gaussian MAC with fixed transmit powers (see \cite{BPS98}, \cite{EGL10}, \cite{Gal94}, \cite{SW97}, \cite{TH98}) when the fading
coefficients (CSIRT) are assumed to be ergodic, but the ergodicity is not assumed in \eqref{eq:vr} (see Remark~\ref{rem:why-ergodic}).
\end{theorem}

The direct part rests on the following property of polymatroids. For a set function $f$ on the subsets of
$[1:K]$, let $\cP(f)\equiv\{\bR\in\R^K_+\,|\,R_{\cS}\le f(\cS)\mbox{ for all nonempty }\cS\subseteq[1:K]\}$, so that $\Pi_{\bP,\theta}=\cP(G_{\bP,\theta})$ by \eqref{eq:PiK}
and the RHS of \eqref{eq:vr} is $\cP(G_{\bP})$.

For ergodic fading, this property is due to Tse and Hanly \cite{TH98}. In the technical-report version \cite{TH96}, Lemma~3.6 states that the
region whose rank function is the expectation of the per-state rank functions is a polymatroid (proved there ``by direct verification''), and
Lemma~3.8 that it is the set of expectations of per-state rate allocations, each vertex being the expectation of the per-state vertices for the same
permutation. We give a complete proof for an arbitrary probability measure $w$ on $\Theta$, including the submodularity of the averaged rank function,
the measurability of the selection, and the points off the dominant face.

\begin{lemma}[Averaging of polymatroids; cf.\ {\cite{TH98}}, {\cite[Lemma~3.8]{TH96}}]\label{lem:avgpoly}
Assume $G_{\bP}([1:K])<\infty$. Then, for every $\bR\in\cP(G_{\bP})$, there
exists a measurable map $\theta\mapsto\bR_\theta$ such that $\bR_\theta\in\Pi_{\bP,\theta}$ for every $\theta\in\Theta$ and $\int_{\Theta}\bR_\theta\,dw(\theta)=\bR$.
\end{lemma}

\begin{proof}
Write $f_\theta\equiv G_{\bP,\theta}$ and $F\equiv G_{\bP}$. For each $\theta$, $f_\theta(\emptyset)=0$, $f_\theta$ is nondecreasing, and $f_\theta$ is submodular: with
$a(\cS)\equiv\sum_{k\in\cS}h_k^2P_k/N_\nu$, which is modular and nonnegative, and the concave function $g(t)\equiv\frac12\log(1+t)$, we
have $a(\cS\cap\cT)\le a(\cS),a(\cT)\le a(\cS\cup\cT)$ and $a(\cS\cup\cT)+a(\cS\cap\cT)=a(\cS)+a(\cT)$, whence $g(a(\cS\cup\cT))+g(a(\cS\cap
\cT))\le g(a(\cS))+g(a(\cT))$ by concavity. Since $0\le f_\theta(\cS)\le f_\theta([1:K])$ and $F([1:K])<\infty$, the function
$F=\int f_\theta\,dw(\theta)$ is finite, and it has the same three properties. Thus $\cP(f_\theta)$ and $\cP(F)$ are polymatroids \cite{Edm70}.

For a permutation $\pi$ of $[1:K]$ and a set function $f$, let $\bv^\pi(f)$ be the greedy vector with components
$v^\pi_{\pi(i)}(f)\equiv f(\{\pi(1),\ldots,\pi(i)\})-f(\{\pi(1),\ldots,\pi(i-1)\})$, $i=1,\ldots,K$. It is linear in $f$, so that $\int\bv^\pi(f_\theta)\,dw(\theta)=
\bv^\pi(F)$. By Edmonds \cite{Edm70}, for a polymatroid $\cP(f)$: (a) every maximal vector of $\cP(f)$ (with respect to the
componentwise order) has component sum $f([1:K])$; and (b) the set of vectors of $\cP(f)$ with component
sum $f([1:K])$ is the convex hull of $\{\bv^\pi(f)\}_\pi$.

\textit{Points on the dominant face.} Let $\bR\in\cP(F)$. Since $\cP(F)$ is compact, there is a $\by\in\cP(F)$ with $\by\ge\bR$ that maximizes the component
sum over $\{\by'\in\cP(F)\,|\,\by'\ge\bR\}$; this $\by$ is maximal in $\cP(F)$, so that, by (a) and (b), $\by=\sum_\pi\lambda_\pi\bv^\pi(F)$
for some convex weights $\lambda_\pi$. Put $\by_\theta\equiv\sum_\pi\lambda_\pi\bv^\pi(f_\theta)$. By (b) applied to $f_\theta$, $\by_\theta\in\cP(f_\theta)=\Pi_{\bP,\theta}$; moreover
$\by_\theta\ge\mathbf 0$, $\theta\mapsto\by_\theta$ is measurable, and $\int\by_\theta\,dw(\theta)=\by$ by linearity. This is the argument of \cite[Lemma~3.8]{TH96}.

\textit{Points off the dominant face.} Let $R_{k,\theta}\equiv(R_k/y_k)y_{k,\theta}$ if $y_k>0$
and $R_{k,\theta}\equiv0$ if $y_k=0$ (in which case $R_k=0$). Then $\mathbf 0\le\bR_\theta\le\by_\theta$, so that $\bR_\theta\in\Pi_{\bP,\theta}$, and $\int R_{k,\theta}\,dw(\theta)=
(R_k/y_k)y_k=R_k$.
\end{proof}

\begin{proof}[Proof of Theorem~\ref{thm:vr}]
\textit{Converse part.} Let $\bR$ be $0$-achievable by variable-rate coding with CSIRT, with the
map $\theta\mapsto\bR_\theta$ and the codes of the definition. For each fixed $\theta$, the code for $\theta$ is an ordinary $K$-user code
for the component $W^n_\theta$ with rates satisfying $\liminf_n\frac1n\log M_{k,\theta,n}\ge R_{k,\theta}$. Hence the $K$-user strong converse
(Remark~\ref{rem:kuser-sc}), with components having zero gains handled as in the proof of Theorem~\ref{thm:kuser}, gives $\varepsilon_{\theta,n}\to1$
for every $\theta$ with $\bR_\theta\notin\Pi_{\bP,\theta}$, and Fatou's lemma gives $w\{\theta\,|\,\bR_\theta\notin\Pi_{\bP,\theta}\}\le\liminf_{n\to\infty}\varepsilon_n=0$. (Since $\varepsilon_n\to0$, the weak converse
would suffice here instead, as noted in Remark~\ref{rem:why-sc}.) This is Lemma~\ref{lem:comp-sc} applied to $\bR_\theta$. Hence
$R_{\cS,\theta}\le G_{\bP,\theta}(\cS)$ for $w$-almost all $\theta$, and
\begin{equation}
R_{\cS}\le\int_{\Theta}R_{\cS,\theta}\,dw(\theta)\le G_{\bP}(\cS)
\end{equation}
for all nonempty $\cS$. Since the RHS of \eqref{eq:vr} is closed, it contains $\tilde C^{\rm rt}_{\bP}(0\,|\,\bW)$.

\textit{Direct part.} Let $\bR$ belong to the RHS of \eqref{eq:vr}, and fix $\delta>0$. By Lemma~\ref{lem:avgpoly}, there is a measurable $\theta\mapsto
\bR_\theta\in\Pi_{\bP,\theta}$ with $\int\bR_\theta\,dw(\theta)=\bR$. Quantize it as $r_k(\theta)\equiv\delta\lfloor R_{k,\theta}/\delta\rfloor-\delta$ if $R_{k,\theta}\ge2\delta$ and $r_k(\theta)\equiv0$ otherwise.
Then $\br(\theta)=(r_1(\theta),\ldots,r_K(\theta))$ takes countably many values, $\theta\mapsto\br(\theta)$ is measurable, $r_k(\theta)\ge R_{k,\theta}-2\delta$,
and $r_k(\theta)\le R_{k,\theta}-\delta$ whenever $r_k(\theta)>0$. Consequently, for every nonempty $\cS$ with $r_{\cS}(\theta)>0$,
\begin{equation}
r_{\cS}(\theta)\le R_{\cS,\theta}-\delta\le G_{\bP,\theta}(\cS)-\delta.
\end{equation}
Choose $0<\rho<\min_kP_k$ with $\frac12\log(P_k/(P_k-\rho))\le\delta/2$ for all $k$. As in \eqref{eq:Crho}, $G_{\bP,\theta}(\cS)\le
G^\rho_{\bP,\theta}(\cS)+\delta/2$ for all $\theta$, where $G^\rho_{\bP,\theta}$ is $G_{\bP,\theta}$ with $P_k$ replaced by $P_k-\rho$; hence $r_{\cS}(\theta)\le G^\rho_{\bP,\theta}(\cS)-\delta/2$ whenever
$r_{\cS}(\theta)>0$.

For each value $\br$ of $\br(\cdot)$ and each $n$, consider the random codebooks in which user $k$ with $r_k>0$ has
$\lceil e^{nr_k}\rceil$ codewords drawn independently from the truncated Gaussian distribution $P_{\overline{X}^n_k}$ of Lemma~\ref{lem:gauss-input}
(with variance $P_k-\rho$ before truncation), and user $k$ with $r_k=0$ has the single codeword
$\mathbf 0$. These codebooks do not depend on $\theta$ beyond the value $\br=\br(\theta)$, which is known to the encoders.
Given $\theta$ with $\br(\theta)=\br$, the decoder, which knows $\theta$, uses the threshold decoder of \cite[Lemma~1]{HY-III} for the
component $W^n_\theta$ (Remark~\ref{rem:kuser-general}) with thresholds $\frac1n\log\prod_{k\in\cS}\lceil e^{nr_k}\rceil+\delta/8$
for the nonempty subsets $\cS$ of $\{k\,|\,r_k>0\}$. Let $\overline\varepsilon_{\theta,n}$ denote the resulting ensemble average error probability
over $W^n_\theta$. By the proof of that lemma, $\overline\varepsilon_{\theta,n}$ is at most the probability that one of the component information densities
of $W^n_\theta$ for the truncated inputs does not exceed its threshold, plus $(2^K-1)e^{-n\delta/8}$. Since $\lceil e^{nr_k}\rceil\le2e^{nr_k}$, for all large $n$ each
threshold is at most $r_{\cS}+\delta/4\le G^\rho_{\bP,\theta}(\cS)-\delta/4$, which lies strictly below the mean $G^\rho_{\bP,\theta}(\cS)$ of the corresponding
normalized information density for the untruncated i.i.d.\ Gaussian inputs. By Lemma~\ref{lem:gauss-input} (in its $K$-user form, applied to the component $W^n_\theta$), $\overline\varepsilon_{\theta,n}\to0$ for every $\theta\in\Theta$.

Let $\cB_{\br}\equiv\{\theta\,|\,\br(\theta)=\br\}$. Averaging over the random codebooks (Fubini's theorem), there is, for each $\br$
and $n$, a realization of the codebooks such that the resulting code satisfies $\int_{\cB_{\br}}\varepsilon_{\theta,n}\,dw(\theta)\le\int_{\cB_{\br}}\overline\varepsilon_{\theta,n}\,dw(\theta)$;
every codeword satisfies the power constraint. Using this code whenever $\theta\in\cB_{\br}$, we obtain a variable-rate
CSIRT code with
\begin{equation}
\varepsilon_n=\sum_{\br}\int_{\cB_{\br}}\varepsilon_{\theta,n}\,dw(\theta)\le\int_{\Theta}\overline\varepsilon_{\theta,n}\,dw(\theta)\to0\quad(n\to\infty)
\end{equation}
by dominated convergence, and with $\frac1n\log M_{k,\theta,n}\ge r_k(\theta)$. Since $\int r_k(\theta)\,dw(\theta)\ge R_k-2\delta$, the tuple
$(\max\{R_k-2\delta,0\})_{k=1}^K$ is $0$-achievable by variable-rate coding with CSIRT. Letting $\delta\downarrow0$ and using the
closedness of $\tilde C^{\rm rt}_{\bP}(0\,|\,\bW)$ completes the proof.
\end{proof}

The RHS of \eqref{eq:vr} is the polymatroid $\cP(G_{\bP})$, which appeared in \cite{BPS98}, \cite{EGL10}, \cite{Gal94}, \cite{SW97}, \cite{TH98}. Since
$w$-ess.inf$\,G_{\bP,F}(\cS)\le G_{\bP}(\cS)$, it contains the region \eqref{eq:kuser0}. Optimal power control on this polymatroid is treated in \cite{TH98}.

\begin{remark}[Why the ergodic capacity region reappears]\label{rem:why-ergodic}
The coincidence in Theorem~\ref{thm:vr} reflects the fact that two averaging operations commute on polymatroids. In the ergodic setting with CSI at the
receiver, the state $F_i$ varies within a codeword, and the normalized information densities of full-power Gaussian inputs concentrate around
$\frac1n\sum_{i=1}^nG_{\bP,F_i}(\cS)\to G_{\bP}(\cS)$: the \emph{rank function} itself is averaged within one codeword, and $\cP(G_{\bP})$ is attained by
fixed-rate coding. In the quasi-static setting the state is fixed during a codeword, and variable-rate coding averages instead the \emph{rate tuples}
over the realizations of $F$. The two regions coincide because
\begin{equation}\label{eq:commute}
\Bigl\{\int_{\Theta}\bR_\theta\,dw(\theta)\st\bR_\theta\in\cP(f_\theta)\ \mbox{for all }\theta\Bigr\}=\cP\Bigl(\int_{\Theta}f_\theta\,dw(\theta)\Bigr)\qquad(f_\theta=G_{\bP,\theta}),
\end{equation}
where ``$\subseteq$'' is the inequality used in the converse part of Theorem~\ref{thm:vr} and ``$\supseteq$'' is Lemma~\ref{lem:avgpoly}; the successive-decoding
order optimal for a given weighted sum rate is the same for every $\theta$. For ergodic fading the identity is stated in \cite{TH98}, \cite[Lemma~3.8]{TH96},
where it is not needed for achievability; in Theorem~\ref{thm:vr} it is the achievability argument, for a state fixed during the whole codeword. Unlike
\cite{TH98}, \eqref{eq:vr} is for constant power allocation, and only the average $\E\bR_F$ equals $\bR$.
\end{remark}

%% ====================================================================
\section{Mixed MACs with General Components and CSI}\label{sec:general-comp}

In this section, we consider a mixed MAC \eqref{eq:mixed} whose components $\bW_\theta=\{W^n_\theta\}_{n=1}^\infty$ ($\theta\in\Theta$) are general MACs, not
necessarily stationary, ergodic or memoryless, with arbitrary alphabets; $\theta\mapsto W^n_\theta(\cdot|\bx,\by)$ is measurable, and densities are
understood as in Sec.~\ref{sec:basic}. The cost functions $c_1,c_2$ and the cost constraint $\Gamma=(\Gamma_1,\Gamma_2)$ are arbitrary but fixed (in
Secs.~\ref{sec:sc}--\ref{sec:vr}, $\Gamma$ was the power constraint); the set of inputs is $\cG_{\Gamma}$, the $\varepsilon$-capacity regions carry the subscript
$\Gamma$, and Theorem~\ref{thm:general} applies. A channel without cost constraint is the case $c_1\equiv c_2\equiv0$, and a single-user channel the case
$|\cY|=1$; in both cases we keep the subscript $\Gamma$. We study how knowledge of the realized component at either terminal affects the
$\varepsilon$-capacity region; for memoryless Gaussian components all four regions coincide (Theorems~\ref{thm:csi} and \ref{thm:kuser}).

\subsection{CSI Scenarios}\label{sec:gen-csi}

The component $\theta$ is drawn once according to $w$ and held fixed for the whole block. We denote by $F$ the realized
component, a $\Theta$-valued random variable with law $w$. As in Sec.~\ref{sec:csi}, besides the setting of Sec.~\ref{sec:general} (\emph{no-CSI}), we consider three
scenarios.
\begin{itemize}
\item \textbf{CSIR:} the decoder observes $\theta$, so the decoding regions, denoted $D_{\theta,jk}$, may depend on $\theta$. Equivalently, the code is
used over the MAC $\bW_{\rm rx}$ with output $(Z^n,F)$, whose transition probability is $W^n_{\rm rx}(\bz,\theta|\bx,\by)=w(\theta)W^n_\theta(\bz|\bx,\by)$ for
countable $\Theta$ (a density with respect to the product of the reference measure and $w$ in general).
\item \textbf{CSIT:} the encoders observe $\theta$, so the codewords, denoted $\bu_{\theta,j}$ and $\bv_{\theta,k}$, may depend on $\theta$, each satisfying the cost
constraint; the decoder does not observe $\theta$.
\item \textbf{CSIRT:} both encoders and the decoder observe $\theta$.
\end{itemize}
These are the scenarios of Sec.~\ref{sec:csi}, now with general components and an arbitrary cost constraint. In every scenario the codebook
sizes $M_n^{(1)}$ and $M_n^{(2)}$ do not depend on $\theta$, and all $\theta$-dependent maps are measurable. Let $\varepsilon_{\theta,n}$ be the average error
probability of the code used for $\theta$ over $W^n_\theta$. As in \eqref{eq:avg-error-csi}, the average error probability is
\begin{equation}\label{eq:gen-avgerr}
\varepsilon_n=\int_{\Theta}\varepsilon_{\theta,n}\,dw(\theta).
\end{equation}
$\varepsilon$-achievability is defined as in Definition~\ref{def:achievable}. The $\varepsilon$-capacity regions with cost constraint $(\Gamma_1,\Gamma_2)$ are denoted
$C^s_{\Gamma}(\varepsilon\,|\,\bW)$ for $s\in\{{\rm no},{\rm rx},{\rm tx},{\rm rt}\}$, with $C^{\rm no}_{\Gamma}(\varepsilon\,|\,\bW)=C_{\Gamma}(\varepsilon\,|\,\bW)$ of Definition~\ref{def:achievable}. As in \eqref{eq:incl1} and \eqref{eq:incl2},
\begin{equation}\label{eq:gen-incl}
C^{\rm no}_{\Gamma}\subseteq C^{\rm rx}_{\Gamma}\subseteq C^{\rm rt}_{\Gamma},\qquad C^{\rm no}_{\Gamma}\subseteq C^{\rm tx}_{\Gamma}\subseteq C^{\rm rt}_{\Gamma}.
\end{equation}
With CSI at the encoders, the channel inputs used for component $\theta$, denoted $X^n_\theta$ and $Y^n_\theta$, are independent and uniform on the codebooks
for $\theta$.

\subsection{The Main Result}\label{sec:gen-main}

Recall from Sec.~\ref{sec:mixed} that, for a general input $(\bX,\bY)$ and $\theta\in\Theta$, $Z^n_\theta$ denotes the output of $W^n_\theta$, $\iota^n_{k,\theta}$ $(k=1,2,3)$
the component densities \eqref{eq:compdens}, and $E^\theta_n(a\,|\,\bX,\bY)$ the event \eqref{eq:Eneta}.

A \emph{CSIT input} is a family $\{(\bX_\theta,\bY_\theta)\}_{\theta\in\Theta}$ of general inputs $(\bX_\theta,\bY_\theta)\in\cG_{\Gamma}$ such that $\theta\mapsto P_{X^n_\theta}$ and $\theta\mapsto P_{Y^n_\theta}$
are measurable (automatic if $\Theta$ is countable); we abbreviate it as $(\bX_{\Theta},\bY_{\Theta})\equiv\{(\bX_\theta,\bY_\theta)\}_{\theta\in\Theta}$ and denote the set of all CSIT inputs by $\cG^{\Theta}_{\Gamma}$. A CSIT input is thus an ordinary general input for each component separately; no random variable is involved in it.
For $a\in\R^3$ define
\begin{equation}\label{eq:Jtilde}
\tilde J_{\bW}(a\,|\,\bX_{\Theta},\bY_{\Theta})\equiv\limsup_{n\to\infty}\int_{\Theta}\Pr\{E^\theta_n(a\,|\,\bX_\theta,\bY_\theta)\}\,dw(\theta),
\end{equation}
which is nondecreasing in each $a_k$, and write $\tilde J_{\bW}(R\,|\,\bX_{\Theta},\bY_{\Theta})\equiv\tilde J_{\bW}(R_1,R_2,R_1+R_2\,|\,\bX_{\Theta},\bY_{\Theta})$ for $R\in\R^2$. Let
\begin{equation}\label{eq:Btilde}
\tilde A_\varepsilon(\bX_{\Theta},\bY_{\Theta})\equiv\bigl\{R\in\R^2:\tilde J_{\bW}(R\,|\,\bX_{\Theta},\bY_{\Theta})\le\varepsilon\bigr\},\qquad\tilde B_{\Gamma}(\varepsilon)\equiv\Cl_+\Bigl(\bigcup_{(\bX_{\Theta},\bY_{\Theta})\in\cG^{\Theta}_{\Gamma}}\tilde A_\varepsilon(\bX_{\Theta},\bY_{\Theta})\Bigr).
\end{equation}
$\tilde J_{\bW}$ is the spectral outage probability when each component is evaluated with its own input and its own densities, averaged over
$w$; neither \eqref{eq:Jtilde} nor \eqref{eq:Btilde} refers to the realized component $F$, which enters only where it is an input or output of a channel
(the CSIR channel $\bW_{\rm rx}$, the CSIT decoder of Lemma~\ref{lem:csit}, and Lemma~\ref{lem:noexp}).

\begin{theorem}\label{thm:gencsi}
Let $\Theta$ be countable and $(\Gamma_1,\Gamma_2)$ arbitrary. For a mixed MAC $\bW$ with general components and every $0\le\varepsilon<1$,
\begin{align}
C^{\rm no}_{\Gamma}(\varepsilon\,|\,\bW)&=C^{\rm rx}_{\Gamma}(\varepsilon\,|\,\bW),\label{eq:norx}\\
C^{\rm tx}_{\Gamma}(\varepsilon\,|\,\bW)&=C^{\rm rt}_{\Gamma}(\varepsilon\,|\,\bW)=\tilde B_{\Gamma}(\varepsilon).\label{eq:Btilde-formula}
\end{align}
\end{theorem}

\begin{remark}\label{rem:gencsi-meaning}
For a countable parameter space, the $\varepsilon$-capacity regions coincide in the no-CSI and CSIR scenarios and in the CSIT and CSIRT scenarios, even
though the components are completely general: whether the region is enlarged depends solely on CSI at the \emph{transmitters}. For memoryless
Gaussian components the two pairs coincide as well (Theorem~\ref{thm:csi}), because a single input is asymptotically optimal for all components; in
general they do not (Example~\ref{ex:zchannel}). Countability cannot simply be dropped in \eqref{eq:norx} or in the first equality $C^{\rm tx}_{\Gamma}=C^{\rm rt}_{\Gamma}$
of \eqref{eq:Btilde-formula} (Remark~\ref{rem:countable} and Example~\ref{ex:uncountable}), while the second equality $C^{\rm rt}_{\Gamma}=\tilde B_{\Gamma}$, the information
spectrum formula for the CSIRT scenario, remains valid for uncountable $\Theta$ under mild measurability conditions (Proposition~\ref{prop:uncountable}).
Thus, for uncountable $\Theta$, \eqref{eq:Btilde-formula} as a whole does not hold in general; only its second equality does.
\end{remark}

\begin{remark}[Capacity versus outage]\label{rem:effros}
For a single-user mixed channel, the proof of \eqref{eq:norx} below shows that the $\varepsilon$-capacity
with cost constraint $\Gamma$ in the scenario of CSIR is
\begin{equation}\label{eq:rx-single}
C^{\rm rx}_{\Gamma}(\varepsilon\,|\,\bW)=\sup_{\bX\in\cG_{\Gamma}}\sup\{R:J_{\bW_{\rm rx}}(R\,|\,\bX)\le\varepsilon\},\quad
J_{\bW_{\rm rx}}(R\,|\,\bX)=\limsup_{n\to\infty}\Pr\Bigl\{\frac1n\log\frac{W^n(Z^n| X^n,F)}{P_{Z^n|F}(Z^n|F)}\le R\Bigr\},
\end{equation}
where $\cG_{\Gamma}$ is, for $|\cY|=1$, the set of general inputs $\bX=\{X^n\}$ with $\Pr\{\frac1nc_1(X^n)\le\Gamma_1\}=1$ and $W^n(\cdot|\bx,\theta)\equiv W^n_\theta(\cdot|\bx)$. Equation \eqref{eq:norx}
implies that this equals $C^{\rm no}_{\Gamma}(\varepsilon\,|\,\bW)=\sup_{\bX\in\cG_{\Gamma}}\sup\{R:J_{\bW}(R\,|\,\bX)\le\varepsilon\}$ with $J_{\bW}(R\,|\,\bX)\equiv\limsup_n\Pr\{\frac1n\log\frac{W^n(Z^n|X^n)}{P_{Z^n}(Z^n)}\le R\}$
when $\Theta$ is countable. Effros et al.~\cite{EGL10} introduced the \emph{capacity versus outage} for single-user composite channels with CSIR, where the decoder may
declare an outage whose probability must be asymptotically at most $\varepsilon$, the error probability conditioned on non-outage vanishing; outages are
thus recognized by the decoder, whereas in the definition of the $\varepsilon$-capacity decoding errors may remain undetected. Their formula has the form
of the right-hand side of \eqref{eq:rx-single}. Hence, for countable $\Theta$, the capacity versus outage with CSIR, the $\varepsilon$-capacity with CSIR and the
$\varepsilon$-capacity without CSI coincide in value, although the first is operationally different; for uncountable $\Theta$ the last coincidence may fail
(Remark~\ref{rem:countable}).
\end{remark}

\subsection{Proof of \eqref{eq:norx}}

\emph{Step 1: CSIR densities.} Let $(\bX,\bY)$ be a general input; it does not depend on $\theta$ in either the no-CSI or the
CSIR scenario. The joint law of $(F,X^n,Y^n,Z^n)$ has density $w(\theta)P_{X^n}(\bx)P_{Y^n}(\by)W^n_\theta(\bz|\bx,\by)$ for countable $\Theta$. Hence
\begin{equation}\label{eq:rx-ident}
\begin{gathered}
W^n_{\rm rx}(\bz,\theta|\bx,\by)=w(\theta)W^n_\theta(\bz|\bx,\by),\qquad P_{Z^nF|Y^n}(\bz,\theta|\by)=w(\theta)P_{Z^n_\theta|Y^n}(\bz|\by),\\
P_{Z^nF|X^n}(\bz,\theta|\bx)=w(\theta)P_{Z^n_\theta|X^n}(\bz|\bx),\qquad P_{Z^nF}(\bz,\theta)=w(\theta)P_{Z^n_\theta}(\bz).
\end{gathered}
\end{equation}
The factor $w(\theta)$ cancels in each of the three information densities \eqref{eq:densities} of $\bW_{\rm rx}$, which therefore equal the component
densities of the realized component: given $F=\theta$, they are $\iota^n_{k,\theta}$. Thus
\begin{equation}\label{eq:Jrx}
J_{\bW_{\rm rx}}(R\,|\,\bX,\bY)=\limsup_{n\to\infty}\int_{\Theta}\Pr\{E^\theta_n(R_1,R_2,R_1+R_2\,|\,\bX,\bY)\}\,dw(\theta).
\end{equation}
(The same identities hold for an arbitrary $(\Theta,w)$, with densities with respect to the product of the reference measure
and $w$.) By Theorem~\ref{thm:general} applied to $\bW$ and to $\bW_{\rm rx}$, whose inputs range over the same set $\cG_{\Gamma}$,
\[
C^{\rm no}_{\Gamma}(\varepsilon\,|\,\bW)=\Cl_+\Bigl(\bigcup_{(\bX,\bY)\in\cG_{\Gamma}}A_\varepsilon(\bX,\bY)\Bigr),\qquad C^{\rm rx}_{\Gamma}(\varepsilon\,|\,\bW)=\Cl_+\Bigl(\bigcup_{(\bX,\bY)\in\cG_{\Gamma}}A^{\rm rx}_\varepsilon(\bX,\bY)\Bigr),
\]
where $A^{\rm rx}_\varepsilon(\bX,\bY)\equiv\{R\in\R^2:J_{\bW_{\rm rx}}(R\,|\,\bX,\bY)\le\varepsilon\}$.

\emph{Step 2: Mixture versus component densities.} Let $E_n(a\,|\,\bX,\bY;\theta)$ be the event \eqref{eq:En} formed with the mixture densities
\eqref{eq:densities} evaluated at $(X^n,Y^n,Z^n_\theta)$, as in Sec.~\ref{sec:mixed}, so that \eqref{eq:mixrep} holds. Let $\theta\in\Theta$ with $w(\theta)>0$, $\gamma>0$, and
$\gamma_\theta\equiv\gamma+\frac1n\log\frac{1}{w(\theta)}$. Then, for every input and every $a\in\R^3$,
\begin{align}
\Pr\{E_n(a\,|\,\bX,\bY;\theta)\}&\le\Pr\{E^\theta_n(a+\gamma_\theta\bone\,|\,\bX,\bY)\}+3e^{-n\gamma},\label{eq:mix-to-comp}\\
\Pr\{E^\theta_n(a\,|\,\bX,\bY)\}&\le\Pr\{E_n(a+\gamma_\theta\bone\,|\,\bX,\bY;\theta)\}+e^{-n\gamma}.\label{eq:comp-to-mix}
\end{align}
Inequality \eqref{eq:mix-to-comp} is Lemma~\ref{lem:bridge} with $c_n=w(\theta)$ and $G_{\theta,n}$ the whole space, since $W^n\ge w(\theta)W^n_\theta$
pointwise for countable $\Theta$ (instance (i) of Remark~\ref{rem:bridge}). For \eqref{eq:comp-to-mix}: Lemma~\ref{lem:basic} with $T=(X^n,Y^n)$, $U=Z^n_\theta$
and $Q=W^n(\cdot| X^n,Y^n)$ shows that $W^n(Z^n_\theta| X^n,Y^n)\le e^{n\gamma}W^n_\theta(Z^n_\theta| X^n,Y^n)$ with probability at least $1-e^{-n\gamma}$.
Since the denominators satisfy $P_{Z^n|Y^n}\ge w(\theta)P_{Z^n_\theta|Y^n}$, $P_{Z^n|X^n}\ge w(\theta)P_{Z^n_\theta|X^n}$ and $P_{Z^n}\ge w(\theta)P_{Z^n_\theta}$, on this event
$\iota^n_k(X^n,Y^n,Z^n_\theta)\le\iota^n_{k,\theta}+\gamma_\theta$ for $k=1,2,3$, which gives \eqref{eq:comp-to-mix}.

\emph{Step 3: Averaging.} Let $\Theta_n\equiv\{\theta:w(\theta)\ge1/n\}$. For $n$ with $\frac1n\log n\le\gamma$, $\gamma_\theta\le2\gamma$ on $\Theta_n$. Integrating \eqref{eq:mix-to-comp}
and \eqref{eq:comp-to-mix} over $\Theta_n$, and bounding the integrals over $\Theta\setminus\Theta_n$ by $w(\Theta\setminus\Theta_n)$, we obtain from \eqref{eq:mixrep} and the monotonicity
of the events
\begin{align*}
\Pr\{E_n(a\,|\,\bX,\bY)\}&\le\int_{\Theta}\Pr\{E^\theta_n(a+2\gamma\bone\,|\,\bX,\bY)\}\,dw(\theta)+w(\Theta\setminus\Theta_n)+3e^{-n\gamma},\\
\int_{\Theta}\Pr\{E^\theta_n(a\,|\,\bX,\bY)\}\,dw(\theta)&\le\Pr\{E_n(a+2\gamma\bone\,|\,\bX,\bY)\}+w(\Theta\setminus\Theta_n)+e^{-n\gamma}.
\end{align*}
Here $w(\Theta\setminus\Theta_n)=\sum_\theta w(\theta)\ind\{w(\theta)<1/n\}\to0$ by dominated convergence. Taking $\limsup_n$ with $a=(R_1-2\gamma,R_2-2\gamma,
R_1+R_2-4\gamma)$, for which $a+2\gamma\bone\le(R_1,R_2,R_1+R_2)$ componentwise, and using \eqref{eq:Jrx}, we obtain
\begin{equation}\label{eq:J-compare}
J_{\bW}(R-2\gamma\,|\,\bX,\bY)\le J_{\bW_{\rm rx}}(R\,|\,\bX,\bY),\qquad J_{\bW_{\rm rx}}(R-2\gamma\,|\,\bX,\bY)\le J_{\bW}(R\,|\,\bX,\bY)
\end{equation}
for every $R\in\R^2$. Hence $R\in A^{\rm rx}_\varepsilon(\bX,\bY)$ implies $R-2\gamma\in A_\varepsilon(\bX,\bY)$, and vice versa. Since $\gamma>0$ is arbitrary, each of
the two unions in Step 1 is contained in the closure of the other, and \eqref{eq:norx} follows.\hfill$\Box$

\subsection{A CSIRT Converse Lemma and CSIT Coding Lemmas}\label{sec:gen-lemmas}

The three lemmas of this subsection are the tools for the proof of Theorem~\ref{thm:gencsi}. In two-user terms:
Lemma~\ref{lem:csirt-conv} is a converse bound valid in all four scenarios: the error probability is at least, up to a vanishing term, the
$w$-average of the probability that one of the three component densities falls below $R_1$, $R_2$ or $R_1+R_2$. Lemma~\ref{lem:csit} is the matching
random-coding bound with CSIT only: each $\theta$ has its own codebooks, and the decoder, not knowing $\theta$, searches jointly over the message pair and
a list of $L_n$ candidate components. Within the correct component, a wrong message of user~1, of user~2, or of both gives the events with $\log M_1$,
$\log M_2$, $\log(M_1M_2)$; a wrong component, whose codebooks are independent of those of $\theta$, acts like an error of both messages and is controlled
by the sum-rate event. The list costs $\frac1n\log L_n$ in every threshold, which vanishes for $L_n=\lfloor e^{\sqrt n}\rfloor$.

The events are indexed below by a family $\cL_{\cA}$, where $\cA$ is the set of users with positive rate; it is needed only for rate pairs on
the axes. If $R_1,R_2>0$, then $\cA=\{1,2\}$ and $\cL_{\cA}$ consists of $\{1\}$, $\{2\}$ and $\{1,2\}$. If, say, $R_2=0$, then $\cA=\{1\}$: the event
for $\{2\}$ is omitted, since no error corresponds to it ($M_2=1$) and its threshold $4\gamma$ has no rate margin (it need not be small even in a
component that supports $(R_1,0)$, e.g., one whose output does not depend on the input of user~2); the event for $\{1,2\}$ is kept, with the
threshold $R_1+4\gamma$, because with CSI at the transmitters even the single codeword of user~2 depends on $\theta$, so that a wrong component
changes the codewords of both users. Such pairs cannot be approximated by pairs with both rates positive: if components of the kind just
described have total weight more than $\varepsilon$, the region of Theorem~\ref{thm:gencsi} may contain $(R_1,0)$ but no pair with $R_2>0$.
Lemma~\ref{lem:compform}, finally, replaces the mixture densities of the decoder by the component densities, so that the bound becomes a
$w$-average of per-component events, which is what $\tilde J_{\bW}$ measures.

The lemmas are stated for $K$ users, because they are also used for $K$-user mixed MACs in Part II \cite[Sec.~VII]{HY-II};
Theorem~\ref{thm:gencsi} uses them with $K=2$, and with general $K$ in part (d) of the proof of \eqref{eq:Btilde-formula}. For $K$ users, the event for a
nonempty $\cS$ corresponds to a wrong message exactly of the users in $\cS$, and a wrong component to $\cS=[1:K]$. If all rates are positive, $\cL_{\cA}$
is the family of all nonempty $\cS\subseteq[1:K]$; otherwise, the sets not meeting $\cA$ are omitted for the reason just given, the other
$\cS\not\subseteq\cA$, $\cS\ne[1:K]$, correspond to no error and are omitted as well, and $[1:K]$ is kept, so that every set in $\cL_{\cA}$ meets $\cA$
and has a positive rate. The case $K=2$ is spelled out at the end of this subsection. The components are general $K$-user MACs
$W^n_\theta:\cX_1^n\times\cdots\times\cX_K^n\to\cZ^n$, with cost functions $c_k$ and cost constraint $\bGamma=(\Gamma_1,\ldots,\Gamma_K)$, and we use the
notation of Remark~\ref{rem:kuser-general}: for a nonempty $\cS\subseteq[1:K]$, $\cS^c\equiv[1:K]\setminus\cS$, $\bx_\cS\equiv(\bx_k)_{k\in\cS}$ and
$M_\cS\equiv\prod_{k\in\cS}M_k$; a factor or a variable indexed by the empty set is omitted. A CSIRT (resp.\ CSIT) code is defined as in
Sec.~\ref{sec:gen-csi}, with codewords $\bu_{k,\theta}(j)$, $j\in[1:M_k]$, for user $k$ and decoding regions $D_{\theta,\bj}$ (resp.\ $D_{\bj}$),
$\bj=(j_1,\ldots,j_K)$. For inputs $X^n_{[1:K],\theta}\equiv(X^n_{1,\theta},\ldots,X^n_{K,\theta})$ with independent components and the output $Z^n_\theta$ of
$W^n_\theta$, the component densities are
\[
\iota^n_{\cS,\theta}\equiv\frac1n\log\frac{W^n_\theta(Z^n_\theta|X^n_{[1:K],\theta})}{P_{Z^n_\theta|X^n_{\cS^c,\theta}}(Z^n_\theta|X^n_{\cS^c,\theta})}\qquad(\emptyset\ne\cS\subseteq[1:K]),
\]
where for $\cS=[1:K]$ the denominator is $P_{Z^n_\theta}(Z^n_\theta)$.

\begin{lemma}[CSIRT converse]\label{lem:csirt-conv}
For each nonempty $\cS\subseteq[1:K]$, let $\{Q^n_{\theta,\cS}\}_{\theta\in\Theta}$ be a family of conditional probability densities
$Q^n_{\theta,\cS}(\cdot|\bx_{\cS^c})$ on $\cZ^n$ (unconditional for $\cS=[1:K]$), measurable in $(\theta,\bx_{\cS^c})$. For any CSIRT code with message sets of
sizes $M_1,\ldots,M_K$, satisfying the measurability requirements of Sec.~\ref{sec:gen-csi}, and every $\gamma>0$,
\begin{equation}\label{eq:csirt-conv}
\varepsilon_n\ge\int_{\Theta}\Pr\Bigl\{\bigcup_{\emptyset\ne\cS\subseteq[1:K]}\Bigl\{\frac1n\log\frac{W^n_\theta(Z^n_\theta|X^n_{[1:K],\theta})}{Q^n_{\theta,\cS}(Z^n_\theta|X^n_{\cS^c,\theta})}
\le\frac1n\log M_\cS-\gamma\Bigr\}\Bigr\}\,dw(\theta)-(2^K-1)e^{-n\gamma},
\end{equation}
where, for each $\theta$, $X^n_{k,\theta}=\bu_{k,\theta}(J_k)$ with $J_1,\ldots,J_K$ independent and $J_k$ uniform on $[1:M_k]$, and $Z^n_\theta$ is the output of
$W^n_\theta$.
\end{lemma}

The proof is given in Appendix~\ref{app:gen}.

\begin{remark}\label{rem:csirt-general}
Lemma~\ref{lem:csirt-conv} holds for arbitrary components, alphabets and probability measures $w$, countable or not. Since the CSIT, CSIR and no-CSI
scenarios are special cases of the CSIRT scenario, it applies to all four scenarios. In the single-user no-CSI scenario it reduces to
\cite[Lemma 7]{YHN16}, which extends the bound of Verd\'u and Han \cite{Han98}, \cite{VH94} (see also Hayashi and Nagaoka \cite{HN03}) to mixed channels.
\end{remark}

The direct part rests on a CSIT random-coding lemma. Let $\Theta$ be countable. For each $\theta$ let $X^n_{1,\theta},\ldots,X^n_{K,\theta}$ be
independent inputs satisfying the cost constraint with probability one, and let $Z^n_\theta$ be the output of $W^n_\theta$ due to $X^n_{[1:K],\theta}$.
The inputs and the output \emph{evaluated at the realized component} are denoted $X^n_{[1:K],F}=(X^n_{1,F},\ldots,X^n_{K,F})$ and $Z^n_F$; their joint
distribution with $F$ is
\begin{equation}\label{eq:csit-joint}
P_{FX^n_{[1:K],F}Z^n_F}(\theta,\bx_{[1:K]},\bz)\equiv w(\theta)\Bigl(\prod_{k=1}^KP_{X^n_{k,\theta}}(\bx_k)\Bigr)W^n_\theta(\bz|\bx_{[1:K]}).
\end{equation}
Thus, given $F=\theta$,
$(X^n_{[1:K],F},Z^n_F)$ has the law of $(X^n_{[1:K],\theta},Z^n_\theta)$. Let $P_{Z^n_F|X^n_{\cS^c,F}}$ and $P_{Z^n_F}$ be computed from \eqref{eq:csit-joint},
and, with $W^n$ the mixture \eqref{eq:mixed}, define for nonempty $\cS$
\[
\breve\iota^n_\cS(\bx_{[1:K]},\bz)\equiv\frac1n\log\frac{W^n(\bz|\bx_{[1:K]})}{P_{Z^n_F|X^n_{\cS^c,F}}(\bz|\bx_{\cS^c})},
\]
where for $\cS=[1:K]$ the denominator is $P_{Z^n_F}(\bz)$. The \emph{mixture densities} are the random variables
$\breve\iota^n_\cS(X^n_{[1:K],F},Z^n_F)$. Under \eqref{eq:csit-joint}, $X^n_{1,F},\ldots,X^n_{K,F}$ are independent only conditionally on $F$, and $W^n$ is in
general not the conditional law of $Z^n_F$ given $X^n_{[1:K],F}$. Thus $\breve\iota^n_\cS$ are decoding metrics rather than information densities in
the sense of \eqref{eq:densities}; the proof of Lemma~\ref{lem:csit} uses only that
$P_{FX^n_{\cS^c,F}Z^n_F}(\theta,\bx_{\cS^c},\bz)=w(\theta)P_{X^n_{\cS^c,\theta}Z^n_\theta}(\bx_{\cS^c},\bz)$ and that $W^n(\cdot|\bx_{[1:K]})$ is a probability density.
For a set ${\cA}\subseteq[1:K]$ of \emph{active users} (those with positive rate) let
\[
\cL_{\cA}\equiv\bigl\{\cS:\emptyset\ne\cS\subseteq\cA\bigr\}\cup\bigl\{[1:K]\bigr\}.
\]

\begin{lemma}[CSIT coding lemma]\label{lem:csit}
In the setup above, let ${\cA}\subseteq[1:K]$, let $M_1,\ldots,M_K$ be positive integers with $M_k=1$ for $k\notin{\cA}$, let
$\cD_n\subseteq\Theta$ be finite with $L_n\equiv|\cD_n|$, and let $\gamma>0$. Then there is a CSIT code with message sets of sizes $M_1,\ldots,M_K$,
satisfying the cost constraint, whose average error probability satisfies
\begin{equation}\label{eq:csit}
\varepsilon_n\le\Pr\Bigl\{\bigcup_{\cS\in\cL_{\cA}}\Bigl\{\breve\iota^n_\cS(X^n_{[1:K],F},Z^n_F)\le\frac1n\log(L_nM_\cS)+\gamma\Bigr\}\Bigr\}+(2^K-1+L_n)e^{-n\gamma}+w(\Theta\setminus\cD_n).
\end{equation}
\end{lemma}

The proof, a random-coding argument in which the decoder searches jointly over the message tuple and over a finite
list of candidate components, is given in Appendix~\ref{app:gen}. It extends \cite[Lemma 3]{Han98} (for $K$ users, \cite[Lemma~1]{HY-III}) to the case of CSIT.

For ${\cA}\subseteq[1:K]$, inputs $X^n_{[1:K],\theta}$ and thresholds $a=(a_\cS)_\cS$, let $\iota^n_{\cS,\theta}$ be the component densities of $W^n_\theta$ for
these inputs, and define
\begin{equation}\label{eq:ES}
E^{(\theta,{\cA})}_n(a\,|\,\bX_{[1:K],\theta})\equiv\bigcup_{\cS\in\cL_{\cA}}\{\iota^n_{\cS,\theta}\le a_\cS\}.
\end{equation}
This is the event formed with the component densities, with the events for the sets $\cS\not\subseteq{\cA}$ omitted, except that the event for
$[1:K]$ is always kept, because of the wrong-component error.

\begin{lemma}[Component form]\label{lem:compform}
In the setup of Lemma~\ref{lem:csit}, let $R\in\R^K_+$, ${\cA}\equiv\{k:R_k>0\}\ne\emptyset$, and $\gamma>0$. Then
there are CSIT codes satisfying the cost constraint, with $M_{k,n}=\lceil e^{nR_k}\rceil$ for $k\in{\cA}$ and $M_{k,n}=1$ for $k\notin{\cA}$, such
that
\begin{equation}\label{eq:compform}
\limsup_{n\to\infty}\varepsilon_n\le\limsup_{n\to\infty}\sum_{\theta\in\Theta}w(\theta)\Pr\bigl\{E^{(\theta,{\cA})}_n\bigl((R_\cS+4\gamma)_{\cS\in\cL_{\cA}}\bigm|\bX_{[1:K],\theta}\bigr)\bigr\}.
\end{equation}
\end{lemma}

The proof is given in Appendix~\ref{app:gen}.

\paragraph*{The case $K=2$}
We write $X^n_{1,\theta}=X^n_\theta$, $X^n_{2,\theta}=Y^n_\theta$ and $M_k=M_n^{(k)}$, and we identify the subsets $\{1\}$, $\{2\}$ and $\{1,2\}$ with the indices
$1$, $2$ and $3$ of \eqref{eq:densities}; then $\iota^n_{\cS,\theta}$ are the component densities \eqref{eq:compdens}. In Lemma~\ref{lem:csirt-conv},
$Q^n_{\theta,1}\equiv Q^n_{\theta,\{1\}}(\cdot|\by)$, $Q^n_{\theta,2}\equiv Q^n_{\theta,\{2\}}(\cdot|\bx)$ and $Q^n_{\theta,3}\equiv Q^n_{\theta,\{1,2\}}$; the union consists of
the three events formed with $\log M_n^{(1)}$, $\log M_n^{(2)}$ and $\log(M_n^{(1)}M_n^{(2)})$, and the last term is $3e^{-n\gamma}$. In Lemma~\ref{lem:csit},
\eqref{eq:csit-joint} reads $P_{FX^n_FY^n_FZ^n_F}(\theta,\bx,\by,\bz)=w(\theta)P_{X^n_\theta}(\bx)P_{Y^n_\theta}(\by)W^n_\theta(\bz|\bx,\by)$, and the second term of
\eqref{eq:csit} is $(3+L_n)e^{-n\gamma}$. For ${\cA}=\{1,2\}$, $\cL_{\cA}$ consists of all three subsets (indices $1$, $2$, $3$); for ${\cA}=\{1\}$, of $\{1\}$ and
$\{1,2\}$ (indices $1$ and $3$); for ${\cA}=\{2\}$, of $\{2\}$ and $\{1,2\}$ (indices $2$ and $3$). Accordingly,
$E^{(\theta,{\cA})}_n(a\,|\,\bX_\theta,\bY_\theta)=\bigcup_{k\in\cL_{\cA}}\{\iota^n_{k,\theta}\le a_k\}$, with $k$ running over these indices; for ${\cA}=\{1,2\}$ it is the
event $E^\theta_n(a\,|\,\bX_\theta,\bY_\theta)$ of \eqref{eq:Eneta}, and $E^{(\theta,{\cA})}_n\subseteq E^\theta_n$ for every ${\cA}$. In \eqref{eq:compform} the threshold for a set $\cS\in\cL_{\cA}$ is
$R_\cS+4\gamma$; e.g., for $\cA=\{1\}$ the event is $\{\iota^n_{1,\theta}\le R_1+4\gamma\}\cup\{\iota^n_{3,\theta}\le R_1+4\gamma\}$.

\subsection{Proof of \eqref{eq:Btilde-formula}}

\emph{(a) $C^{\rm rt}_{\Gamma}(\varepsilon\,|\,\bW)\subseteq\tilde B_{\Gamma}(\varepsilon)$.} Let $R\in\R^2_+$ be $\varepsilon$-achievable in the CSIRT scenario. Then there are CSIRT codes with
$\limsup_n\varepsilon_n\le\varepsilon$ and $\frac1n\log M_n^{(\alpha)}\ge R_\alpha-\gamma$ $(\alpha=1,2)$ for all $n\ge n_0$, where $\gamma>0$ is arbitrary. For each $\theta$, let $X^n_\theta$ and
$Y^n_\theta$ be independent and uniform on the codebooks for $\theta$; the resulting family $(\bX_{\Theta},\bY_{\Theta})$ belongs to $\cG^{\Theta}_{\Gamma}$. Apply Lemma~\ref{lem:csirt-conv}
with $K=2$ and $Q^n_{\theta,\cS}=P_{Z^n_\theta|X^n_{\cS^c,\theta}}$, i.e., $Q^n_{\theta,1}=P_{Z^n_\theta|Y^n_\theta}$, $Q^n_{\theta,2}=P_{Z^n_\theta|X^n_\theta}$ and $Q^n_{\theta,3}=P_{Z^n_\theta}$, so that the three densities in \eqref{eq:csirt-conv} are the component
densities $\iota^n_{k,\theta}$. By the monotonicity of the events, for $n\ge n_0$,
\[
\varepsilon_n\ge\int_{\Theta}\Pr\bigl\{E^\theta_n(R_1-2\gamma,R_2-2\gamma,R_1+R_2-4\gamma\,|\,\bX_\theta,\bY_\theta)\bigr\}\,dw(\theta)-3e^{-n\gamma}.
\]
Taking $\limsup_n$ gives $\varepsilon\ge\tilde J_{\bW}(R-2\gamma\,|\,\bX_{\Theta},\bY_{\Theta})$, that is, $R-2\gamma\in\tilde A_\varepsilon(\bX_{\Theta},\bY_{\Theta})$. Since $\gamma>0$ is arbitrary and $R\in\R^2_+$, we
have $R\in\tilde B_{\Gamma}(\varepsilon)$. This part does not use the countability of $\Theta$.

\emph{(b) $\tilde B_{\Gamma}(\varepsilon)\subseteq C^{\rm tx}_{\Gamma}(\varepsilon\,|\,\bW)$.} Since $C^{\rm tx}_{\Gamma}(\varepsilon\,|\,\bW)$ is closed by definition, it suffices to show the following:
if $R^*\in\R^2_+$, $\gamma>0$, and $R\in\tilde A_\varepsilon(\bX_{\Theta},\bY_{\Theta})$ satisfies $|R_\alpha-R^*_\alpha|\le\gamma$ $(\alpha=1,2)$, then the pair $r$ with $r_\alpha\equiv\max\{R^*_\alpha-8\gamma,0\}$
is $\varepsilon$-achievable in the CSIT scenario. (Every point of $\tilde B_{\Gamma}(\varepsilon)$ is such an $R^*$ for every $\gamma>0$, and $|r_\alpha-R^*_\alpha|\le8\gamma$.) If $r_1=r_2=0$,
the claim is trivial, since a code with $M_n^{(1)}=M_n^{(2)}=1$ has zero error probability. Otherwise let ${\cA}\equiv\{\alpha:r_\alpha>0\}\ne\emptyset$
and apply Lemma~\ref{lem:compform} with $r$ in place of $R$. Note that $R_\alpha\ge R^*_\alpha-\gamma\ge-\gamma$. For the sets in $\cL_{\cA}$, the thresholds in \eqref{eq:compform} satisfy:
\begin{itemize}
\item if $\alpha\in {\cA}$, then $r_\alpha+4\gamma=R^*_\alpha-4\gamma\le R_\alpha-3\gamma\le R_\alpha$;
\item if ${\cA}=\{1,2\}$, then $r_1+r_2+4\gamma=R^*_1+R^*_2-12\gamma\le R_1+R_2-10\gamma$;
\item if ${\cA}=\{2\}$, i.e., $r_1=0<r_2$, then $r_1+r_2+4\gamma=R^*_2-4\gamma\le R_2-3\gamma\le R_1+R_2-2\gamma$, since $R_1\ge-\gamma$; the case ${\cA}=\{1\}$ is
symmetric.
\end{itemize}
Since $E^{(\theta,{\cA})}_n$ grows with each threshold and $E^{(\theta,{\cA})}_n\subseteq E^\theta_n$, \eqref{eq:compform} yields
\[
\limsup_{n\to\infty}\varepsilon_n\le\tilde J_{\bW}(R_1,R_2,R_1+R_2\,|\,\bX_{\Theta},\bY_{\Theta})\le\varepsilon.
\]
Since $\liminf_n\frac1n\log M_n^{(\alpha)}\ge r_\alpha$, the pair $r$ is $\varepsilon$-achievable in the CSIT scenario.

\emph{(c) Conclusion.} By \eqref{eq:gen-incl}, $C^{\rm tx}_{\Gamma}\subseteq C^{\rm rt}_{\Gamma}$. Together with (a) and (b), this gives \eqref{eq:Btilde-formula}.

\emph{(d) $K$ users.} Theorem~\ref{thm:gencsi} extends to $K$ users, with $\bGamma$ in place of $\Gamma$, the $2^K-1$ component densities $\iota^n_{\cS,\theta}$ of
Sec.~\ref{sec:gen-lemmas}, the thresholds $R_\cS$ for $R\in\R^K$, and $\Cl_+$ taken in $\R^K$; only  the thresholds changes. The proof of
\eqref{eq:norx} is unchanged, with Remark~\ref{rem:kuser-general}, Lemma~\ref{lem:bridge} with $2^K-1$ denominators, and the thresholds $R_\cS-2|\cS|\gamma$ in Step 3.
Part (a) is Lemma~\ref{lem:csirt-conv} with $Q^n_{\theta,\cS}=P_{Z^n_\theta|X^n_{\cS^c,\theta}}$ and the thresholds $R_\cS-2|\cS|\gamma$.
In part (b), let $R^*$, $\gamma$ and $R$ be as there, with $|R_k-R^*_k|\le\gamma$ for every $k$, so that $R_k\ge-\gamma$; put $r_k\equiv\max\{R^*_k-(K+4)\gamma,0\}$
and, unless $r=0$ (which is trivial), ${\cA}\equiv\{k:r_k>0\}\ne\emptyset$, and apply Lemma~\ref{lem:compform} with $r$. Every $\cS\in\cL_{\cA}$ satisfies $r_\cS+4\gamma\le R_\cS$. Indeed, for a
nonempty $\cS\subseteq{\cA}$,
\[
r_\cS+4\gamma=R^*_\cS-(K+4)|\cS|\gamma+4\gamma\le R_\cS-(K+3)|\cS|\gamma+4\gamma\le R_\cS;
\]
and if ${\cA}\ne[1:K]$, then $r_{[1:K]}=r_{\cA}$ and $R_{[1:K]}=R_{\cA}+R_{{\cA}^c}\ge(R^*_{\cA}-|{\cA}|\gamma)-|{\cA}^c|\gamma=R^*_{\cA}-K\gamma$, so that
\[
r_{[1:K]}+4\gamma=R^*_{\cA}-(K+4)|{\cA}|\gamma+4\gamma\le R^*_{\cA}-K\gamma\le R_{[1:K]}.
\]
Since $E^{(\theta,{\cA})}_n$ grows with each threshold and is contained in the event with all $2^K-1$ densities, \eqref{eq:compform} gives
$\limsup_n\varepsilon_n\le\tilde J_{\bW}(R\,|\,\cdot)\le\varepsilon$, and, as $|r_k-R^*_k|\le(K+4)\gamma$, the rest is unchanged.\hfill$\Box$

\subsection{Formulas without Expansion with Respect to $w$}

Theorem~\ref{thm:general} applied to $\bW$ and $\bW_{\rm rx}$ expresses $C^{\rm no}_{\Gamma}$ and $C^{\rm rx}_{\Gamma}$ in terms of information densities of the mixture, without
expanding with respect to $w$. The following lemma gives formal counterparts for $C^{\rm tx}_{\Gamma}$ and $C^{\rm rt}_{\Gamma}$. Here the realized component $F$ must appear explicitly, because it is an input of the CSIT channel $\bW_{\rm tx}$ and both an input and an output of the CSIRT channel $\bW_{\rm rt}$. For a CSIT input $(\bX_{\Theta},\bY_{\Theta})$,
let $(F,X^n_{F},Y^n_{F},Z^n_{F})$ be distributed as in \eqref{eq:csit-joint} with $K=2$, write $W^n(\bz|\bx,\by,\theta)\equiv W^n_\theta(\bz|\bx,\by)$ and $W^n(\bz,\theta|\bx,\by,\theta)\equiv
W^n_\theta(\bz|\bx,\by)$ for the CSIT and CSIRT channels, and define
\begin{align}
J_{\bW_{\rm tx}}(R\,|\,\bX_{\Theta},\bY_{\Theta})&\equiv\limsup_{n\to\infty}\Pr\bigl\{E^{\rm tx}_n(R_1,R_2,R_1+R_2)\bigr\},\label{eq:Jtx}\\
J_{\bW_{\rm rt}}(R\,|\,\bX_{\Theta},\bY_{\Theta})&\equiv\limsup_{n\to\infty}\Pr\bigl\{E^{\rm rt}_n(R_1,R_2,R_1+R_2)\bigr\},\label{eq:Jrt}
\end{align}
where $E^{\rm tx}_n(a)$ is the event \eqref{eq:En} formed with the three information densities of $\bW_{\rm tx}$ with input $(X^n_F,Y^n_F,F)$ and output
$Z^n_F$, i.e., with numerator $W^n(Z^n_F|X^n_F,Y^n_F,F)$ and denominators $P_{Z^n_F|Y^n_FF}$, $P_{Z^n_F|X^n_FF}$ and $P_{Z^n_F|F}$, and $E^{\rm rt}_n(a)$
is formed likewise for $\bW_{\rm rt}$, with output $(Z^n_F,F)$, numerator $W^n(Z^n_F,F|X^n_F,Y^n_F,F)$ and denominators $P_{Z^n_FF|Y^n_FF}$,
$P_{Z^n_FF|X^n_FF}$ and $P_{Z^n_FF|F}$.

\begin{lemma}\label{lem:noexp}
For every CSIT input $(\bX_{\Theta},\bY_{\Theta})$,
\begin{equation}\label{eq:noexp}
J_{\bW_{\rm tx}}(R\,|\,\bX_{\Theta},\bY_{\Theta})=J_{\bW_{\rm rt}}(R\,|\,\bX_{\Theta},\bY_{\Theta})=\tilde J_{\bW}(R\,|\,\bX_{\Theta},\bY_{\Theta}).
\end{equation}
\end{lemma}

\begin{proof}
Given $F=\theta$, $(X^n_{F},Y^n_{F},Z^n_{F})$ has the law of $(X^n_\theta,Y^n_\theta,Z^n_\theta)$ by \eqref{eq:csit-joint}. For CSIT, $W^n(\bz|\bx,\by,\theta)=W^n_\theta(\bz|\bx,\by)$,
and the conditional laws $P_{Z^n_{F}|Y^n_{F}F}(\cdot|\cdot,\theta)$, $P_{Z^n_{F}|X^n_{F}F}(\cdot|\cdot,\theta)$ and $P_{Z^n_{F}|F}(\cdot|\theta)$ are $P_{Z^n_\theta|Y^n_\theta}$, $P_{Z^n_\theta|X^n_\theta}$
and $P_{Z^n_\theta}$; hence, given $F=\theta$, the three densities in \eqref{eq:Jtx} are the component densities $\iota^n_{k,\theta}$, and $J_{\bW_{\rm tx}}=\tilde J_{\bW}$.
For CSIRT, $W^n(\bz,\theta|\bx,\by,\theta)=W^n_\theta(\bz|\bx,\by)$ because $P_{F|X^n_{F}Y^n_{F}Z^n_{F}F}=1$, and, e.g.,
$P_{Z^n_{F}F|Y^n_{F}F}(\bz,\theta|\by,\theta)=\sum_{\bx}P_{X^n_{F}|F}(\bx|\theta)W^n_\theta(\bz|\bx,\by)=P_{Z^n_\theta|Y^n_\theta}(\bz|\by)$; thus the densities in \eqref{eq:Jrt} are again the
component densities, and $J_{\bW_{\rm rt}}=\tilde J_{\bW}$.
\end{proof}

As natural counterparts of the formulas for $C^{\rm no}_{\Gamma}$ and $C^{\rm rx}_{\Gamma}$ in Step 1 of the proof of \eqref{eq:norx}, Theorem~\ref{thm:gencsi} and Lemma~\ref{lem:noexp}
give the following.

\begin{corollary}\label{cor:noexp}
Let $\Theta$ be countable. For every $0\le\varepsilon<1$,
\begin{align}
C^{\rm tx}_{\Gamma}(\varepsilon\,|\,\bW)&=\Cl_+\Bigl(\bigcup_{(\bX_{\Theta},\bY_{\Theta})\in\cG^{\Theta}_{\Gamma}}\bigl\{R\in\R^2:J_{\bW_{\rm tx}}(R\,|\,\bX_{\Theta},\bY_{\Theta})\le\varepsilon\bigr\}\Bigr),\label{eq:tx-form}\\
C^{\rm rt}_{\Gamma}(\varepsilon\,|\,\bW)&=\Cl_+\Bigl(\bigcup_{(\bX_{\Theta},\bY_{\Theta})\in\cG^{\Theta}_{\Gamma}}\bigl\{R\in\R^2:J_{\bW_{\rm rt}}(R\,|\,\bX_{\Theta},\bY_{\Theta})\le\varepsilon\bigr\}\Bigr).\label{eq:rt-form}
\end{align}
\end{corollary}

In engineering terms, \eqref{eq:tx-form} and \eqref{eq:rt-form} say that the CSIT and CSIRT channels behave as if the realized
component were an additional input of the channel (and, for CSIRT, an additional output): the region is computed with the
information densities of the channel that has $F$ as a side input, exactly as \eqref{eq:general} is computed for an ordinary channel.
The only difference between the two scenarios, the decoder's knowledge of $F$, leaves the region unchanged for countable $\Theta$,
because the decoder of Lemma~\ref{lem:csit} can identify the component at no cost in rate.

\subsection{Contrast with Gaussian Components: An Example}\label{sec:contrast}

For memoryless Gaussian components, the two pairs of regions in Theorem~\ref{thm:gencsi} coincide with each other (Theorem~\ref{thm:csi}).
The following example shows that this is not the case in general, even for two memoryless components with finite alphabets.

\begin{example}[Two Z-channels]\label{ex:zchannel}
It suffices to show $C^{\rm no}_{\Gamma}(\varepsilon\,|\,\bW)<C^{\rm tx}_{\Gamma}(\varepsilon\,|\,\bW)$ for some $0\le\varepsilon<1$ for a single-user mixed memoryless channel
$\bW$ with finite alphabets and without cost constraint, which is a mixed MAC with $|\cY|=1$. Let $\Theta=\{1,2\}$ with arbitrary weights
$w(1),w(2)>0$, $\cX=\cZ=\{0,1\}$, and let $W_1$ and $W_2$ be Z-channels with $W_1(1\,|\,0)=\frac12$ and $W_2(z\,|\,x)=W_1(1-z\,|\,1-x)$, so that $W_2$ is
obtained from $W_1$ by interchanging the labels $0$ and $1$ at the input and at the output (see \cite[Figs.~2 and 3]{HY-II} for the
transition diagrams and the mutual-information curves). With $p_\theta=P_{X_\theta}(0)$ and the binary entropy
function $h(t)\equiv-t\log t-(1-t)\log(1-t)$, we have $I(X_1;Z_1)=h(\frac{p_1}2)-p_1\log2$ and $I(X_2;Z_2)=h(\frac{1-p_2}2)-(1-p_2)\log2$. Both are
strictly concave, an elementary computation gives the capacities $C_1=C_2=C\equiv\log\frac54$, attained at $p_1^*=\frac25$ and $p_2^*=\frac35$, and, with
$p=P_X(0)$, $I(X;Z_1)$ is decreasing on $[\frac25,1]$ and $I(X;Z_2)$ is increasing on $[0,\frac35]$, the two coinciding at $p=\frac12$. Hence
$C_0\equiv\max_X\min_\theta I(X;Z_\theta)=h(\frac14)-\frac12\log2=\frac34\log\frac43<C$. The two components have the same capacity, but different
capacity-achieving inputs. In bits, $C\simeq0.3219$ and $C_0\simeq0.3113$ (in nats, $C\simeq0.2231$ and $C_0\simeq0.2158$).

\emph{No CSI.} By the information spectrum formula of Yagi et al.~\cite[Theorem 1]{YHN16} for the $\varepsilon$-capacity of mixed memoryless
channels, it is not difficult to check that
\begin{align}
C^{\rm no}_{\Gamma}(\varepsilon\,|\,\bW)&=\sup_{X}\biggl\{R\st\int_{\{\theta\,|\,I(X;Z_\theta)\le R\}}dw(\theta)\le\varepsilon\biggr\}\label{eq:ce-no}\\
&=\begin{cases}\frac34\log\frac43&\mbox{for }0\le\varepsilon<\min_\theta w(\theta),\\[2pt] \log\frac54&\mbox{for }\min_\theta w(\theta)\le\varepsilon<1,\end{cases}\label{eq:ce-no-val}
\end{align}
where the supremum is over all input distributions on $\cX$: for $\varepsilon<\min_\theta w(\theta)$ no component may be in outage, while for
$\varepsilon\ge\min_\theta w(\theta)$ the outage of a component of smallest weight is affordable; and $I(X;Z_\theta)\le C$ for both $\theta$.

\emph{CSIT.} By Theorem~\ref{thm:gencsi}, $C^{\rm tx}_{\Gamma}(\varepsilon\,|\,\bW)=C^{\rm rt}_{\Gamma}(\varepsilon\,|\,\bW)=\tilde B_{\Gamma}(\varepsilon)$. For a single user, the information density
$\iota^n_{2,\theta}$ for the trivial second input vanishes identically, and, as in Remark~\ref{rem:closure}, the constraint on $R_2$ is removed by letting
$R_2\uparrow0$ before the closure is taken; since $\iota^n_{1,\theta}=\iota^n_{3,\theta}$, formula \eqref{eq:Btilde-formula} reduces to
$C^{\rm tx}_{\Gamma}(\varepsilon\,|\,\bW)=\sup_{\bX_{\Theta}}\{R\in\R\,|\,\tilde J_{\bW}(R\,|\,\bX_{\Theta})\le\varepsilon\}$, where $\tilde J_{\bW}(R\,|\,\bX_{\Theta})$ is \eqref{eq:Jtilde}
with the single density $\frac1n\log\frac{W^n_\theta(Z^n_\theta|X^n_\theta)}{P_{Z^n_\theta}(Z^n_\theta)}$. Let $X^n_\theta$ ($\theta=1,2$) be stationary memoryless with
$P_{X_\theta}(0)=p^*_\theta$, so that $I(X_1;Z_1)=I(X_2;Z_2)=C$. By the weak law of large numbers, $\tilde J_{\bW}(R\,|\,\bX_{\Theta})\le\sum_\theta w(\theta)\ind[I(X_\theta;Z_\theta)\le R]=0$
for $R<C$, and hence $C^{\rm tx}_{\Gamma}(\varepsilon\,|\,\bW)\ge\log\frac54$ for $0\le\varepsilon<1$. Since $\log\frac54>\frac34\log\frac43$, this gives an example of a mixed
channel for which $C^{\rm no}_{\Gamma}(\varepsilon\,|\,\bW)<C^{\rm tx}_{\Gamma}(\varepsilon\,|\,\bW)$ for $0\le\varepsilon<\min_\theta w(\theta)$.

In fact, equality holds in the last lower bound. Let $R$ be $\varepsilon$-achievable in the CSIRT scenario. For each $\theta$, the code used for $\theta$ is an
ordinary code for the discrete memoryless channel $W_\theta$, and by the strong converse for discrete memoryless channels in the
sense of Definition~\ref{def:sc} (Wolfowitz \cite{Wol57}; cf.\ Gallager \cite{Gal68} and Remark~\ref{rem:fongtan}), its error probability $\varepsilon_{\theta,n}$ tends to one
whenever $R>C_\theta$. By \eqref{eq:gen-avgerr} and Fatou's lemma, $\varepsilon\ge\sum_{\theta:\,C_\theta<R}w(\theta)$, and $R>C$ would give $\varepsilon\ge1$. Therefore
$C^{\rm tx}_{\Gamma}(\varepsilon\,|\,\bW)=C^{\rm rt}_{\Gamma}(\varepsilon\,|\,\bW)=\log\frac54$ for $0\le\varepsilon<1$, the $\varepsilon$-outage capacity $\sup\{R\,|\,w\{\theta\,|\,C_\theta\le R\}\le\varepsilon\}$. With CSI at the
transmitter the $\varepsilon$-capacity does not depend on $\varepsilon$, while without CSI it does, by \eqref{eq:ce-no-val}; the gap is caused solely by the mismatch
$p_1^*\ne p_2^*$ of the optimal inputs, not by a weaker component. The same example, viewed as a MAC, is used in Part II \cite[Proposition~3]{HY-II}.
\end{example}

\subsection{Uncountable Parameter Spaces}

Table~\ref{tab:countable} summarizes which parts of Theorem~\ref{thm:gencsi} survive when $\Theta$ is uncountable; the details follow.

\begin{table}[htbp]
\centering
\caption{The three equalities of Theorem~\ref{thm:gencsi} for countable and for uncountable parameter spaces.}
\label{tab:countable}
\renewcommand{\arraystretch}{1.25}
\begin{tabular}{lll}
\hline
Equality & countable $\Theta$ & uncountable $\Theta$ under (M1)--(M3)\\
\hline
$C^{\rm no}_{\Gamma}=C^{\rm rx}_{\Gamma}$ \eqref{eq:norx} & holds (Theorem~\ref{thm:gencsi}) & may fail (Remark~\ref{rem:countable})\\
$C^{\rm tx}_{\Gamma}=C^{\rm rt}_{\Gamma}$ \eqref{eq:Btilde-formula} & holds (Theorem~\ref{thm:gencsi}) & may fail (Example~\ref{ex:uncountable})\\
$C^{\rm rt}_{\Gamma}=\tilde B_{\Gamma}$ \eqref{eq:Btilde-formula} & holds (Theorem~\ref{thm:gencsi}) & holds (Proposition~\ref{prop:uncountable})\\
\hline
\end{tabular}
\end{table}

\begin{remark}[Countability cannot simply be dropped]\label{rem:countable}
Countability of $\Theta$ is used in the proof of \eqref{eq:norx} through $W^n\ge w(\theta)W^n_\theta$ with $w(\theta)>0$, and in that of \eqref{eq:Btilde-formula}
only through Lemma~\ref{lem:csit}, in part (b). It cannot simply be dropped in \eqref{eq:norx}. To see this, let $\cX=\cZ=\{0,1\}$ (a single user), let $F$ be uniformly distributed on $\Theta=[0,1)$, let
$b_i(\theta)$ denote the $i$th digit of the binary expansion of $\theta$, and let the component channels be deterministic:
$W^n_\theta(\bz|\bx)=\ind[z_i=x_i\oplus b_i(\theta)\ (i=1,\ldots,n)]$. Since $b_1(F),b_2(F),\ldots$ are i.i.d.\ and uniform on $\{0,1\}$, the mixed
channel is $W^n(\bz|\bx)=2^{-n}$ for all $(\bx,\bz)$, i.e., its output is independent of its input, so that $C^{\rm no}_{\Gamma}(\varepsilon\,|\,\bW)=0$ for
every $0\le\varepsilon<1$. With CSIR, however, the receiver can compute $b_i(\theta)$ and recover $\bx$ without error, so that $C^{\rm rx}_{\Gamma}(\varepsilon\,|\,\bW)=\log2$ (CSIT also gives
$\log2$). Whether \eqref{eq:norx} and \eqref{eq:Btilde-formula} hold for uncountable $\Theta$ when the components are suitably insensitive to $\theta$, as for the
quasi-static fading Gaussian MAC (cf.\ \eqref{eq:numer} and \eqref{eq:cell}), is left open; the measurability conditions (M1)--(M3) below are not of this kind
and guarantee only $C^{\rm rt}_{\Gamma}=\tilde B_{\Gamma}$ (Table~\ref{tab:countable}).
\end{remark}

\begin{remark}\label{rem:uncountable-q}
For uncountable $\Theta$, formula \eqref{eq:rt-form} for the CSIRT scenario holds under mild measurability conditions
(Proposition~\ref{prop:uncountable}): since the decoder knows $\theta$, it can decode with the component densities. Formula \eqref{eq:tx-form}
for the CSIT scenario fails in general (Example~\ref{ex:uncountable}, with memoryless but nonstationary components); for
finite-alphabet mixed memoryless MACs with stationary components, the question remains open \cite[Remark~9]{HY-II}.
\end{remark}

For an uncountable $\Theta$ we impose the following measurability conditions.
\begin{itemize}
\item[(M1)] $\Theta$, $\cX$, $\cY$ and $\cZ$ are standard Borel spaces, and $w$ is a probability measure on $\Theta$.
\item[(M2)] For each $n$, $\cZ^n$ carries a $\sigma$-finite reference measure $\mu_n$, and $W^n_\theta(\bz|\bx,\by)$ is a density with respect to
$\mu_n$ that is jointly measurable in $(\theta,\bx,\by,\bz)$; the cost functions $c_1$, $c_2$ are measurable, and the sets
$A_n^{(1)}\equiv\{\bx|\frac1nc_1(\bx)\le\Gamma_1\}$ and $A_n^{(2)}\equiv\{\by|\frac1nc_2(\by)\le\Gamma_2\}$ are nonempty.
\item[(M3)] The CSIT inputs $(\bX_{\Theta},\bY_{\Theta})$ in the union of \eqref{eq:Btilde} and \eqref{eq:rt-form} are given by probability kernels $\theta\mapsto P_{X^n_\theta}$
and $\theta\mapsto P_{Y^n_\theta}$ with $P_{X^n_\theta}(A_n^{(1)})=P_{Y^n_\theta}(A_n^{(2)})=1$ for $w$-almost all $\theta$; and a CSIRT code is admissible only
if its encoders $\theta\mapsto(\bu_{\theta,j},\bv_{\theta,k})$ are measurable and the sets $\{(\theta,\bz)|\bz\in D_{\theta,jk}\}$ are measurable, so that
$\theta\mapsto\varepsilon_{\theta,n}$ is measurable.
\end{itemize}
Under (M2) and (M3), $P_{Z^n_\theta|Y^n_\theta}(\bz|\by)=\int W^n_\theta(\bz|\bx,\by)\,P_{X^n_\theta}(d\bx)$, $P_{Z^n_\theta|X^n_\theta}$ and $P_{Z^n_\theta}$ are jointly measurable in $\theta$
and their arguments, so that the component information densities in \eqref{eq:Jtilde} and the integrand $\theta\mapsto\Pr\{\cdots\}$
there are measurable. Lemma~\ref{lem:csirt-conv} is stated for an arbitrary $w$, and the identities in the proof of Lemma~\ref{lem:noexp} hold,
for $w$-almost all $\theta$, with densities taken with respect to $\mu_n\otimes w$, so that $J_{\bW_{\rm rt}}=\tilde J_{\bW}$ also in this case.

\begin{proposition}\label{prop:uncountable}
Assume \textup{(M1)--(M3)}, with $(\Theta,w)$ otherwise arbitrary (not necessarily countable). Then, for
every $0\le\varepsilon<1$, $C^{\rm rt}_{\Gamma}(\varepsilon\,|\,\bW)=\tilde B_{\Gamma}(\varepsilon)$, and hence formula \eqref{eq:rt-form} holds.
\end{proposition}

The proof, which replaces the joint search over components in Lemma~\ref{lem:csit} by a measurable randomization and
Fubini's theorem, is given in Appendix~\ref{app:gen}.

\begin{remark}\label{rem:no-search}
Unlike Lemma~\ref{lem:csit}, Proposition~\ref{prop:uncountable} needs no finite search set and no bound $w(\theta)>0$, because the decoder knows $\theta$; countability
is needed only for the CSIT scenario.
\end{remark}

\begin{example}[Failure of \eqref{eq:tx-form} for uncountable $\Theta$]\label{ex:uncountable}
Consider a single user (i.e., $|\cY|=1$ and $R_2=0$) without
cost constraint, and let $\cX=\cZ=\{0,1\}$. Let $F$ be uniformly distributed on $\Theta=[0,1)$, and let $b_i(\theta)$ denote the $i$th
digit of the binary expansion of $\theta$ that does not end in infinitely many 1's, as in Remark~\ref{rem:countable}. The component
channels are
\[
W^n_\theta(\bz|\bx)=\prod_{i:\,b_i(\theta)=0}\ind[z_i=x_i]\prod_{i:\,b_i(\theta)=1}\frac12,
\]
that is, the $i$th letter is transmitted noiselessly if $b_i(\theta)=0$ and is replaced by an independent uniform bit if
$b_i(\theta)=1$. Conditions (M1)--(M3) are satisfied. Let $N_0(\theta)\equiv|\{i\le n| b_i(\theta)=0\}|$. Since $b_1(F),b_2(F),\ldots$ are
i.i.d.\ and uniform on $\{0,1\}$, the strong law of large numbers gives $N_0(\theta)/n\to\frac12$ for $w$-almost all $\theta$. We show that
\[
\mbox{RHS of \eqref{eq:tx-form}}=\mbox{RHS of \eqref{eq:rt-form}}=C^{\rm rt}_{\Gamma}(\varepsilon\,|\,\bW)=\Bigl[0,\frac12\log2\Bigr]\quad(0\le\varepsilon<1),
\]
whereas
\[
C^{\rm tx}_{\Gamma}(\varepsilon\,|\,\bW)\le\frac{\log2-h(\frac14)}{1-\varepsilon}<\frac12\log2\quad\mbox{for }0\le\varepsilon<\frac{2h(\frac14)}{\log2}-1\simeq0.6226,
\]
where $h$ is the binary entropy function as in Example~\ref{ex:zchannel}. Thus \eqref{eq:tx-form} and the first equality $C^{\rm tx}_{\Gamma}=C^{\rm rt}_{\Gamma}$ of \eqref{eq:Btilde-formula} fail, while \eqref{eq:rt-form}, i.e., the
second equality $C^{\rm rt}_{\Gamma}=\tilde B_{\Gamma}$ of \eqref{eq:Btilde-formula}, holds in accordance with Proposition~\ref{prop:uncountable}.

\emph{Proof (main steps).} (a) By Lemma~\ref{lem:noexp}, which holds under (M1)--(M3), both right-hand sides equal $\tilde B_{\Gamma}(\varepsilon)$. With
$\cC_n(\theta)\equiv\{i\le n| b_i(\theta)=0\}$, for any input the component information density equals $\frac1n\log\frac{1}{P_{X^n_{\theta,\cC}}(X^n_{\theta,\cC})}$ with probability
one, where $X^n_{\theta,\cC}\equiv(X_{\theta,i})_{i\in\cC_n(\theta)}$ takes at most $2^{N_0(\theta)}$ values; for the uniform input it equals $\frac{N_0(\theta)}{n}\log2$, and for any input it
exceeds $\frac{N_0(\theta)}{n}\log2+\gamma$ with probability at most $e^{-n\gamma}$. Since $N_0(\theta)/n\to\frac12$ $w$-a.s., dominated convergence and Fatou's lemma give
$\tilde B_{\Gamma}(\varepsilon)=[0,\frac12\log2]$, and Proposition~\ref{prop:uncountable} gives $C^{\rm rt}_{\Gamma}(\varepsilon\,|\,\bW)=[0,\frac12\log2]$.

(b) With the message $M$, $S_i\equiv b_i(F)$ $(i\le n)$ and $T\equiv(b_{n+1}(F),b_{n+2}(F),\ldots)$, the triple $(M,T,S^n)$ is mutually independent, $S^n$ is i.i.d.\
uniform, $X^n$ is a function of $(M,T,S^n)$, and $Z_i$ is the output of $V(z| x,0)=\ind[z=x]$, $V(z| x,1)=\frac12$ with state $S_i$. The CSIT scenario is thus
the channel $V$ with i.i.d.\ states known noncausally at the encoder only, used with the stochastic encoder $(M,T)\mapsto X^n$. Fano's inequality and the
converse of the Gelfand--Pinsker theorem \cite{GP80}, whose proof (with $U_i=(M,T,Z^{i-1},S^n_{i+1})$ and the Csisz\'ar sum identity) uses only that the
message is independent of $S^n$ and thus applies to $(M,T)$, give $R(1-\varepsilon)\le C_{\rm GP}\equiv\max\{I(U;Z)-I(U;S)\}$, the maximum being over all $P_{U|S}$ and $P_{X|US}$ with $S$ uniform on $\{0,1\}$ and
$Z$ generated by $V(\cdot|X,S)$. For $u$ with $P_U(u)>0$ and $a_u\equiv P_{S|U}(0|u)$, $P_{Z|U}(1|u)$ lies between $\frac{1-a_u}{2}$ and $\frac{1+a_u}{2}$, so
that $\mathrm{H}(Z|U=u)\ge h(\frac{1+a_u}{2})$. Since $\mathrm{H}(Z)\le\log2$, $I(U;S)=\log2-\E h(a_U)$ and $\E a_U=P_S(0)=\frac12$, we get, with
$\varphi(a)\equiv h(a)-h(\frac{1+a}{2})$, which is concave since $\varphi''(a)=-\frac{1}{a(1-a)(1+a)}<0$, $I(U;Z)-I(U;S)\le\E[\varphi(a_U)]\le\varphi(\frac12)=\log2-h(\frac14)$
for $U$ with an arbitrary alphabet, with equality for $U=X$ uniform and independent of $S$.

Without CSI the mixed channel is the binary symmetric channel with crossover probability $\frac14$, so that $C^{\rm no}_{\Gamma}(\varepsilon\,|\,\bW)=\log2-h(\frac14)$
for all $\varepsilon$ by its strong converse property; with CSIR it behaves as a binary erasure channel with erasure probability $\frac12$, and
$C^{\rm rx}_{\Gamma}(\varepsilon\,|\,\bW)=\frac12\log2$, the upper bound being $C^{\rm rt}_{\Gamma}(\varepsilon\,|\,\bW)=\frac12\log2$ of (a). Thus $C^{\rm no}_{\Gamma}(0\,|\,\bW)=C^{\rm tx}_{\Gamma}(0\,|\,\bW)<C^{\rm rx}_{\Gamma}(0\,|\,\bW)=C^{\rm rt}_{\Gamma}(0\,|\,\bW)$: knowledge of $\theta$ at the encoder alone is
of no use here, in contrast to Remark~\ref{rem:countable}. Both \eqref{eq:norx} and the first equality of \eqref{eq:Btilde-formula} fail, while the second holds,
although (M1)--(M3) are satisfied.
\end{example}

%% ====================================================================
\section{Concluding Remarks}\label{sec:conc}

We have presented a unified information spectrum treatment of the $\varepsilon$-capacity regions of mixed MACs with cost constraints, organized around
the general formula (Theorem~\ref{thm:general}) and the random-coding and converse lemmas from which it is derived. For quasi-static fading Gaussian
MACs, in all four CSI scenarios and with fixed-rate and with variable-rate coding, the converse parts for $0<\varepsilon<1$ rest on one more ingredient, the strong
converse property of the Gaussian MAC (Theorem~\ref{thm:sc}; Part III \cite{HY-III} gives an information spectrum proof), applied to each
component separately; for $\varepsilon=0$ a weak converse suffices (Remark~\ref{rem:why-sc}). These results recover, in the fixed-power setting, the outage and throughput characterizations underlying
\cite{BPS98}, \cite{CTB99}, \cite{EGL10}, \cite{OSW94} and \cite{SW97}.

For mixed MACs with general components, the same formula, together with a CSIRT converse lemma and a CSIT coding lemma, shows that, for a
countable parameter space, CSI at the receiver alone never enlarges the $\varepsilon$-capacity region and adds nothing to CSI at the transmitters
(Theorem~\ref{thm:gencsi}); for uncountable parameter spaces only $C^{\rm rt}_{\Gamma}=\tilde B_{\Gamma}$ survives in general (Proposition~\ref{prop:uncountable},
Remark~\ref{rem:countable} and Example~\ref{ex:uncountable}).

For finite-alphabet mixed memoryless MACs, Part II \cite{HY-II} derives from the same formula a single-letter inner bound, the single-letter $0$-capacity
region and a single-letter outer bound valid in all four CSI scenarios, which is the $\varepsilon$-capacity region with CSI at the transmitters for a
countable parameter space; without CSI the region for $0<\varepsilon<1$ remains open in general. The counterpart of Theorem~\ref{thm:qsf} fails in general there
because, whereas a single input is asymptotically optimal for all Gaussian components simultaneously, the optimal inputs of finite-alphabet
components generally differ (cf.\ Example~\ref{ex:zchannel}). For well-ordered finite-alphabet mixed memoryless MACs, all four regions coincide and equal the
outage region \cite[Theorem~8]{HY-II}; there, too, the strong converse property (\emph{not} the $\varepsilon$-strong converse) plays the crucial role.

A natural next step is to extend the framework to block-fading models, in which the state changes over several coherence blocks within a
codeword. The general formulas remain applicable, but the information spectrum calculus of Sec.~\ref{sec:intro-unified} would have to be carried out
afresh for a state process with nontrivial temporal correlation.

%% ====================================================================
\appendices
\section{Proof of Lemma~\ref{lem:basic}}\label{app:lemma1}

For each $t$, let
\begin{equation}
B_t\equiv\bigl\{u\st P_{U|T}(u|t)<e^{-a}Q(u|t)\bigr\}.
\end{equation}
Then
\begin{equation}
\Pr\{U\in B_t\,|\,T=t\}=\sum_{u\in B_t}P_{U|T}(u|t)\le e^{-a}\sum_{u\in B_t}Q(u|t)\le e^{-a}.
\end{equation}
Averaging over $t$ gives \eqref{eq:basic}, since the event in \eqref{eq:basic} is $\{U\in B_T\}$. For general alphabets, the sums are replaced
by integrals with respect to the reference measure. \hfill$\Box$

\section{Proofs for Section~\ref{sec:general-comp}}\label{app:gen}

\begin{proof}[Proof of Lemma~\ref{lem:csirt-conv}]
For each $\theta\in\Theta$, write $\bu_\theta(\bj)\equiv(\bu_{1,\theta}(j_1),\ldots,\bu_{K,\theta}(j_K))$ and $\bu_{\cS^c,\theta}(\bj_{\cS^c})\equiv(\bu_{k,\theta}(j_k))_{k\in\cS^c}$ for the
code used when $F=\theta$, and let $D_{\theta,\bj}\subseteq\cZ^n$ be its decoding regions (in the CSIRT scenario the decoder observes $(\bz,\theta)$, so that the
decoding regions may depend on $\theta$; for each fixed $\theta$, the sets $D_{\theta,\bj}$ are mutually disjoint). Let $\cM\equiv[1:M_1]\times\cdots\times[1:M_K]$ and
\begin{equation}\label{eq:epstheta}
\varepsilon_{\theta,n}=\frac{1}{M_{[1:K]}}\sum_{\bj\in\cM}W^n_\theta(D^c_{\theta,\bj}|\bu_\theta(\bj)),
\end{equation}
so that $\varepsilon_n=\int\varepsilon_{\theta,n}\,dw(\theta)$ by \eqref{eq:gen-avgerr}. Fix $\theta\in\Theta$ and $\gamma>0$, and define, for nonempty $\cS$ and $\bj\in\cM$,
\[
B^{(\cS)}_{\theta,\bj}=\Bigl\{\bz\in\cZ^n:\frac1n\log\frac{W^n_\theta(\bz|\bu_\theta(\bj))}{Q^n_{\theta,\cS}(\bz|\bu_{\cS^c,\theta}(\bj_{\cS^c}))}\le\frac1n\log M_\cS-\gamma\Bigr\},\qquad
B_{\theta,\bj}=\bigcup_{\emptyset\ne\cS\subseteq[1:K]}B^{(\cS)}_{\theta,\bj}.
\]
Since $D^c_{\theta,\bj}\cap B_{\theta,\bj}=B_{\theta,\bj}\setminus(D_{\theta,\bj}\cap B_{\theta,\bj})$, the union bound gives
\begin{equation}\label{eq:csirt-split}
\varepsilon_{\theta,n}\ge\frac{1}{M_{[1:K]}}\sum_{\bj}W^n_\theta(B_{\theta,\bj}|\bu_\theta(\bj))-\sum_{\emptyset\ne\cS\subseteq[1:K]}\frac{1}{M_{[1:K]}}\sum_{\bj}W^n_\theta\bigl(D_{\theta,\bj}\cap B^{(\cS)}_{\theta,\bj}\bigm|\bu_\theta(\bj)\bigr).
\end{equation}
Fix a nonempty $\cS$. By definition, $W^n_\theta(\bz|\bu_\theta(\bj))\le M_\cS e^{-n\gamma}Q^n_{\theta,\cS}(\bz|\bu_{\cS^c,\theta}(\bj_{\cS^c}))$ for $\bz\in B^{(\cS)}_{\theta,\bj}$. Writing
$\bj=(\bj_\cS,\bj_{\cS^c})$ and using $M_{[1:K]}=M_\cS M_{\cS^c}$,
\[
\frac{1}{M_{[1:K]}}\sum_{\bj}W^n_\theta\bigl(D_{\theta,\bj}\cap B^{(\cS)}_{\theta,\bj}\bigm|\bu_\theta(\bj)\bigr)\le\frac{e^{-n\gamma}}{M_{\cS^c}}\sum_{\bj_{\cS^c}}\sum_{\bj_\cS}Q^n_{\theta,\cS}\bigl(D_{\theta,(\bj_\cS,\bj_{\cS^c})}\bigm|\bu_{\cS^c,\theta}(\bj_{\cS^c})\bigr)\le e^{-n\gamma},
\]
where the last inequality holds because, for each fixed $\bj_{\cS^c}$, the decoding regions $D_{\theta,(\bj_\cS,\bj_{\cS^c})}$ with $\bj_\cS$ ranging over the
messages of the users in $\cS$ are mutually disjoint and $Q^n_{\theta,\cS}(\cdot|\bu_{\cS^c,\theta}(\bj_{\cS^c}))$ is a probability density (for $\cS=[1:K]$,
$M_{\cS^c}=1$ and the inner sum is $\sum_{\bj}Q^n_{\theta,[1:K]}(D_{\theta,\bj})\le1$). Plugging these $2^K-1$ bounds into \eqref{eq:csirt-split} yields
\[
\varepsilon_{\theta,n}\ge\frac{1}{M_{[1:K]}}\sum_{\bj}W^n_\theta(B_{\theta,\bj}|\bu_\theta(\bj))-(2^K-1)e^{-n\gamma}=\Pr\{Z^n_\theta\in B_{\theta,\bJ}\}-(2^K-1)e^{-n\gamma},
\]
where $\bJ=(J_1,\ldots,J_K)$ is the uniformly distributed message tuple, so that $\bu_\theta(\bJ)=X^n_{[1:K],\theta}$ and $Z^n_\theta\in B_{\theta,\bJ}$ is exactly the event in
\eqref{eq:csirt-conv}. Integrating over $\theta$ gives \eqref{eq:csirt-conv}. This is the proof for a single code over the mixed MAC without CSI (cf.\ \cite[Lemma 7]{YHN16}, \cite[Lemma 4]{Han98}, \cite[Lemma~2]{HY-III}), carried out
for each $\theta$ separately; the bound on the second term of \eqref{eq:csirt-split} uses only that each $Q^n_{\theta,\cS}(\cdot|\bx_{\cS^c})$ is a probability density.
\end{proof}

\begin{proof}[Proof of Lemma~\ref{lem:csit}]
Removing from the candidate list $\cD_n$ the components of zero weight decreases $L_n$, which only lowers the thresholds and the term
proportional to $L_n$, and leaves $w(\Theta\setminus\cD_n)$ unchanged; hence we may assume $w(\theta)>0$ for every $\theta\in\cD_n$. Then all the decoding
metrics below are well defined at every codeword of every listed component, because $P_{X^n_{\cS^c,F}}\ge w(\theta')P_{X^n_{\cS^c,\theta'}}$ by \eqref{eq:csit-joint}.
For $\bx_{[1:K]}\in\cX^n_1\times\cdots\times\cX^n_K$ we write $P_{X^n_{\cS,\theta}}\equiv\prod_{k\in\cS}P_{X^n_{k,\theta}}$ and $P_{X^n_\theta}\equiv P_{X^n_{[1:K],\theta}}$. By \eqref{eq:csit-joint},
$P_{X^n_{\cS^c,F}Z^n_F}=\sum_\theta w(\theta)P_{X^n_{\cS^c,\theta}Z^n_\theta}$; hence, for every $\theta$ with $w(\theta)>0$ and every nonempty $\cS$,
\begin{equation}\label{eq:csit-dom}
P_{X^n_{\cS^c,\theta}Z^n_\theta}\le\frac{P_{X^n_{\cS^c,F}Z^n_F}}{w(\theta)}.
\end{equation}

\emph{Random code.} For each $\theta$ and $k$, draw codewords $\bu_k(\theta,j)$, $1\le j\le M_k$, i.i.d.\ from $P_{X^n_{k,\theta}}$, all independent of each other
and of $F$; for a message tuple $\bj$ write $\bu(\theta,\bj)\equiv(\bu_1(\theta,j_1),\ldots,\bu_K(\theta,j_K))$. Given $F=\theta$, encoder $k$ sends $\bu_k(\theta,j_k)$; every
codeword satisfies the cost constraint. For $\cS\in\cL_{\cA}$ let
\[
\cT_\cS\equiv\Bigl\{(\bx_{[1:K]},\bz):\breve\iota^n_\cS(\bx_{[1:K]},\bz)>\frac1n\log(L_nM_\cS)+\gamma\Bigr\},\qquad\cT\equiv\bigcap_{\cS\in\cL_{\cA}}\cT_\cS.
\]
On observing $\bz$, the decoder outputs $\bj$ if there is a unique pair $(\theta,\bj)$ with $\theta\in\cD_n$ and $(\bu(\theta,\bj),\bz)\in\cT$, and declares an error otherwise.

\emph{Error analysis.} By symmetry we may assume that $\mathbf 1\equiv(1,\ldots,1)$ is sent. For definiteness, realize the codewords of this message tuple
as $\bu(\theta,\mathbf 1)\equiv X^n_{[1:K],\theta}$ for every $\theta$, independently across $\theta$, of $F$ and of the other codewords. Then the transmitted codeword tuple
is $X^n_{[1:K],F}$, the received sequence is $Z^n_F$, and $(F,X^n_{[1:K],F},Z^n_F)$ has the joint distribution \eqref{eq:csit-joint}. Let $G\equiv\{F\in\cD_n\}$. On $G$,
an error implies one of the following events: $\cE_0$, that $(X^n_{[1:K],F},Z^n_F)\notin\cT$; $\cE_\cS$, for a nonempty $\cS\subseteq{\cA}$, that
$(\bu(F,\bj'),Z^n_F)\in\cT$ for some $\bj'$ with $\{k:j'_k\ne1\}=\cS$; or $\cE'$, that $(\bu(\theta',\bj'),Z^n_F)\in\cT$ for some $\theta'\in\cD_n\setminus\{F\}$ and some $\bj'$.
Every $\bj'\ne\mathbf 1$ is covered by exactly one $\cE_\cS$, because $j'_k=1$ for $k\notin{\cA}$. The probability of $\cE_0$ is at most the first term of \eqref{eq:csit}.

\emph{(i) Wrong messages.} Fix a nonempty $\cS\subseteq{\cA}$; then $\cS\in\cL_{\cA}$ and $\cT\subseteq\cT_\cS$. Given $F=\theta$, the codewords $\bu_\cS(\theta,\bj')$ are
independent of $(\bu_{\cS^c}(\theta,\mathbf 1),Z^n_F)$, whose conditional law is $P_{X^n_{\cS^c,\theta}Z^n_\theta}$, and $\bu_{\cS^c}(\theta,\bj')=\bu_{\cS^c}(\theta,\mathbf 1)$. On $\cT_\cS$,
\[
P_{X^n_{\cS^c,F}Z^n_F}(\bx_{\cS^c},\bz)=P_{X^n_{\cS^c,F}}(\bx_{\cS^c})P_{Z^n_F|X^n_{\cS^c,F}}(\bz|\bx_{\cS^c})<P_{X^n_{\cS^c,F}}(\bx_{\cS^c})W^n(\bz|\bx_{[1:K]})\frac{e^{-n\gamma}}{L_nM_\cS}.
\]
Hence, by \eqref{eq:csit-dom}, and since there are fewer than $M_\cS$ tuples $\bj'$ with $\{k:j'_k\ne1\}=\cS$,
\begin{align*}
\Pr\{\cE_\cS\cap G\}&\le\sum_{\theta\in\cD_n}w(\theta)\sum_{\bj'}\sum_{(\bx_{[1:K]},\bz)\in\cT_\cS}P_{X^n_{\cS,\theta}}(\bx_\cS)P_{X^n_{\cS^c,\theta}Z^n_\theta}(\bx_{\cS^c},\bz)\\
&\le\sum_{\theta\in\cD_n}\sum_{\bj'}\sum_{\bx_{[1:K]},\bz}P_{X^n_{\cS,\theta}}(\bx_\cS)P_{X^n_{\cS^c,F}}(\bx_{\cS^c})W^n(\bz|\bx_{[1:K]})\frac{e^{-n\gamma}}{L_nM_\cS}\le e^{-n\gamma}.
\end{align*}

\emph{(ii) Wrong component.} Given $F=\theta$, the codewords of $\theta'\ne\theta$ are independent of $Z^n_F$, whose conditional law is $P_{Z^n_\theta}\le P_{Z^n_F}/w(\theta)$,
and $\cT\subseteq\cT_{[1:K]}$, on which $P_{Z^n_F}(\bz)<W^n(\bz|\bx_{[1:K]})e^{-n\gamma}/(L_nM_{[1:K]})$. Hence
\begin{align*}
\Pr\{\cE'\cap G\}&\le\sum_{\theta\in\cD_n}w(\theta)\sum_{\theta'\in\cD_n}\sum_{\bj'}\sum_{(\bx_{[1:K]},\bz)\in\cT}P_{X^n_{\theta'}}(\bx_{[1:K]})P_{Z^n_\theta}(\bz)\\
&\le\sum_{\theta\in\cD_n}\sum_{\theta'\in\cD_n}\sum_{\bj'}\sum_{\bx_{[1:K]},\bz}P_{X^n_{\theta'}}(\bx_{[1:K]})W^n(\bz|\bx_{[1:K]})\frac{e^{-n\gamma}}{L_nM_{[1:K]}}\le L_ne^{-n\gamma}.
\end{align*}
Adding $\Pr\{G^c\}=w(\Theta\setminus\cD_n)$ and the at most $2^K-1$ terms of (i), the average error probability $\bar\varepsilon_n$ over the random code satisfies
\begin{equation}\label{eq:ens-avg}
\bar\varepsilon_n\le\Pr\{\cE_0\}+(2^K-1+L_n)e^{-n\gamma}+w(\Theta\setminus\cD_n),
\end{equation}
which is the right-hand side of \eqref{eq:csit}, so some realization of the random code satisfies \eqref{eq:csit}.
\end{proof}

\begin{proof}[Proof of Lemma~\ref{lem:compform}]
Enumerate $\Theta$ and let $\cD_n$ consist of its first $\lfloor e^{\sqrt n}\rfloor$ elements (all of $\Theta$ if it is smaller). Then
$L_ne^{-n\gamma}\to0$ and $w(\Theta\setminus\cD_n)\to0$. The sizes $M_{k,n}$ of the statement satisfy $M_{k,n}\ge2$ exactly for $k\in{\cA}$. Since
$\lceil e^{nR}\rceil\le2e^{nR}$ and $\frac1n\log L_n\le n^{-1/2}$, we have $\frac1n\log(L_nM_\cS)\le R_\cS+\gamma$ for $\cS\in\cL_{\cA}$ and all large $n$. Apply
Lemma~\ref{lem:csit}. By \eqref{eq:csit-joint}, given $F=\theta$ the pair $(X^n_{[1:K],F},Z^n_F)$ has the law of $(X^n_{[1:K],\theta},Z^n_\theta)$; conditioning on
$F$, the first term of \eqref{eq:csit} is therefore at most
\[
\sum_\theta w(\theta)\Pr\Bigl\{\bigcup_{\cS\in\cL_{\cA}}\bigl\{\breve\iota^n_\cS(X^n_{[1:K],\theta},Z^n_\theta)\le R_\cS+2\gamma\bigr\}\Bigr\}.
\]
Let $\theta$ satisfy $w(\theta)\ge1/n$. Since $\Theta$ is countable, $W^n\ge w(\theta)W^n_\theta\ge n^{-1}W^n_\theta$. By Lemma~\ref{lem:basic} with $T=X^n_{\cS^c,\theta}$,
$U=Z^n_\theta$ and $Q=P_{Z^n_F|X^n_{\cS^c,F}}$ (a well-defined conditional density on the support of $X^n_{\cS^c,\theta}$, since $P_{X^n_{\cS^c,F}}\ge w(\theta)P_{X^n_{\cS^c,\theta}}$
by \eqref{eq:csit-joint}), and with $T$ trivial and $Q=P_{Z^n_F}$ for $\cS=[1:K]$, we have
$P_{Z^n_F|X^n_{\cS^c,F}}(Z^n_\theta|X^n_{\cS^c,\theta})\le e^{n\gamma}P_{Z^n_\theta|X^n_{\cS^c,\theta}}(Z^n_\theta|X^n_{\cS^c,\theta})$ with probability at least $1-e^{-n\gamma}$. Hence
$\breve\iota^n_\cS(X^n_{[1:K],\theta},Z^n_\theta)\ge\iota^n_{\cS,\theta}-\gamma-\frac{\log n}{n}$ outside an event of probability $e^{-n\gamma}$. For $n$ with $(\log n)/n\le\gamma$,
the $\theta$-th probability above is therefore at most
$\Pr\{E^{(\theta,{\cA})}_n((R_\cS+4\gamma)_{\cS\in\cL_{\cA}}\,|\,\bX_{[1:K],\theta})\}+2^Ke^{-n\gamma}$, since $|\cL_{\cA}|\le2^K$. The components with $w(\theta)<1/n$ contribute
at most $\sum_\theta w(\theta)\ind\{w(\theta)<1/n\}$, which tends to $0$ by dominated convergence. Taking $\limsup_n$ gives \eqref{eq:compform}.
\end{proof}

\begin{proof}[Proof of Proposition~\ref{prop:uncountable}]
The inclusion $C^{\rm rt}_{\Gamma}(\varepsilon\,|\,\bW)\subseteq\tilde B_{\Gamma}(\varepsilon)$ is part (a) of the proof of \eqref{eq:Btilde-formula}, which
uses only Lemma~\ref{lem:csirt-conv} and therefore holds for an arbitrary $w$ (the sums over $\theta$ being read as integrals; by (M3), the input
induced by an admissible code is a pair of probability kernels as required). We show the reverse inclusion.

\emph{Step 1 (Measurable randomization).} Fix $n$ and an input pair as in (M3). By the randomization lemma for
probability kernels into a standard Borel space \cite[Lemma 3.22]{Kal02}, there are measurable maps $\phi_n:\Theta\times[0,1]\to
\cX^n$ and $\psi_n:\Theta\times[0,1]\to\cY^n$ such that, for $\xi$ uniformly distributed on $[0,1]$, $\phi_n(\theta,\xi)$ has distribution $P_{X^n_\theta}$
and $\psi_n(\theta,\xi)$ has distribution $P_{Y^n_\theta}$ for every $\theta$. Fix $\bx_0\in A_n^{(1)}$ and $\by_0\in A_n^{(2)}$, and replace $\phi_n(\theta,\xi)$ by $\bx_0$
whenever $\phi_n(\theta,\xi)\notin A_n^{(1)}$, and similarly for $\psi_n$. The modified maps are still measurable, take values in $A_n^{(1)}$
and $A_n^{(2)}$, and, for $w$-almost all $\theta$, still produce the distributions $P_{X^n_\theta}$ and $P_{Y^n_\theta}$, because the modification
takes place on a $\xi$-null set for such $\theta$.

\emph{Step 2 (Random codes with $\theta$-dependent decoders).} Let ${\cA}\subseteq\{1,2\}$, let $M^{(1)},M^{(2)}\ge1$ be integers with
$M^{(\alpha)}=1$ for $\alpha\notin {\cA}$, write $M_1\equiv M^{(1)}$, $M_2\equiv M^{(2)}$ and $M_3\equiv M^{(1)}M^{(2)}$, and let $\gamma>0$. Let $\boldsymbol\xi=
(\xi_1,\ldots,\xi_{M_1},\xi'_1,\ldots,\xi'_{M_2})$ be i.i.d.\ uniform on $[0,1]$, and set $\bu(\theta,j)\equiv\phi_n(\theta,\xi_j)$ and $\bv(\theta,k)\equiv\psi_n(\theta,\xi'_k)$.
For each fixed $\theta$, these codewords are independent with distributions $P_{X^n_\theta}$ and $P_{Y^n_\theta}$ (for $w$-almost all $\theta$),
exactly as in the random code of the proof of Lemma~\ref{lem:csit}. Since the decoder knows $\theta$, it uses, for the component
$W^n_\theta$, the threshold decoder of that proof with the component information densities $\iota^n_{k,\theta}$ of \eqref{eq:compdens} and the
thresholds $\frac1n\log M_k+\gamma$ for $k\in \cL_{\cA}$, and with the search restricted to the single component $\theta$: it outputs $(j,k)$ if
$(j,k)$ is the unique pair with $\iota^n_{k',\theta}(\bu(\theta,j),\bv(\theta,k),\bz)>\frac1n\log M_{k'}+\gamma$ for all $k'\in \cL_{\cA}$, and declares an error otherwise.
This is the code of the proof of Lemma~\ref{lem:csit} for the one-point parameter space $\{\theta\}$ (so that $\cD_n=\{\theta\}$, $L_n=1$, $w(\{\theta\}^c)=0$,
and the mixture $W^n$ and the mixture densities reduce to $W^n_\theta$ and the component densities); hence, by \eqref{eq:ens-avg}, the
ensemble average $\bar\varepsilon_{\theta,n}\equiv\E_{\boldsymbol\xi}[\varepsilon_{\theta,n}(\boldsymbol\xi)]$ of the error probability over $W^n_\theta$ satisfies, for $w$-almost all $\theta$,
\[
\bar\varepsilon_{\theta,n}\le\Pr\Bigl\{\bigcup_{k\in \cL_{\cA}}\Bigl\{\iota^n_{k,\theta}(X^n_\theta,Y^n_\theta,Z^n_\theta)\le\frac1n\log M_k+\gamma\Bigr\}\Bigr\}+4e^{-n\gamma}.
\]

\emph{Step 3 (Derandomization by Fubini's theorem).} The map $(\theta,\boldsymbol\xi,\bz)\mapsto\ind[\bz\in D_{\theta,jk}(\boldsymbol\xi)]$ is measurable, because the
decoding rule is a finite Boolean combination of the events $\{\iota^n_{k',\theta}(\phi_n(\theta,\xi_j),\psi_n(\theta,\xi'_k),\bz)>\frac1n\log M_{k'}+\gamma\}$, each a
preimage of a measurable set under measurable maps. By Tonelli's theorem, the error probability $\varepsilon_{\theta,n}(\boldsymbol\xi)$, the average over $(j,k)$ of the integrals
$\int W^n_\theta(\bz|\phi_n(\theta,\xi_j),\psi_n(\theta,\xi'_k))\ind[\bz\notin D_{\theta,jk}(\boldsymbol\xi)]\,\mu_n(d\bz)$, is jointly measurable in $(\theta,\boldsymbol\xi)$,
and the expectation over $\boldsymbol\xi$ of $\int_{\Theta}\varepsilon_{\theta,n}(\boldsymbol\xi)\,dw(\theta)$ equals $\int_{\Theta}\bar\varepsilon_{\theta,n}\,dw(\theta)$. Hence there is $\boldsymbol\xi^*$ with $\int\varepsilon_{\theta,n}(\boldsymbol\xi^*)\,dw(\theta)\le
\int\bar\varepsilon_{\theta,n}\,dw(\theta)$; with $\boldsymbol\xi=\boldsymbol\xi^*$ the encoders and decoding regions are measurable in $\theta$ and every codeword satisfies the cost
constraint by Step 1. This gives admissible CSIRT codes with message sets of sizes $M^{(1)},M^{(2)}$ and
\begin{align*}
\varepsilon_n\le{}&\int_{\Theta}\Pr\Bigl\{\bigcup_{k\in \cL_{\cA}}\Bigl\{\iota^n_{k,\theta}(X^n_\theta,Y^n_\theta,Z^n_\theta)\le\frac1n\log M_k+\gamma\Bigr\}\Bigr\}\\
&\times dw(\theta)+4e^{-n\gamma}.
\end{align*}

\emph{Step 4 (Rates).} Let $(R^*_1,R^*_2)\in\tilde B_{\Gamma}(\varepsilon)$ and $\gamma>0$. By the definition of $\tilde B_{\Gamma}(\varepsilon)$, there are an input pair as
in (M3) and $(R_1,R_2)\in\tilde A_\varepsilon(\bX_{\Theta},\bY_{\Theta})$ with $|R_\alpha-R^*_\alpha|\le\gamma$; in particular $R_\alpha\ge-\gamma$. Put $r_\alpha\equiv\max\{R^*_\alpha-4\gamma,0\}$. If
$r_1=r_2=0$, then $(r_1,r_2)$ is trivially achievable. Otherwise let ${\cA}\equiv\{\alpha:r_\alpha>0\}$, $M_n^{(\alpha)}\equiv\lceil e^{nr_\alpha}\rceil$ for $\alpha\in {\cA}$
and $M_n^{(\alpha)}\equiv1$ for $\alpha\notin {\cA}$, and $a\equiv(R_1,R_2,R_1+R_2)$. For all sufficiently large $n$, $\frac1n\log M_k+\gamma\le r_k+2\gamma$
for $k\in \cL_{\cA}$, where $r_3\equiv r_1+r_2$. If $\alpha\in {\cA}$, then $r_\alpha+2\gamma=R^*_\alpha-2\gamma\le R_\alpha-\gamma\le a_\alpha$; if ${\cA}=\{1,2\}$, then
$r_1+r_2+2\gamma=R^*_1+R^*_2-6\gamma\le R_1+R_2-4\gamma\le a_3$; if ${\cA}=\{2\}$, then $r_1+r_2+2\gamma=R^*_2-2\gamma\le R_2-\gamma\le R_1+R_2=a_3$, since
$R_1\ge-\gamma$; the case ${\cA}=\{1\}$ is symmetric.
Hence each threshold is at most the corresponding $a_k$, and the bound of Step 3 gives, after adding the
conditions for $k\notin \cL_{\cA}$ (which only enlarges the events),
\[
\limsup_{n\to\infty}\varepsilon_n\le\tilde J_{\bW}(R_1,R_2,R_1+R_2\,|\,\bX_{\Theta},\bY_{\Theta})\le\varepsilon.
\]
Since $\liminf_n\frac1n\log M_n^{(\alpha)}\ge r_\alpha$, the pair $(r_1,r_2)$ is $\varepsilon$-achievable in the CSIRT scenario. As $|r_\alpha-R^*_\alpha|\le4\gamma$
and $\gamma>0$ is arbitrary, the closedness of $C^{\rm rt}_{\Gamma}(\varepsilon\,|\,\bW)$ gives $(R^*_1,R^*_2)\in C^{\rm rt}_{\Gamma}(\varepsilon\,|\,\bW)$.
\end{proof}

\section{Proof of Lemma~\ref{lem:comp-sc}}\label{app:comp-sc}

Fix $\theta=(h,\nu)$ with $(R_1,R_2)\in\R^2_+\setminus\Pi_{P_1,P_2,\theta}$, and consider codes with codewords $\bu_j,\bv_k$ and decoding regions $D_{jk}$.
Since $W^n_\theta(\bz|\bu,\bv)$ is the transition probability of the Gaussian memoryless MAC with noise power $N_\nu$ evaluated at
the inputs $(h_1\bu,h_2\bv)$, we distinguish the following cases.
\begin{enumerate}[leftmargin=1.8em,label=(\alph*)]
\item $h_1\ne0$ and $h_2\ne0$. The codebooks $\{h_1\bu_j\}$ and $\{h_2\bv_k\}$ together with the decoding regions $\{D_{jk}\}$ form
an $(n,M_n^{(1)},M_n^{(2)},\varepsilon_{\theta,n})$ code for the Gaussian memoryless MAC with noise power $N_\nu$ satisfying the
power constraint $(h_1^2P_1,h_2^2P_2)$. By Theorem~\ref{thm:wyner}, the $0$-capacity region of this MAC is $\Pi_{P_1,P_2,\theta}$, which does
not contain $(R_1,R_2)$. Hence Theorem~\ref{thm:sc} yields $\varepsilon_{\theta,n}\to1$.
\item $h_1=0$. In this case $\sfC_1(\theta)=0$ and $\sfC_{12}(\theta)=\sfC_2(\theta)$, so that $(R_1,R_2)\notin\Pi_{P_1,P_2,\theta}$ means $R_1>0$ or
$R_2>\sfC_2(\theta)$. If $R_1>0$, then $M_n^{(1)}\to\infty$; since the output distribution does not depend on the message
of user 1, we have $\sum_jW^n_\theta(D_{jk}|\bu_j,\bv_k)=W^n_\theta(\bigcup_jD_{jk}|\bu_1,\bv_k)\le1$ for each $k$, and hence $1-\varepsilon_{\theta,n}\le
1/M_n^{(1)}\to0$. If $R_1=0$ and $R_2>\sfC_2(\theta)$, then $R_2>0$; if moreover $h_2=0$, the same argument with the
roles of the users exchanged gives $\varepsilon_{\theta,n}\to1$. If $h_2\ne0$, let $D'_k\equiv\bigcup_jD_{jk}$. Since $D_{jk}\subseteq D'_k$,
\begin{equation}
\varepsilon_{\theta,n}\ge\frac{1}{M_n^{(2)}}\sum_{k=1}^{M_n^{(2)}}W^n_\theta\bigl((D'_k)^c|\bu_1,\bv_k\bigr),
\end{equation}
and the right-hand side is the error probability of the $(n,1,M_n^{(2)})$ code $(\{\mathbf 0\},\{h_2\bv_k\},\{D'_k\})$ for the
Gaussian memoryless MAC with noise power $N_\nu$ under the power constraint $(1,h_2^2P_2)$, say. Since $R_2>
\frac12\log(1+h_2^2P_2/N_\nu)$, the rate pair $(0,R_2)$ lies outside the $0$-capacity region of this MAC, and Theorem~\ref{thm:sc}
yields $\varepsilon_{\theta,n}\to1$.
\item $h_2=0$ and $h_1\ne0$. This case is symmetric to (b).
\end{enumerate}

\emph{$K$ users.} Let $\bR\notin\Pi_{\bP,\theta}$ of \eqref{eq:PiK}. If $R_k>0$ for some $k$ with $h_k=0$, the argument of case (b) shows that the
error probability over the component tends to one. Otherwise, there is a nonempty $\cS$ with $R_{\cS}>G_{\bP,\theta}(\cS)$ (see \eqref{eq:GPeta}),
and $\cS'\equiv\{k\in\cS\,|\,h_k\ne0\}$ satisfies $R_{\cS'}=R_{\cS}>G_{\bP,\theta}(\cS)=G_{\bP,\theta}(\cS')$. Merging the decoding regions over the messages of
the users with zero gains, as in case (b), yields a code for the Gaussian MAC of the users with nonzero gains whose error
probability is at most $\varepsilon_{\theta,n}$ and whose rate tuple lies outside its capacity region; the $K$-user strong converse
(Remark~\ref{rem:kuser-sc}) gives $\varepsilon_{\theta,n}\to1$.\hfill$\Box$

\section*{Acknowledgements}
\textit{Generative-AI use disclosure.} The authors used Anthropic's Claude (Sonnet 5 and Opus 5.5) during
the preparation of this manuscript. These systems assisted in the presentation and organization of several
Remarks including Concluding Remarks, in the reorganization of the manuscript into the present form
(including the arrangement of proofs and the redistribution of material between Parts I and II, by which the
theory of mixed MACs with general components in Sec.~\ref{sec:general-comp} was moved back to this paper), in the extension of Lemmas~\ref{lem:csirt-conv}--\ref{lem:compform} and their proofs in Appendix~\ref{app:gen} to $K$ users, and in
editing the English sentences. All mathematical results and proofs have been checked by the authors, who take
full responsibility for the content.

\end{document}